\RequirePackage[l2tabu, orthodox]{nag}	

\documentclass[
fontsize=11pt,					
a4paper,								
final,									
titlepage=true,					
bibliography=totoc,			
headsepline,						
parskip=half,						
oneside,								
BCOR=5mm								
]
{scrreprt}								

\usepackage[utf8]{inputenc}				      
\usepackage[T1]{fontenc}  				      
\usepackage[english]{babel}				      
\usepackage{lmodern} 					        
\usepackage{microtype}					        
\usepackage[onehalfspacing]{setspace}	  

\usepackage[mathscr]{eucal}             
\usepackage{dutchcal}

\DeclareFontFamily{OT1}{pzc}{}
\DeclareFontShape{OT1}{pzc}{m}{it}{<-> s * [1.15] pzcmi7t}{}
\DeclareSymbolFont{UPM}{U}{eur}{m}{n}
\DeclareMathAlphabet{\mathpzc}{OT1}{pzc}{m}{it}   
\DeclareMathSymbol{\partialup}{0}{UPM}{"40}       

\usepackage[
autooneside,							
]
{scrpage2}								
\automark{section}
\renewcommand*{\sectionmarkformat}{} 

\usepackage{titlesec}					
\usepackage[perpage, symbol*]{footmisc}	
\usepackage{epigraph}     

\usepackage[
final,                   
hidelinks,               
hypertexnames=false         
]           
{hyperref}                  

\usepackage[intlimits]{amsmath}			
\usepackage{amssymb}					
\usepackage{amsthm}						
\usepackage{thmtools}					
\usepackage{amsfonts}					
\usepackage{xfrac}						
\usepackage{mathtools}					
\mathtoolsset{mathic=true}				
\usepackage{upgreek}          

\usepackage{xcolor}                		
\usepackage{graphicx}             		
\usepackage{tikz} 						
\usetikzlibrary{matrix,arrows,chains,positioning,scopes}	
\usepackage{tikz-cd} 					
\definecolor{halfgray}{gray}{0.55}

\usepackage{picins}						

\usepackage{xparse}						
\usepackage[
	textsize=small,
	textwidth=0.8in,
	obeyDraft,							
	colorinlistoftodos
]{todonotes} 							

\usepackage{enumitem}					
\setlist[enumerate,1]{label=(\roman*), ref=(\roman*)}	
\setlist[enumerate,2]{label=\alph*), ref=(\theenumi.\alph*)} 
\setlist[enumerate,3]{label=\arabic*., ref=\theenumii.\roman*} 	
\usepackage[normalem]{ulem}				
\usepackage{tabto}						
\usepackage{seqsplit}               

\usepackage[
final,
notref,                 
notcite]                
{showkeys}                

\usepackage[backend=biber,
bibencoding=utf8,						
style=alphabetic,						
url=false,								  
eprint=true,							  
firstinits=true,            
maxbibnames=99              
]{biblatex}	                
\usepackage{csquotes} 			

\usepackage[
noabbrev]                   
{cleveref}                  
\AtBeginDocument{
    \RenewDocumentCommand \autoref {m} {\cref{#1}}
    \RenewDocumentCommand \Autoref {m} {\Cref{#1}}
    }
\usepackage{autonum}          

\newfont{\chapterNumber}{eurb10 scaled 5000}
\usepackage{titlesec}           
\titleformat{\chapter}[display]
  {\bfseries\Large}
  {\filleft \textcolor{halfgray}{\chapterNumber\thechapter}}
  {3ex}
  {\titlerule\vspace{2ex}\filright}
  [\vspace{2ex}\titlerule]

\DeclareFieldFormat{linked}{%
  \ifboolexpr{ test {\ifhyperref} and not test {\ifentrytype{online}} }
    {\iffieldundef{doi}
       {\iffieldundef{url}
          {\iffieldundef{isbn}
             {\iffieldundef{issn}
                {#1}
                {\href{\worldcatsearch\thefield{issn}}{#1}}}
             {\href{\worldcatsearch\thefield{isbn}}{#1}}}
          {\href{\thefield{url}}{#1}}}
       {\href{http://dx.doi.org/\thefield{doi}}{#1}}}
    {#1}}

\def\worldcatsearch{http://www.worldcat.org/search?qt=worldcat_org_all&q=}

\DeclareFieldFormat{title}{\printtext[linked]{\mkbibemph{#1}}}
\DeclareFieldFormat[article,inbook,incollection,inproceedings,patent,thesis,unpublished]{title}{\mkbibquote{\printtext[linked]{#1}\isdot}}
\DeclareFieldFormat[suppbook,suppcollection,suppperiodical]{title}{\printtext[linked]{#1}}

\AtEveryBibitem{\ifentrytype{article}{\clearfield{issn}}{}}			
\AtEveryBibitem{\clearlist{language}} 								          
\AtEveryBibitem{\clearfield{note}}    								          
\NewBibliographyString{bathesis}        
\DefineBibliographyStrings{english}{
  bathesis = {{Bachelor's thesis}{BA\addabbrvspace thesis}},    
}

\crefname{chapter}{Chapter}{Chapters}
\crefname{section}{Section}{Sections}
\crefname{subsection}{Subsection}{Subsections}
\crefname{subsubsection}{Subsubsection}{Subsubsections}
\crefname{figure}{Figure}{Figures}
\crefname{table}{Table}{Tables}
\crefname{appendix}{Appendix}{Appendices}

\makeatletter
\def\thmt@refnamewithcomma #1#2#3,#4,#5\@nil{%
  \@xa\def\csname\thmt@envname #1utorefname\endcsname{#3}%
  \ifcsname #2refname\endcsname
    \csname #2refname\expandafter\endcsname\expandafter{\thmt@envname}{#3}{#4}%
  \fi
}
\makeatother

\makeatletter
\def\endmathdisplay@a{%
  \if@eqnsw \gdef\df@tag{\tagform@\theequation}\fi
  \if@fleqn \@xp\endmathdisplay@fleqn
  \else 
    \ifx\df@tag\alt@tag \else \veqno \alt@tag \df@tag \fi 
    \ifx\df@label\@empty \else \@xp\ltx@label\@xp{\df@label}\fi
  \fi
  \ifnum\dspbrk@lvl>\m@ne
    \postdisplaypenalty -\@getpen\dspbrk@lvl
    \global\dspbrk@lvl\m@ne
  \fi
}

\def\equation@qed{%
    \iftagsleft@
      \hbox{\phantom{\quad\qedsymbol}}%
      \gdef\alt@tag{%
        \rlap{\hbox to\displaywidth{\hfil\qedsymbol}}%
        \global\let\alt@tag\@empty
      }%
    \else
      \gdef\alt@tag{%
        \global\let\alt@tag\@empty
        \vtop{\ialign{\hfil####\cr
                \print@eqnum\cr       
                \qedsymbol\cr}}%
        \setbox\z@
      }%
    \fi
  }
\makeatother

\makeatletter
\def\csxdefaux#1{%
    \protected@write\@mainaux{}{%
        \string\expandafter\string\xdef
          \string\csname\string\detokenize{#1}\string\endcsname{}%
    }%
}
\makeatother

\DefineFNsymbols*{mylamport}[math]{\dagger\ddagger\S\P\|{**}{\dagger\dagger}{\ddagger\ddagger}}
\setfnsymbol{mylamport}
\declaretheoremstyle[
        spacebelow=\topsep,
        spaceabove=\topsep,
        headfont=\normalfont\bfseries,
        bodyfont=\itshape,
        postheadspace=\newline,
        qed=$\diamondsuit$,
        headpunct=
]{myTheorem}

\declaretheoremstyle[
        spacebelow=\topsep,
        spaceabove=\topsep,
        headfont=\normalfont\bfseries,
        bodyfont=\normalfont,
        postheadspace=\newline,
        qed=$\diamondsuit$,
        headpunct=
]{myDefinition}

\declaretheorem[style=myTheorem, numberwithin=section, name=Theorem, refname=Theorem, Refname=Theorem]{theorem}
\declaretheorem[style=myTheorem, sibling=theorem, name=Proposition, refname=Proposition, Refname=Proposition]{proposition}
\declaretheorem[style=myDefinition, sibling=theorem, name=Lemma, refname=Lemma, Refname=Lemma]{lemma}
\declaretheorem[style=myDefinition, sibling=theorem, name=Corollary, refname=Corollary, Refname=Corollary]{corollary}
\declaretheorem[style=myDefinition, sibling=theorem, name=Definition, refname=Definition, Refname=Definition]{defn}
\declaretheorem[style=myDefinition, sibling=theorem, name=Remark, refname=Remark, Refname=Remark]{remark}
\declaretheorem[style=myDefinition, sibling=theorem, name=Remarks, refname=Remark, Refname=Remark]{remarksinternal}
\declaretheorem[style=myDefinition, sibling=theorem, name=Example, refname={Example, Examples}, Refname={Example, Examples}]{example}
\declaretheorem[style=myDefinition, sibling=theorem, name=Counterexample, refname=Counterexample, Refname=Counterexample]{counterexample}
\NewDocumentEnvironment{remarks}{s}{\begin{remarksinternal}\begin{thmenumerate}*}{\IfBooleanTF{#1}{}{\qedhere}\end{thmenumerate}\end{remarksinternal}}

\newcommand\EnumPrefix{}
\AtBeginEnvironment{proposition}{\crefalias{thmenumerateInterni}{proposition} \renewcommand{\EnumPrefix}{\theproposition}}
\AtBeginEnvironment{theorem}{\crefalias{thmenumerateInterni}{theorem} \renewcommand{\EnumPrefix}{\thetheorem}}
\AtBeginEnvironment{lemma}{\crefalias{thmenumerateInterni}{lemma} \renewcommand{\EnumPrefix}{\thelemma}}
\AtBeginEnvironment{corollary}{\crefalias{thmenumerateInterni}{corollary} \renewcommand{\EnumPrefix}{\thecorollary}}
\AtBeginEnvironment{remark}{\crefalias{thmenumerateInterni}{remark} \renewcommand{\EnumPrefix}{\theremark}}
\AtBeginEnvironment{remarksinternal}{\crefalias{thmenumerateInterni}{remarksinternal} \renewcommand{\EnumPrefix}{\theremarksinternal}}
\AtBeginEnvironment{example}{\crefalias{thmenumerateInterni}{example} \renewcommand{\EnumPrefix}{\theexample}}
\AtBeginEnvironment{counterexample}{\crefalias{thmenumerateInterni}{counterexample} \renewcommand{\EnumPrefix}{\thecounterexample}}

\numberwithin{equation}{section}        

\makeatletter 
\renewenvironment{proof}[1][\proofname] {\par\pushQED{\qed}\normalfont\topsep0\p@\@plus0\p@\relax\trivlist\item[\hskip\labelsep\bfseries#1\@addpunct{:} ]\ignorespaces ~\\}{\popQED\endtrivlist\@endpefalse} 
\makeatother

\newlist{thmenumerateIntern}{enumerate}{1}
\setlist[thmenumerateIntern,1]{label=(\roman*), ref=\EnumPrefix~(\roman*)} 
\NewDocumentEnvironment{thmenumerate}{s o}{\IfBooleanTF{#1}{~\vspace*{-\baselineskip}}{} \IfNoValueTF{#2}{\begin{thmenumerateIntern}}{\begin{thmenumerateIntern}[#2]}}{\end{thmenumerateIntern}}

\DeclareMathOperator{\sgn}{sgn}
\DeclareMathOperator{\pr}{pr}												
\DeclareMathOperator{\trace}{Tr}
\DeclareMathOperator{\divergence}{div}

\let\ker\relax                          
\DeclareMathOperator{\ker}{Ker}         
\DeclareMathOperator{\img}{Im}          
\DeclareMathOperator{\id}{id}
\DeclareMathOperator{\ev}{ev}
\DeclareMathOperator{\Ev}{Ev}
\DeclareMathOperator{\vol}{vol}         

\DeclareMathOperator{\inv}{inv}         
\DeclareMathOperator{\mult}{mult}       
\DeclareMathOperator{\comp}{comp}       
\DeclareMathOperator{\triv}{triv}       
\DeclareMathOperator{\AdjAction}{Ad}       
\DeclareMathOperator{\conjAction}{conj}       

\DeclareMathOperator{\difLie}{\mathpzc{Lie}}                              
\DeclareMathOperator{\flow}{FL}         

\DeclareMathOperator{\co}{co}                   
\newcommand{\cco}{\overline \co\,}                          

\newcommand{\liminv}{\varprojlim}

\DeclareMathOperator{\La}{L}    
\DeclareMathOperator{\Ra}{R}    

\NewDocumentCommand\dif{d()}{ \IfNoValueTF{#1}{\difIntern \!}{\difIntern(#1)} }         
\NewDocumentCommand\Dif{d()}{ \IfNoValueTF{#1}{\DifIntern \!}{\DifIntern(#1)} }         
\NewDocumentCommand\diF{d()}{ \IfNoValueTF{#1}{\diFIntern \!}{\diFIntern(#1)} }         
\NewDocumentCommand\dip{d()}{ \IfNoValueTF{#1}{\difpIntern \!}{\difpIntern(#1)} }       
\DeclareMathOperator{\difIntern}{d}             
\DeclareMathOperator{\DifIntern}{D}           
\DeclareMathOperator{\diFIntern}{\updelta} 
\DeclareMathOperator{\difpIntern}{\partialup}              

\newcommand{\diff}[2]{\frac{\dif #1}{\dif #2}}		                                
\newcommand{\difff}[3]{\frac{\dif^{#3} #1}{\dif #2^{#3}}}                               
\newcommand{\diffAt}[3]{\frac{\dif #1}{\dif #2}                                   
        \mathchoice
                {\bigg|_{\raisebox{1.8pt}{{$\scriptstyle #3$}}}}   
                {\Big|_{\raisebox{1.5pt}{{$\scriptscriptstyle #3$}}}}
                {UNREFENCEND-C}
                {UNREFENCEND-D}
        }

\newcommand{\diFF}[2]{\frac{\diF #1}{\diF #2}}                                          

\newcommand{\difp}[2]{\frac{\dip #1}{\dip #2}}		                                
\newcommand{\difpp}[3]{\frac{\dip^{#3} #1}{\dip #2^{#3}}}                            
\newcommand{\difpm}[3]{\frac{\dip^{2} #1}{\dip #2 \dip #3}}                          
\newcommand{\emphDef}[1]{\textbf{\boldmath #1}}			
\newcommand{\field}[1]{\mathbb{#1}}			
\newcommand{\R}{\field{R}}                                     
\newcommand{\C}{\field{C}}
\newcommand{\N}{\field{N}}
\newcommand{\I}{\mathrm{i}}
\DeclarePairedDelimiter\abs{\lvert}{\rvert}						
\DeclarePairedDelimiter\norm{\lVert}{\rVert}					

\makeatletter
\def\@tvsp{\mathchoice{{}\mkern-4.5mu}{{}\mkern-4.5mu}{{}\mkern-2.5mu}{}}
\def\ltrivert{\left|\@tvsp\left|\@tvsp\left|}
\def\rtrivert{\right|\@tvsp\right|\@tvsp\right|}
\makeatother

\DeclareMathSymbol{:}{\mathpunct}{operators}{"3A}
\let\colon\relax
\DeclareMathSymbol{\colon}{\mathrel}{operators}{"3A}

\newcommand\normdot{\lVert \,\cdot\, \rVert}                  
\DeclarePairedDelimiterX\setc[2]{\{}{\}}{\, #1 \colon \, #2 \,}      
\DeclarePairedDelimiterX\set[1]{\{}{\}}{\, #1 \,}      

\newcommand{\defeq}{\vcentcolon=}								
\newcommand{\contr}{\mathbin{\rule[.1ex]{.4em}{.04em}\rule[.1ex]{.04em}{.4em}}\,}       
\newcommand{\vectorf}[1]{\mathfrak{X}(#1)}

\NewDocumentCommand\diffform{m g g}{ \Omega^{#1} \IfNoValueTF{#3}{\IfNoValueTF{#2}{}{(#2)}}{(#2, #3)} }
\NewDocumentCommand\diffformBi{m m g g}{ \Omega^{#1, #2} \IfNoValueTF{#4}{\IfNoValueTF{#3}{}{(#3)}}{(#3, #4)} }
\newcommand{\liea}[1]{\mathfrak{#1}}
\newcommand{\lieaE}[1]{\mathop{Lie}(#1)}

\NewDocumentCommand \DiffGr{}{\mathcal{Diff}}    

\newcommand{\secspaceEx}[1]{\Gamma^\infty(#1)}    
\newcommand{\secspace}[1]{\mathcal{#1}}    
\newcommand{\secmap}[1]{\mathcal{#1}}          

\newcommand{\Gau}{\secspace{Gau}}

\newcommand{\Conn}{\secspace{Conn}}
\newcommand{\transf}{\secmap{transf}}

\DeclarePairedDelimiterX\scalarprod[2]{\langle}{\rangle}{#1,#2}	
\DeclarePairedDelimiterX\scalarprodr[2]{(}{)}{#1,#2}			
\DeclarePairedDelimiterX\poisson[2]{\{}{\}}{#1,#2}				
\DeclarePairedDelimiterX\commutator[2]{[}{]}{#1,#2}				
\DeclarePairedDelimiterX\dualpair[2]{\langle}{\rangle}{#1,#2}	
\DeclarePairedDelimiterX\pcommutator[2]{\{}{\}}{#1,#2}                              
\DeclarePairedDelimiterX\wedgeLie[2]{[}{]}{#1 \wedge #2}                             
\DeclarePairedDelimiter\equivClass{[}{]}                               

\newcommand\bcdot{\ensuremath{%
  \mathchoice%
   {\mskip\thinmuskip\lower0.1ex\hbox{\scalebox{0.8}{$\bullet$}}\mskip\thinmuskip}%
   {\mskip\thinmuskip\lower0.1ex\hbox{\scalebox{0.8}{$\bullet$}}\mskip\thinmuskip}%
   {\lower-0.1ex\hbox{\scalebox{0.6}{$\bullet$}}}%
   {\lower0.2ex\hbox{\scalebox{0.6}{$\bullet$}}}%
}}

\NewDocumentCommand\AtiyahBundle{s m}
{
        \opbundle{#1}{\AtiyahBundleOp}{#2}
}
\NewDocumentCommand\AtiyahBundleInt{s m o}
{
        \AtiyahBundleOp \IfNoValueTF{#3}{ #1 }{ _{#2} #1 }       
}
\DeclareMathOperator\AtiyahBundleOp{At}
\def\AtiyahBundleAssignSubArg_#1{
}

\NewDocumentCommand\AdjBundle{m} {\AdjBundleOp #1}
\DeclareMathOperator\AdjBundleOp{Ad}

\makeatletter
\NewDocumentCommand\opbundle{m m m}{ 
        \def\opbundle@{ #2 }
        \def\opbundle@@{ \IfBooleanTF #1 {(#3)}{#3}}
        \@ifnextchar_{\@opbundleSub}{\opbundle@ \opbundle@@}}
\def\@opbundleSub_#1{\opbundle@_{#1} \opbundle@@}

\newcommand{\verteq}{\rotatebox{90}{$\,=$}}

\title{Slice theorem for Fréchet group actions \\[0.4cm] and covariant symplectic field theory}
\author{Tobias Diez}
\date{October 2013}
\newcommand{\thethesistype}{Master’s Thesis}
\newcommand{\theassessor}{Prof. Dr. G. Rudolph \\ Prof. Dr. R. Verch}
\newcommand{\theuniversity}{Universität Leipzig}
\newcommand{\thedepartment}{Fakultät für Physik und Geowissenschaften}
\newcommand{\theplace}{Leipzig}

\makeatletter
\let\thetitle\@title
\let\theauthor\@author
\let\thedate\@date
\makeatother

\begin{document}
\pagenumbering{roman}

\begin{spacing}{1} 		
	\begin{titlepage}
\newcommand{\HRule}{\rule{\linewidth}{0.5mm}}

\centering

\textsc{ {\LARGE \theuniversity{}} \\[0.5cm] {\Large \thedepartment} }\\[3.5cm] 
\textsc{\Large \thethesistype{}}\\[0.5cm]

\HRule \\[0.4cm] 
{\huge \bfseries \thetitle{}}\\[0.4cm] 
\HRule \\[2.5cm]
 
{\large \emph{Presented by}\\[0.3cm] {\LARGE \bfseries \theauthor{}} }
\\[1.5cm]
{\large \emph{Assessors:} \\[0.2cm]
	\renewcommand{\arraystretch}{1.5}
	\begin{tabular}{l} 
  		\theassessor{}
	\end{tabular}} 
 
\vfill

{\large \theplace{}, \thedate{}} \vspace{-2cm}

\end{titlepage}
\end{spacing}

\thispagestyle{empty}
\mbox{} 
\clearpage

\addtocontents{toc}{\protect\thispagestyle{empty}}	
\thispagestyle{empty}
\addtocontents{toc}{\vspace{-\baselineskip}} 
\tableofcontents
\thispagestyle{empty}
\clearpage

\setcounter{page}{-10}
\listoftodos
\clearpage

\pagenumbering{arabic}

\chapter{Introduction}
\sectionmark{Introduction}

Symmetries play a decisive role throughout all areas of physics as they can be exploited to reduce the number of necessary variables and thus simplify the problem at hand. On the mathematical side, this `grand theme' of classical mechanics is implemented by the rich interaction of Lie group theory with symplectic geometry. This fruitful symbiosis finally culminates in the modern formulation of Noether's theorem and the powerful technique of symplectic reduction.

The aim of this thesis is to carry over these concepts and methods to classical field theory in a mathematically sound framework. In this context, configuration spaces as well as symmetry groups are infinite-dimensional and hence one has to cope with functional analytic caveats. Usually one finds refuge in the well-documented recipe of Sobolev techniques. The spaces under consideration are completed with respect to Sobolev norms and thereby yield Banach spaces, where the classical theory is then applied to construct a weak solution. Given enough regularity one can infer on the smooth case. Not only is the physical significance of weak solutions hardly verifiable but also the Sobolev method presents its own particular challenges. For example, the completion of the diffeomorphism group is no Lie group since the left composition fails to be smooth due to a `loss of derivatives'. In conclusion, the finite-dimensional theory generalises relatively straightforward to the Banach setting but often is not directly applicable to concrete problems.

The situation is reversed when smoothness is incorporated right from the beginning. Then the group of diffeomorphism is a genuine Fréchet Lie group, but the general theory faces serious difficulties: The inverse function theorem is not applicable and, closely related, differential equations may have no or infinitely many solutions. Nevertheless, the theorem of Nash and Moser provides remedy for a large subcategory of Fréchet spaces. Based on this celebrated result, the present thesis investigates how the ordinary theory of Lie groups and their actions transfer to the Fréchet setting. The main, original contribution to the long-term aim of a symplectic reduction theory for Fréchet Lie group actions consists in the following general slice theorem: 
\begin{theorem}[Slice theorem]
	Let $G$ be a tame Fréchet Lie group and $M$ a tame Fréchet manifold. A tame smooth proper $G$-action $\Upsilon$ has a slice at some point $m \in M$ if the following conditions are fulfilled:
	\begin{enumerate}
		\item The stabilizer $G_m$ is a tame Fréchet Lie subgroup of $G$ such that the natural projection $G \to G / G_m$ is a principal $G_m$- bundle.
		\item The image of $(\Upsilon_{m})'_e$ is topologically complemented in $T_mM$.
		\item Locally, for $g \in G$ near $e$, the derivative $(\Upsilon_{m})'_g$ is invertible.
		\item $M$ carries a $G$-invariant Riemannian metric $g$ such that the exponential map exists and the restriction to the normal bundle, $\exp: NO \to M$, is a local diffeomorphism at every point of the zero section.\qedhere
	\end{enumerate}
\end{theorem}
Please note that the theorem in the stated form is not entirely correct since the assumptions partly lack the required mathematical precision, see \cref{prop::liegroup:sliceTheorem} for a rigorous formulation. Nevertheless, the presented theorem visualizes the conditions which are additionally necessary to transfer the finite-dimensional result of Palais to the Fréchet setting. 

The thesis is organised as follows. The first two chapters lay the foundations of this work by presenting the differential geometric calculus for locally convex spaces and manifolds. It is here where the reader finds the discussion of the Nash-Moser theorem. Subsequently, the exposition works towards the above mentioned slice theorem and finally culminates in its proof. To illustrate how the general statement is applied to concrete problems, the case of gauge theory is exemplarily investigated in \cref{sec::gaugeTheory}. Thereby, the existence of slices for the action of the gauge transformation group on the space of connections is established.

Interposed in this Lie group-focused material, the second pillar of symplectic reduction is studied, namely symplectic geometry and the notion of the momentum map. Usually, the symplectic formulation of field theory is approached using the instantaneous formalism, where covariance is manifestly broken by decomposing spacetime $M = T \times \Sigma$ in temporal and spatial directions. Instead, \cref{sec:cft} presents a novel covariant ansatz, which is based on the following fundamental concept: the field degrees of freedom~$\secspace{F}$ and spacetime $M$ are incorporated into the space $\secspace{F} \times M$ and all objects under considerations live directly on this product. In particular, the introduced formalism places all spacetime symmetries on an equal footing regardless whether they are spatial or not. 

\chapter{Locally convex spaces}\label{cha::locallyconvexspaces:locally_convex_spaces}

This chapter develops the fundamental concepts needed for the study of differential geometric aspects in infinite dimensions. First, the notion of a locally convex space is introduced and its topological properties are investigated. The emphasis here lies on similarities and differences to the well-known category of Banach spaces. Then the focus shifts to differential and integral calculus, where one realizes that many known facts can be transferred to this new setting with only slight adaptations. Finally, the failure of the inverse function theorem and its replacement by the seminal theorem of Nash and Moser in the category of tame Fréchet spaces is discussed.

\section{Topological preliminaries}\label{section::locallyconvexspaces:topological_preliminaries}
This section is not meant to aggregate all known facts about locally convex spaces but instead is supposed to give a rather broad overview and to lay a conceptual and notational basis. Details and further information are covered in \parencite{MeiseVogt1992, Jarchow1981, Rudin1973}.

The core of functional analysis is formed by the study of function spaces and the maps between them. A reasonable topology on $C^\infty(M)$ should reflect the fact that two functions are `near' to each other if not only their image is comparable but also the higher derivatives show a similar behaviour. Thus $C^\infty(M)$ is topologized by infinitely many norms and not by a single one. After recognizing that these norms are actually not positive definite, one is led to the following generalization.

\begin{defn} \label{defn::locallyconvexspaces:seminorm}
Let $X$ be a vector space over $\R$. A non-negative function $\normdot: X \to \R$ is called a \emphDef{seminorm} if the following conditions hold:
\begin{enumerate}
\NumTabs{10}
\item absolute homogeneity: \tab $\norm{ \lambda x} = \abs{\lambda} \, \norm{x},$ \tab for $\lambda \in \R, x \in X$
\item triangle inequality: 	\tab $\norm{ x + y} \leq \norm{x} + \norm{y},$ \tab for $x, y \in X$
\end{enumerate}
In the following, all Greek indices are understood to be elements of some fixed, but not necessarily countable `index set'. A family of seminorms $\normdot_\alpha$ is a \emphDef{directed fundamental system} if every two seminorms can be estimated by a third one, that is, for all $\alpha$ and $\beta$ there exist $\gamma$ and some constant $C \in \R$ such that $\normdot_\alpha + \normdot_\beta \leq C \normdot_\gamma$. Moreover, an \emphDef{increasing fundamental system} is a countable family fulfilling $\normdot_n \leq \normdot_{n+1}$.
\end{defn}   
Note that, for a seminorm, $\norm{x} = 0$ does not necessarily imply $x = 0$, but when this condition holds a seminorm is a norm. A directed fundamental system $\normdot_\alpha$ of seminorms defines a topology on $X$ by declaring the tubes\footnote{As $\normdot_\alpha$ is not necessarily positive definite, the usual $\varepsilon$-ball is merely a tube extending to infinity in the direction of the kernel.} 
\begin{equation}
U_{\alpha, \varepsilon} \defeq \setc{x \in X}{ \norm{x}_\alpha < \varepsilon}
\end{equation}
to form a $0$-neighbourhood base. Equivalently, $X$ carries the natural initial topology with respect to the family of seminorms. 
\begin{example}[Smooth functions] \label{ex::locallyConvexSpaces:SmoothFunctionOnOpenSubset}
	Let $U \subseteq \R$ be an open subset. On the space $C^\infty(U)$ of smooth real-valued functions a convenient topology is induced by seminorms
	\begin{equation}
		\norm{\phi}_{K, k} \defeq \sup_{\substack{x \in K \\ i \leq k}}{\abs*{\difff{\phi}{x}{i}(x)}} 	\qquad \text{for } \phi \in C^\infty(U), k \in \N, K \subseteq U \text{ compact.}
	\end{equation}  
	This example will be extended in \autoref{sec::sectionSpace} to sections of fibre bundles.
\end{example} 
A real\footnote{In the following only vector spaces over the real numbers are considered, though almost all results extend also to the complex case.} vector space together with a Hausdorff topology generated by seminorms is a topological vector space and is called a \emphDef{locally convex space}. The name arises from the convexity property of $U_{\alpha, \varepsilon}$, which is a direct consequence of the triangle inequality. 

The following proposition collects additional properties of $U_{\alpha, \varepsilon}$ and gives a more geometrically characterization of locally convex spaces by their $0$-neighbourhoods.
\begin{proposition} \label{prop::lcs:equivalenceSeminormsAndConvexBalencedAbsorbingSets}
	Let $(X, \normdot_\alpha)$ be a locally convex space. There exists a $0$-neighbourhood base $\mathcal{U}$ such that its members $U \in \mathcal{U}$  possess the following properties:  
	\begin{thmenumerate}
		\item Convex: The straight line connecting every two points in $U$ lies completely in $U$. 
		\item Absorbing: For every $x \in X$ there exists $\lambda > 0$ with $\lambda \cdot x \in U$.
		\item Balanced: $\eta\cdot U \subseteq U$ for every $\abs{\eta} \leq 1$.  
	\end{thmenumerate}
	Conversely, for every subset $V \subseteq X$, $0 \in V$, with these properties the Minkowski functional
	\begin{equation}
		\norm{x}_V = \inf\setc{\kappa \in \R}{x \in \kappa \cdot V}
	\end{equation}
	is a seminorm. Thus a topological vector space with a $0$-neighbourhood base $\mathcal{V}$ consisting of convex, absorbing and balanced subsets $V$ is a locally convex space $(X, \normdot_V)$. The seminorms $\normdot_V$ induce a topology, which is equivalent to the original one.   
\end{proposition}
\begin{proof}
	The above constructed $\varepsilon$-tubes $U_{\alpha, \varepsilon}$ are clearly absorbing and balanced due to the absolute homogeneity of seminorms. Convexity was already noted. Conversely, for an absorbing $V$ containing $0$ the Minkowski functional is well-defined and finite. As $\normdot_V$ is homogeneous for positive scalars it is enough to show $\norm{- x}_V = \norm{x}_V$. But this obviously holds because $V$ is balanced. In order to prove the triangle inequality fix $x, y \in X$ and let $\varepsilon  > 0$ be arbitrary. There exists $\lambda, \eta > 0$ with $\sfrac{x}{\lambda}, \sfrac{y}{\eta} \in V$ and $\lambda < \norm{x}_V + \varepsilon$, $\eta < \norm{y}_V + \varepsilon$. Since $V$ is convex
	\begin{equation}
		\frac{x+y}{\lambda + \eta} = \frac{\lambda}{\lambda + \eta}\, \frac{x}{\lambda} + \frac{\eta}{\lambda + \eta}\, \frac{y}{\eta} \: \in V
	\end{equation}
	and thus $\norm{x+y}_V \leq \lambda + \eta < \norm{x}_V + \norm{y}_V + 2 \varepsilon$. Now the triangle inequality follows by arbitrariness of $\varepsilon$. By 
	\begin{equation}
		\setc{x \in X}{\norm{x}_V < 1} = \mathring{V} \subseteq V \subseteq \bar{V} = \setc{x \in X}{\norm{x}_V \leq 1}	
	\end{equation}
	the topology generated by the seminorms is equivalent to the original one. 
\end{proof}

Despite this one-to-one relation, seminorms are often easier to work with because many well-known topological concepts can be conveniently expressed in terms of them. 
\begin{proposition} 
	Let $(X, \normdot_\alpha)$ and $(Y, \normdot_\beta)$ be locally convex spaces. The following holds:
	\begin{thmenumerate}
		\item \label{prop::locallyConvexSpaces:seminormCharateristEquivalentTopology}
			Another directed fundamental system $\normdot_\gamma$ of seminorms on $X$ generates an equivalent topology if and only if
			\begin{itemize}[label={}, noitemsep]
				\item for every $\alpha$ there exist some $\gamma$ and $C>0$ such that $\normdot_\alpha \leq C \normdot_\gamma$ and
				\item for every $\gamma$ there exist some $\alpha$ and $D>0$ such that $\normdot_\gamma \leq D \normdot_\alpha$.
			\end{itemize}
		\item $X$ is a Hausdorff space if and only if $\norm{x}_\alpha = 0$ for all $\alpha$ implies $x=0$. \label{prop::locallyconvexspaces:HausdorffBySeminorms}
		\item A net $(x_\iota)$ in $X$ converges to $x \in X$ if and only if $\norm{x_\iota - x}_\alpha \overset{\iota}{\longrightarrow} 0$ for all $\alpha$.
		\item A net $(x_\iota)$ in $X$ is a Cauchy net if and only if it is Cauchy in each seminorm. That is, for every $\varepsilon > 0$ and each $\alpha$ there exists an index $\iota_0$ such that for $\mu, \nu > \iota_0$: $\norm{x_\mu - x_\nu}_\alpha < \varepsilon$.
		\item A linear map $T: X \to Y$ is continuous if and only if for every $\beta$ there exist $\alpha$ and $C>0$ such that $\norm{Tx}_\beta \leq C \norm{x}_\alpha$ for all $x \in X$. \label{prop::lcs:linearMapContiniousBySeminorms} \qedhere
	\end{thmenumerate}
\end{proposition} 
The proof is immediate from the definitions. The proposition uses the notion of nets instead of the more handy concept of sequences since countability of the index set of sequences is too weak to characterize topological properties in spaces which are not first-countable. The following theorem expands on this remark.\newpage
\begin{theorem}\label{thm::locallyconvexspaces:countableSeminormsEqualMetrizable}
For every locally convex space $X$ the following are equivalent:
	\begin{enumerate}
		\item $X$ has a countable $0$-neighbourhood, equivalently, it is first-countable. \label{thm::locallyconvexspaces:countableSeminormsEqualMetrizable_firstcountable}
		\item The topology on $X$ can be generated by a countable family of seminorms. \label{thm::locallyconvexspaces:countableSeminormsEqualMetrizable_countableSeminorms}
		\item X is metrizable. The metric $d:X \times X \to \R$ can be chosen to fulfil: \label{thm::locallyconvexspaces:countableSeminormsEqualMetrizable_metrizable}
			\begin{itemize}
				\item $d$ is invariant under translations: $d(x+z, y+z) = d(x,y)$.
				\item $d$ has the same Cauchy-sequences as the original topology, thus completeness as a topological vector space coincides with completeness as a metric space.
			\end{itemize}
	\end{enumerate}
\end{theorem}  
\begin{proof}
	The equivalence of \ref{thm::locallyconvexspaces:countableSeminormsEqualMetrizable_firstcountable} and \ref{thm::locallyconvexspaces:countableSeminormsEqualMetrizable_countableSeminorms} follows from the above discussed bijective correspondence between seminorms and convex, absorbing, balanced $0$-neighbourhoods through the Minkowski functional (see \cref{prop::lcs:equivalenceSeminormsAndConvexBalencedAbsorbingSets}).
	
	For countably many seminorms it is easily observed that the metric
	\begin{equation}
		d(x, y) \defeq \sum_{n=1}^\infty 2^{-n} \frac{\norm{x-y}_\alpha}{1 + \norm{x-y}_\alpha}
	\end{equation}
	has all the claimed properties. Conversely, the balls $B_n \defeq \setc{x \in X}{d(0, x) < \sfrac{1}{n}}$ form a $0$-neighbourhood base indexed by the natural numbers. 
\end{proof}

\begin{defn}
	A complete, metrizable, locally convex space is called a \emphDef{Fréchet space}. Furthermore, a \emphDef{graded Fréchet space} is a Fréchet space together with a fixed increasing fundamental system of seminorms.
\end{defn}
\begin{example}[Smooth functions] \label{ex::lcs:smoothFunctionOnOpenSubsetFrechet}
The space of smooth functions $C^\infty(U)$ on an open subset $U \subseteq \R$ is the prime example of a Fréchet space.
Recall from \cref{ex::locallyConvexSpaces:SmoothFunctionOnOpenSubset} that the seminorms
\begin{equation}
	\norm{\phi}_{K, k} \defeq \sup_{\substack{x \in K \\ i \leq k}}{\abs*{\difff{\phi}{x}{i}(x)}} 	\qquad \text{for } \phi \in C^\infty(U), k \in \N, \text{compact } K \subseteq U
\end{equation}  
define a locally convex topology on $C^\infty(U)$. All defining Fréchet properties are easily seen, except metrizability which needs a few explanatory words. The decisive technical ingredient is delivered by a countable exhaustion of $U$ with compact sets. This is a countable family $(K_n)$ of compact subsets $K_n \subseteq U$ such that every $K_n$ lies in the interior of its successor $K_{n+1}$ and the union of all $K_n$ yields $U$. Such a system always exists (for finite-dimensional manifolds) 
and has the property that every other compact subset $K \subseteq U$ is contained in some $K_n$. Thus the collection $\norm{\cdot}_{K_n, k}$ is a countable collection of seminorms equivalent to the original one (since every $\norm{\cdot}_{K, k}$ is dominated by some $\norm{\cdot}_{K_n, k}$).
\end{example}

The fact that the topology on a Fréchet space, or more generally a locally convex space, is defined by a family of seminorms suggests to think of such spaces as a collection of semi-normed spaces. In fact, one can mathematically substantiate this intuition by considering limits of topological spaces.
\begin{defn}
	\parpic(3.9cm, 5\baselineskip)[r]{
	\begin{tikzcd}[column sep=small, row sep=small, ampersand replacement=\&]
			X_\alpha \& \& X_\gamma \arrow[swap]{ll}{\Psi_{\alpha\gamma}} \arrow{dl}{\Psi_{\beta\gamma}}\\
			\& 	X_\beta \arrow{ul}{\Psi_{\alpha\beta}} \&
	\end{tikzcd}
	}
	Let $X_\alpha$ be a family of topological vector spaces indexed by some directed set. For $\alpha < \beta$ let $\Psi_{\alpha\beta}:X_\beta \to X_\alpha$ be continuous linear maps, which will be named bonding maps. If the diagram on the right commutes for every $\alpha < \beta < \gamma$, then the tuple $(X_\alpha, \Psi_{\alpha\beta})$ is called an \emphDef{inverse directed system}.

	The \emphDef{inverse limit} of an inverse system is defined to be the set 
	\begin{equation} 
	\liminv X_\alpha \defeq \setc*{(x_\alpha) \in \prod X_\alpha}{\Psi_{\alpha\beta}(x_\beta) = x_\alpha \quad \text{for } \alpha < \beta}
	\end{equation}
	 equipped with the initial topology with respect to the natural projections $\Psi_\alpha: \liminv X_\alpha \hookrightarrow \prod X_\alpha \overset{\pr_\alpha\;}{\longrightarrow} X_\alpha$.	
\end{defn}
\begin{corollary}[Universal property] \label{cor::locallyconvexspaces:inverseLimitUniversalProperty}
	\parpic[r]{
	\begin{tikzcd}[column sep=small, ampersand replacement=\&]
			          \& Y \arrow[swap]{ldd}{f_\alpha} \arrow{d}{f} \arrow{rdd}{f_\beta} 
									               	\&  \\[0.8cm]
			          \& \liminv X_\alpha \arrow{dl}{\Psi_\alpha} \arrow[swap]{rd}{\Psi_\beta}
			       								   	\&  \\
			X_\alpha  \&                  			\& X_\beta \arrow{ll}{\Psi_{\alpha\beta}}
	\end{tikzcd}
	}
	Let $Y$ be a topological space and $f_\alpha: Y \to X_\alpha$ continuous maps which are compatible with a given inverse system $(X_\alpha, \Psi_{\alpha\beta})$, that is, $f_\alpha = \Psi_{\alpha \beta} \, \circ \, f_\beta$ for all $\alpha < \beta$. There exists a unique continuous map $f: Y \to \liminv X_\alpha$ making the diagram on the right commutative for all $\alpha < \beta$.
	Moreover, the resulting map $f$ inherits a possible linearity from the defining maps.
\end{corollary}
\begin{proof}
	The universal property of the product $\prod X_\alpha$ directly translates to the inverse limit due to the compatibility condition on $f_\alpha$.
\end{proof}

If all spaces $X_\alpha$ in an inverse system are Hausdorff, so is their product and hence the inverse limit is a closed subspace of $\prod X_\alpha$. This directly implies that the inverse limit of complete Hausdorff spaces is again complete. Furthermore, in the category of locally convex spaces the limit space carries compatible seminorms explicitly given by composing seminorms on the building blocks with the corresponding projections. In particular, the limit of countably many metrizable spaces is metrizable as well.  

\begin{theorem}[Locally convex spaces as inverse limits]
	The following holds:
	\begin{itemize}
		\item Every locally convex space is topologically isomorphic to a dense subset of an inverse limit of Banach spaces.
		\item Every metrizable locally convex space is topologically isomorphic to a dense subset of an inverse limit of countably many Banach spaces.
		\item Every complete locally convex space is topologically isomorphic to an inverse limit of Banach spaces.		
		\item Every Fréchet space is topologically isomorphic to an inverse limit of countably many Banach spaces.		
	\end{itemize}
	The building blocks of these inverse systems are often called \emphDef{local Banach spaces}.
\end{theorem}
\begin{proof}
	Let $(X, \normdot_\alpha)$ be a locally convex space. Due to the absolute homogeneity and the triangle inequality the kernel $\ker_\alpha$ of $\normdot_\alpha$ is a linear subspace of $X$ and $p_\alpha(\equivClass{x}_\alpha) \equiv p_\alpha(x + \ker_\alpha) \defeq \norm{x}_\alpha$ is a well-defined norm on $X / \ker_\alpha$. Hence the completion $X_\alpha$ of this quotient space is a Banach space. The projection $X \to X / \ker_\alpha$ induces a linear, continuous map $\Psi_\alpha: X \to X_\alpha$. 
	To complete the construction of the inverse system, define the partial order $\alpha \leq \beta$ if and only if $\normdot_\alpha \leq \normdot_\beta$. Then $\ker_\alpha \supseteq \ker_\beta$, so that the assignment $\equivClass{x}_\beta \mapsto \equivClass{x}_\alpha$ defines a linear, continuous map $\Psi_{\alpha\beta}: X_\beta \to X_\alpha$. Now it is easy to see that $(X_\alpha, \Psi_{\alpha\beta})$ is indeed an inverse system, which is countable for metrizable spaces $X$. 

	By construction the family $\Psi_\alpha$ is compatible with the inverse system and thus by \autoref{cor::locallyconvexspaces:inverseLimitUniversalProperty} defines a linear, continuous map $\Psi: X \to \liminv X_\alpha$. Moreover, $\Psi$ is injective because $X$ is a Hausdorff space and so the only point lying in the kernel of all defining seminorms is $0$ (see \cref{prop::locallyconvexspaces:HausdorffBySeminorms}). Thus it is left to show that $\Psi(X)$ is dense in $\liminv X_\alpha$. But this follows from the fact that for every $0$-neighbourhood $U \subseteq \liminv X_\alpha$ there exists a $0$-neighbourhood $V \subseteq X_\alpha$ (for some $\alpha$) such that $\Psi_\alpha^{-1}(V) \subseteq U$ and since $\Psi_\alpha(X) = \sfrac{X}{\ker_\alpha}$ is dense in $X_\alpha$ (by the definition of completion).
\end{proof}

\begin{example}[Smooth functions]
	Consider the space of smooth functions on an open subset $U \subseteq \R$ equipped with the Fréchet topology discussed in \cref{ex::locallyConvexSpaces:SmoothFunctionOnOpenSubset}. The previous theorem renders $C^\infty(U)$ as an inverse limit of the spaces $C^k(U)$ with their usual $C^k$-norms. 

	Due to the Sobolev inequalities an equivalent topology is induced by the usual Sobolev seminorms (see \cref{sec::sectionSpace} for an explicit formula). This observation results in the Sobolev chain $C^\infty(U) = \liminv H^k(U)$. 
\end{example}

The representation of Fréchet spaces as an inverse limit of countably many Banach spaces is a valuable tool for transferring well-known theorems to this new context (the prime example being the Nash-Moser-Theorem discussed in \autoref{sec:locallyconvexspaces:nashMoserInverseTheorem}). In a related approach one first studies the problem at hand on suitable Banach spaces and then tries to glue them together. For example, regularisation of a Sobolev-solution of a partial differential equation equates roughly speaking to proving compatibility with the inverse system $H^k(U)$. 

Unfortunately, the category of Fréchet spaces has the big drawback of not being closed under dualizing (this observation is probably due to \textcite{Grothendieck1954}). Therefore dual spaces, and more generally spaces of linear maps, will be avoided in the development of differential calculus. 
\begin{theorem}[{\parencite[Proposition 29.1.7]{Koethe1969}}] \label{prop::locallyConvexSpaces:dualOfFrechetNotMetrizable}
	The strong\footnote{The strong topology on the dual space $F'$ is the topology $b^*$ of uniform convergence on bounded subsets generated by the seminorms $p_M(\xi) \defeq \sup_{x \in M} \abs{\xi(x)}$. Here, $M \subseteq F$ runs through the bounded subset, that is, each generating seminorm is bounded on $M$.} dual $(F', b^*)$ of a Fréchet space $F$ is metrizable if and only if $F$ is normable. In particular, duals of non-Banach Fréchet spaces are not Fréchet.
\end{theorem}
\begin{proof}
	The idea is to prove that the strong double dual $(F'', b^{**})$ is normable. From this the claim follows, since $F \subseteq F''$ and the subspace topology induced from $b^{**}$ coincides with the original topology on $F$ \parencite[Theorem 23.23]{MeiseVogt1992}.

	Let $\normdot_k$ be a countable directed fundamental system of seminorms with associated unit tubes $U_k$. The polars $\mathring{U}_k \defeq \setc{\xi \in F'}{\abs{\xi(x)} \leq 1 \quad \text{for } x \in U_k}$ constitute a family of bounded subsets of $F'$ absorbing all other bounded subsets, i.e. for every bounded subset $\tilde B \subset F'$ there exist a $k$ and $\eta > 0$ with $\tilde B \subseteq \eta \cdot \mathring{U}_k$. This fact is proven in \parencite[Lemma 25.5]{MeiseVogt1992} by appealing to the Banach-Mackey theorem.

	Assume that $F'$ is metrizable, then there exist a countable $0$-neighbourhood basis of convex and absorbing subsets $V_k$. For each $k$, choose $\lambda_k > 0$ such that $\lambda_k \cdot \mathring{U}_k \subseteq V_k$, which is possible because the polars are bounded. Now the set $B \defeq \cup_k \,\lambda_k \cdot \mathring{U}_k$ is a bounded set absorbing all $\mathring{U}_k$ and thus all other bounded subsets. This has the consequence that the seminorms $\norm{y}_k \defeq \sup_{\xi \in \mathring{U}_k} \abs{y(\xi)}$ defining the strong topology $b^{**}$ on the double dual $F''$ degenerate to a single norm. This completes the proof.
\end{proof}
An intuitively appealing way to make the previous theorem plausible resorts to the inverse limit representation of Fréchet spaces. Taking the dual of the inverse system of local Banach spaces results in a direct system (roughly speaking, an inverse system with reversed arrows). Thus a complete, metrizable dual is a countable inverse and direct limit at the same time, which is only possible if both families are finite and hence Banach spaces. However this argument is only completely valid for a subclass of Fréchet spaces. It implicitly requires that the topology induced from the direct system equals the strong topology, which only holds for so called distinguished Fréchet spaces (see \parencite[Remark 25.13]{MeiseVogt1992}).

\section{Integral and differential calculus} \label{sec::lcs:integralDifferentalCalculus}
Closely following the spirit of variational calculus, this section develops the integral and differential calculus for mappings between arbitrary locally convex spaces. An extensive discussion was given by \textcite{Hamilton1982} in the case of Fréchet spaces and later generalized to sequentially complete locally convex spaces \parencite{Milnor1984}. Recently, \textcite{Gloeckner2002} observed that the same argumentation holds without any completeness condition. The following presentation is based on the just mentioned references.

The fundamental theorem of calculus is the pillar on which the remaining theory rests and thus the theory of integration in locally convex spaces is developed first. The weak integral generalizes the idea to integrate vector-valued functions component-wise, just that now linear, continuous functionals play the role of components.
\begin{defn}
	Let $\gamma: \R \supseteq [a,b] \to X$ be a continuous curve in a locally convex space $X$. If there exists an element of $X$, denoted by $\int_a^b \gamma(t) \dif t$, such that
	\begin{equation}
		\dualpair*{\xi}{\int_a^b \gamma(t) \dif t} = \int_a^b \dualpair{\xi}{\gamma(t)} \dif t \qquad \text{for all } \xi \in X',
	\end{equation}
	then it is called the \emphDef{Gelfand-Pettis} or \emphDef{weak integral} of $\gamma$. In this case, $\gamma$ is called \emphDef{weakly integrable}. Here, $\dualpair{\,\cdot\,}{\,\cdot\,}$ denotes the dual pairing.
\end{defn}
If the integral exists, then it is uniquely defined since the linear continuous functionals separate points in $X$ by the Hahn-Banach theorem \parencite[Theorem 22.12]{MeiseVogt1992}. Existence is ensured for sequentially complete spaces, since in that case every weak integral can be obtained as the limit of Riemann partial sums (see \parencite[Proposition 1.1.12]{Gloeckner2005}).
\newpage
\begin{defn}
 	A continuous curve $\gamma: \R \supseteq I \to X$ defined on a (non-empty) open subset $I$ is called \emphDef{continuously differentiable} or \emphDef{$C^1$} if the limit
 	\begin{equation}
 		\gamma'(t) \defeq \lim_{h \rightarrow 0} \frac{\gamma(t+h) - \gamma(t)}{h}
 	\end{equation}
 	exists for all $t \in I$ and the assignment $\gamma': t \mapsto \gamma'(t)$ is continuous on $I$.  
\end{defn} 

\begin{theorem}[Fundamental theorem of calculus] \label{prop::lcs:fundamentalTheoremPath}
 	Let $X$ be a locally convex space. The following holds:
 	\begin{enumerate}
 		\item For every continuously differentiable curve $\gamma: I \to X$ with $[a,b] \subseteq I$:
 			\begin{equation}
 				\gamma(b) - \gamma(a) = \int_a^b \gamma'(t) \dif t.
 			\end{equation}
 			In particular, the weak integral $\int_a^b \gamma'(t) \dif t$ exists.
 		\item Let $\gamma: [a,b] \to X$ be a continuous curve. If the weak integral $\Gamma(t) \defeq \int_a^t \gamma(s) \dif s$ exists for every $t \in [a,b]$, then $\Gamma$ is a continuously differentiable curve and fulfils $\Gamma'(t) = \gamma(t)$. \qedhere 
 	\end{enumerate}
\end{theorem}
\begin{proof}
 	\begin{thmenumerate}*
 		\item Let $\xi \in X'$ be arbitrary. Then a simple calculation using linearity and continuity of $\xi$ shows that $\xi \circ \gamma$ is a $C^1$-curve in $\R$,
 			\begin{align}
 				(\xi \circ \gamma)' &= \lim_{h \rightarrow 0} \frac{\dualpair{\xi}{\gamma(t+h)} - \dualpair{\xi}{\gamma(t)}}{h}	\\
 									&= \dualpair*{\xi}{\lim_{h \rightarrow 0} \frac{\gamma(t+h) - \gamma(t)}{h}} = \xi \circ \gamma'.
 			\end{align}	
 			By resorting back to the `ordinary' fundamental theorem of calculus one has
 			\begin{equation}
 			 	\dualpair{\xi}{\gamma(b) - \gamma(a)} = (\xi \circ \gamma)(b) - (\xi \circ \gamma)(a) =  \int_a^b (\xi \circ \gamma)'(t) \dif t =  \int_a^b \dualpair{\xi}{\gamma'(t)} \dif t.
 			 \end{equation} 
 			 Hence $\gamma(b) - \gamma(a)$ fulfils the defining relation of the Pettis integral $\int_a^b \gamma'(t) \dif t$.
  		\item As this result is not needed in the following, but its proof requires additional preparation, the reader is referred to \parencite[Proposition 1.1.9]{Gloeckner2005}. \qedhere %
 	\end{thmenumerate}
\end{proof}

Before proceeding to differential calculus, a few properties of integration are collected for later use.
\begin{proposition}
	Let $X$ be a locally convex space. Suppose the continuous curves $\gamma$ and $\eta$ in $X$ are weakly integrable on $[a,b]$ with $b>a$. Then,
	\begin{thmenumerate}
		\item (Linearity)\\
			$\gamma + c \cdot \eta$ is weakly integrable for all $c \in \R$ and
			\begin{equation}
				\int_a^b (\gamma(t) + c \cdot \eta(t)) \dif t = \int_a^b \gamma(t) \dif t + c \int_a^b \eta(t) \dif t.
			\end{equation}

		\item \label{prop:locallyconvexspaces:integrationMeanValueTheorem} 
			(Mean value theorem)\\
			\begin{equation}
				\frac{1}{b-a} \, \int_a^b \gamma(t) \dif t \: \in \: \cco (\img \gamma).
			\end{equation}
			The closed convex hull $\cco A$ of a subset $A \subseteq X$ is the smallest closed convex subset of $X$ containing $A$. 

		\item for every continuous seminorm $\normdot$ on $X$,
			\begin{equation}\label{eq::locallyconvexspaces:seminormIntegralInequality}
				\norm*{\int_a^b \gamma(t) \dif t} \leq (b-a) \cdot \max_{t \in [a,b]} \norm{\gamma(t)}.  
			\end{equation}
	\end{thmenumerate}
\end{proposition}
\begin{proof}
	\begin{thmenumerate}*
		\item Follows directly from the definition, since $\xi$ and the one-dimensional Riemann integration are linear.
		\item The proof is by contradiction, so assume $\frac{1}{b-a} \, \int_a^b \gamma(t) \dif t \notin \cco (\img \gamma)$. 
			Then by the Hahn-Banach separation theorem \parencite[Theorem B.5.8]{Gloeckner2005} there exist $\xi \in X'$ and $\varepsilon > 0$ such that $\dualpair{\xi}{\cco (\img \gamma)} + \varepsilon \leq \dualpair{\xi}{\frac{1}{b-a} \, \int_a^b \gamma(t) \dif t}$. In particular, the image of~$\gamma$ can be separated from the integral in an analogous way. Using this inequality, integration yields
			\begin{equation}
				\dualpair*{\xi}{\int_a^b \gamma(t) \dif t} = \int_a^b \dualpair{\xi}{\gamma(t)} \dif t \leq \dualpair*{\xi}{\int_a^b \gamma(t) \dif t} - (b-a)\cdot \varepsilon
			\end{equation} 
			and thus the desired contradiction by recalling $b>a$.
		\item Note that $\setc{x \in X}{\norm x \leq \max_{t \in [a,b]} \norm{\gamma(t)}}$ is a closed convex set containing $\img \gamma$. The claim is now a consequence of the above mean value theorem. \qedhere
	\end{thmenumerate}
\end{proof}
\begin{proposition}[Continuity of parameter-dependent integrals] %
	\label{prop::locallyconvexspaces:continuityParameterIntegral}
	Let $X$ be a locally convex space and $\gamma: I \times [a,b] \rightarrow X$ a continuous map, which describes a family of curves $\gamma(s, \cdot)$ parametrized by an interval $I \subseteq \R$. Furthermore, suppose that $\gamma(s, \cdot)$ is weakly integrable for all $s \in I$. Then the map
	\begin{equation}
		I \ni s \mapsto \int_a^b \gamma(s, t) \dif t \in X
	\end{equation}
	is continuous.
\end{proposition}
\begin{proof}
	Let $\normdot$ be a defining seminorm on $X$ and $\varepsilon > 0$. By the above inequality \eqref{eq::locallyconvexspaces:seminormIntegralInequality} one estimates
	\begin{equation}
	  	\norm*{\int_a^b \gamma(s, t) \dif t - \int_a^b \gamma(s_0, t) \dif t} \leq (b-a) \max_{t \in [a,b]} \norm{\gamma(s,t) - \gamma(s_0, t)}
	\end{equation}
	Hence it is enough to find an upper bound for the second term to conclude continuity of the parametrized integral at $s_0 \in I$. Denote by $t_0 \in [a,b]$ the point where the maximum is attained. Since $\gamma(\cdot, t_0)$ is continuous at $s_0$, there exists for every $\tilde \varepsilon > 0$ a $\delta > 0$ such that $\norm{\gamma(s,t_0) - \gamma(s_0, t_0)} < \tilde \varepsilon$ holds for $\abs{s-s_0} < \delta$. Now the choice $\tilde \varepsilon = \sfrac{\varepsilon}{(b-a)}$ completes the proof.
\end{proof}

As already indicated, the definition of the derivative is roughly along the lines of variational calculus.
\begin{defn}
 	Let $X$ and $Y$ be locally convex spaces and $U \subseteq X$ a open subset. For a continuous map $f: U \to Y$ the \emphDef{Gâteaux derivative} at $x \in U$ in the direction $h \in X$ is defined to be
 	\begin{equation}
 		(\dif f)_x (h) \defeq \lim_{t \rightarrow 0} \frac{f(x+th)-f(x)}{t}\, \in Y.
 	\end{equation}
 	If the limit exists for all $h \in X$, the function $f$ is differentiable at $x \in U$. A function differentiable at every point of $U$ is simply called \emphDef{differentiable}. 

 	A \emphDef{continuously differentiable} or \emphDef{$C^1$} function $f$ is by definition differentiable in $U$ and the derivative as a map $\dif f: U \times X \to Y$ is continuous. The higher derivatives are obtained recursively
 	\begin{equation}
 		(\dif^{n} f)_x (h_1, \dotsc ,h_n) \defeq \lim_{t \rightarrow 0} \frac{(\dif^{n-1} f)_{x+th_n}(h_1, \dotsc, h_{n-1})-(\dif^{n-1} f)_{x}(h_1, \dotsc, h_{n-1})}{t}\, \in Y.
 	\end{equation}
 	and $f$ is called of class $C^n$, if all limits exists and $\dif f:U \times X \to Y$ is a $C^{n-1}$ map. A map is \emphDef{smooth} if it is of class $C^n$ for all $n \in \N$.
\end{defn}
\begin{remarks}
 	\item In the above definition the function is required to be continuous, which in contrast to the Banach case does not immediately follow from the existence of the differential quotient. This is seen by noticing the existence of discontinuous linear maps and on the other hand that all linear maps are differentiable. 

 	Nevertheless, continuity of the derivative $\dif f: U\times X \to Y$ is strong enough to ensure continuity of $f$, see \parencite[Lemma II.2.3 (iii)]{Neeb2006}. %
 	\item In the context of Banach spaces this notion of continuous differentiability is strictly weaker than the standard $C^1_{\text{Banach}}$ property, where the map $U \ni x \mapsto (\dif f)_x \in (L(X,Y), \normdot_{\text{op}})$ is required to be continuous. See \textcite[Example~I.6]{Neeb2006} for a $C^1$ function violating the latter condition. However, $C^2$ implies $C^1_{\text{Banach}}$ and thus the notion of smoothness coincides in both approaches \parencite[Remark~I.2.2]{Neeb2006}. 
 	\item Once one leaves the secure realm of Banach spaces, the number of admissible and in some sense natural definitions of continuous differentiability are increased by many orders of magnitude. \textcite{Keller1974} reviews 9 different variations of $C^1$-maps which at least collapse to three inequivalent notions of smoothness in Fréchet spaces (6 for arbitrary locally convex spaces). Let alone the 25 different formulations of the first derivative discussed in \parencite{AverbukhSmolyanov1968}.

	In particular, the `convenient setting' developed in \parencite{KrieglMichor1997} is to be pointed out, since it has the advantage to form a cartesian closed category. But the technical avail comes at the cost that convenient smooth maps are not necessarily continuous. Thus, for example, insights from topological group theory are not available in this context.

 	The approach to differential calculus given in this work was first established by \textcite{Michal1938, Bastiani1964} and has since then proven itself reliable for global analytic problems. It is also the calculus used by \textcite{Hamilton1982, Milnor1984, Neeb2006}. \qedhere		     

\end{remarks}

 \begin{proposition}[Properties of the derivative]
 	Let $f: X \supseteq U \to Y$ be a continuously differentiable map defined on an open subset $U$.
 	\begin{thmenumerate}
 		\item $(\dif f)_x: X \to Y$ is linear and continuous for every $x \in U$. \label{prop::locallycovnexspaces:derivativeLinearContInSecondArgument}
 		\item \label{prop::locallycovnexspaces:fundamentalTheorem}
 			(Fundamental theorem of calculus)\\
 			If the straight path from $x$ to $x+h$ lies in $U$, then 
 			\begin{equation}
 				f(x+h) - f(x) = \int_0^1 (\dif f)_{x+th} (h) \dif t.
 			\end{equation}
 		\item (Chain rule)\\
 			Let $V \subseteq Y$ be an open subset containing the image of $f$ and $g:Y \supseteq V \to Z$ a $C^1$-map. Then $g \circ f: X \supseteq U \to Z$ is continuously differentiable and
 			\begin{equation}
 				\dif (g\circ f)_x = (\dif g)_{f(x)} \circ (\dif f)_x\,.
 			\end{equation}
 		\newpage
 		\item (Taylor series)\\
 			If $f$ is of class $C^n$ and the straight path from $x$ to $x+h$ lies in $U$, then
 			\begin{multline}
 				f(x+h) = f(x) + (\dif f)_x(h) + \dotsb + \frac{1}{(n-1)!} (\dif^{n} f)_x(h, \dotsc, h) \\
 				+ \frac{1}{(n-1)!} \int_0^1 (1-t)^{n-1} (\dif^{n} f)_{x+th}(h, \dotsc, h) \dif t 
 			\end{multline} 
 			holds. \qedhere
 	\end{thmenumerate}
 \end{proposition}
 \begin{proof}
 	Note that the proof is in a slightly different order then the statements above.
 	\begin{enumerate}
 		\item[(ii)] %
 			Consider the map $[0,1] \ni t \mapsto f(x+th) \in X$, which is a $C^1$-path because its derivative is given by
 			\begin{equation}
 				\diffAt{}{t}{t}f(x+th) = \lim_{s \rightarrow 0} \frac{f(x+(t+s)h) - f(x+th)}{s} = (\dif f)_{x+th}(h).
 			\end{equation}
 			Due to this expression for the derivative, application of the fundamental theorem for paths (\cref{prop::lcs:fundamentalTheoremPath}) now yields
 		\begin{equation}
 			\int_0^1 (\dif f)_{x+th}(h) \dif t = \int_0^1 \diffAt{}{t}{t}f(x+th) \dif t = f(x+h) - f(x).
 		\end{equation}
 		\item[(i)] %
 			Scalar homogeneity $(\dif f)_x (\lambda h) = \lambda (\dif f)_x (h)$ follows immediately from the definition. %
 			To prove additivity the fundamental theorem is applied 
 			\begin{align}
 				\frac{f(x+t(h_1 + h_2)) - f(x)}{t} 	&= \frac{f(x+t(h_1 + h_2)) - f(x + t h_1)}{t} + \frac{f(x+t h_1) - f(x)}{t}\\
 													&= \int_0^1 (\dif f)_{x+th_1 + s t h_2}(h_2) \dif s + \int_0^1 (\dif f)_{x+sth_1}(h_1) \dif s.
 			\end{align} 
 			Continuity of parametrized integrals, \cref{prop::locallyconvexspaces:continuityParameterIntegral}, allows one to take the limit $t \rightarrow 0$ under the integral sign to conclude $(\dif f)_x (h_1 + h_2) = (\dif f)_x (h_1)  + (\dif f)_x(h_2)$. Finally, $(\dif f)_x$ is continuous because $\dif f$ is.
 		\item[(iii)] %
 			For the proof of the chain rule the following characterisation of $C^1$-maps is useful. It follows directly from the fundamental theorem, \cref{prop::locallycovnexspaces:fundamentalTheorem}.
			\begin{lemma} \label{prop::locallyConvexSpace:charaterisationDifferentialByDifferentialQuotient} %
			 	A function $f:U \to Y$ defined on an open subset $U \subseteq X$ is continuously differentiable if and only if there exists a continuous function $f^{[1]}:U^{[1]} \rightarrow Y$ defined on the open set $U^{[1]} \defeq \setc{(x,h,t) \in U \times X \times \R}{x+th \in U}$ and fulfilling $f^{[1]}(x,h,t) \defeq \frac{1}{t} (f(x+th) - f(x))$ for $t \neq 0$. In this case, $(\dif f)_x (h) = f^{[1]}(x, h, 0)$.
			\end{lemma}
			Equipped with this result, the chain rule is obtained by purely algebraic rearrangements
			\begin{align}
				\frac{(g \circ f) (x + th) - (g \circ f) (x)}{t} 	&= \frac{g(f(x) + t f^{[1]}(x,h,t)) - g (f(x))}{t} \\
																	&= g^{[1]}(f(x), f^{[1]}(x,h,t), t).
			\end{align}
			Since this expression is continuous on $U^{[1]}$, the previous lemma shows that $g \circ f$ is of class $C^1$. Moreover, evaluation at $t=0$ yields the desired formula.
		\item[(iv)]
			Taylor's theorem can be directly obtained by iterating the results from the first two points and integrating by parts \parencite[Theorem 3.5.2 and 3.5.6]{Hamilton1982}.\qedhere %
 	\end{enumerate}
 \end{proof}

\begin{proposition}[Partial derivatives]
	Let $X_1$ and $X_2$ be locally convex spaces and $f: X_1 \times X_2 \supseteq U \to Y$ a continuous map defined on an open subset $U$ of the product space $X_1 \times X_2$. The map $f$ is $C^1$ if and only if the \emphDef{partial derivatives} 
	\begin{align}
		(\dif_1 f)_{x_1, x_2} (h) &\defeq \lim_{t \rightarrow 0} \frac{f(x_1+th, x_2)-f(x_1, x_2)}{t}\, \in Y\\
		(\dif_2 f)_{x_1, x_2} (h) &\defeq \lim_{t \rightarrow 0} \frac{f(x_1, x_2+th)-f(x_1, x_2)}{t}\, \in Y
	\end{align}
	exists for all $(x_1, x_2) \in U$ and the resulting map $\dif_1 f: U \times X_1 \to Y$ as well as  $\dif_2 f: U \times X_2 \to Y$ is continuous. In this case the total derivative of $f$ is the sum of the partial derivatives:
	\begin{equation}
		(\dif f)_{x_1, x_2} (h_1, h_2) = (\dif_1 f)_{x_1, x_2} (h_1) + (\dif_2 f)_{x_1, x_2} (h_2). 
	\end{equation}
\end{proposition}
\begin{proof}
	It is obvious that the partial derivatives exists and are continuous if $f$ is a $C^1$-map. Conversely, continuity of $\dif_j f: U \times X_j \to Y$ allows one to establish the identity
	\begin{equation}\begin{split}
		\MoveEqLeft \frac{f(x_1+th_1, x_2 + th_2)) - f(x_1, x_2)}{t}\\[2ex]
		&=\frac{f(x_1+th_1, x_2 + th_2)) - f(x_1+th_1, x_2)}{t} + \frac{f(x_1+th_1, x_2)) - f(x_1, x_2)}{t} \\
		&= \int_0^1 (\dif_2 f)_{x_1 + th_1, x_2+sth_2}(h_2) \dif s + \int_0^1 (\dif_1 f)_{x_1 + s t h_1, x_2}(h_1) \dif s
	\end{split}\end{equation}
	for sufficiently small $t$ (compare with the proof of \cref{prop::locallycovnexspaces:derivativeLinearContInSecondArgument}). By continuity of the partial derivatives $\dif_j$, \autoref{prop::locallyconvexspaces:continuityParameterIntegral} about parameter-dependent integrals is applicable. Hence the limit $t \rightarrow 0$ yields $(\dif f)_{x_1, x_2} (h_1, h_2) = (\dif_1 f)_{x_1, x_2} (h_1) + (\dif_2 f)_{x_1, x_2} (h_2)$. In particular, $f$ is $C^1$.
\end{proof}

\section{Nash-Moser inverse function theorem} \label{sec:locallyconvexspaces:nashMoserInverseTheorem}
Many problems and complications in the study of locally convex spaces originate from the absence of a general inverse function theorem in  categories more general than the Banach setting. Hence the search for an adequate replacement becomes all the more important. The theorem of Nash and Moser \parencites{Nash1956, Moser1966} provides such a suitable extension of the classical result to a subcategory of Fréchet spaces, but it requires stronger assumptions. This section discusses the relevant notion of the so called tame Fréchet spaces and points out why stronger assumptions are needed for inverse-function-like results. An extensive presentation of these ideas can be found in \parencite{Hamilton1982}, on which the discussion in this section is based.

The basic idea underlying the Nash-Moser theorem is to exploit the representation of a Fréchet space as a projective limit of local Banach spaces. Using the inverse function theorem on each building block separately, one gets a family of inverse functions, which then need to be glued together in a reasonable way. This requires an approximation of the behaviour of the function in one local Banach space by the knowledge of its properties in another one. Recall that a graded Fréchet space has a distinguished increasing fundamental system of seminorms $\normdot_k$. If $f: X \to Y$ is a map between graded spaces, then in general one has no knowledge about the series $\norm{f(x)}_k$ except that it is increasing in $k$. The concept of tame maps allows one to control the behaviour more explicitly while still leaving enough scope for applications.

\begin{defn} \label{defn::FrechetTame:TameNonLinearMap}
	A continuous map $f: X \supseteq U \to Y$ between graded Fréchet spaces defined on an open subset $U \subseteq X$ is called \emphDef{$r$-tame of base $b$} if every point $x \in U$ has a neighbourhood $V \subseteq U$ such that the following tame estimate holds: 
	\begin{equation} \label{eq::locallyconvexspaces:tameEstimate}
		\text{For all } n \geq b \text{ there exists } C_{n, V} \in \R \text{ such that } \norm{f(x')}_n \leq C_{n, V} (1 + \norm{x'}_{n+r}) \text{ for } x' \in V.
	\end{equation}
	Furthermore, a smooth map $f$ is called \emphDef{$r$-tame smooth of base $b$} if it is tame and all derivatives $\dif^{n} f: U \times X^n \to Y$ are also tame\footnote{The product $X \times Y$ of graded products carries the natural product grading $\norm{(x,y)}_k = \norm{x}_n + \norm{y}_n$.} with the same degree $r$ and base $b$.
\end{defn} %
To simplify the notation, the degree and the base are often not explicitly named. One easily checks that the composition of tame maps is tame again. Furthermore, if the domain space $X$ is finite-dimensional or the target $Y$ is a Banach space, then any continuous map $f: X \supseteq U \to Y$ is automatically tame \parencite[Example 2.1.4]{Hamilton1982}. In the case of linear maps, tameness is equivalent to an apparently stronger condition as the next lemma demonstrates.

\begin{lemma}
	A linear map $T: X \to Y$ between graded Fréchet spaces is $r$-tame if and only if for all $x \in X$ holds:
	\begin{equation} \label{eq::locallyconvexspaces:tameLinearEstimate}
		\text{For all } n \geq b \text{ there exists } C_n \in \R \text{ such that } \norm{T(x)}_n \leq C_n \norm{x}_{n+r}.
	\end{equation}
	The starting index $b$ in this inequality can be greater then the base of the tame map.
\end{lemma}
\begin{proof}
	Sufficiency of \eqref{eq::locallyconvexspaces:tameLinearEstimate} for the tame estimate is clear and, moreover, implies continuity by \cref{prop::lcs:linearMapContiniousBySeminorms}. 

	Conversely, tameness of $T$ ensures the existence of a neighbourhood $V \subseteq X$ of $0$ such that $\norm{T(y)}_n \leq C_n (1 + \norm{y}_{n+r})$ holds on it for $n \geq b$. Since the $\varepsilon$-tubes of the family of seminorms form a $0$-neighbourhood base, the particular form $V = \setc{y \in X}{\norm{y}_k \leq \varepsilon}$ for some $k \geq 0$ and $\varepsilon > 0$ can be assumed. The original tame estimate \eqref{eq::locallyconvexspaces:tameEstimate} is amplified to the linear approximation \eqref{eq::locallyconvexspaces:tameLinearEstimate} by exploiting the homogeneity symmetry of the left side of the inequality. That is, replacing $y$ by $c y$ for some positive constant $c$ yields $\norm{T(y)}_n \leq C_n (\sfrac{1}{c} + \norm{y}_{n+r})$. 

	Every non-zero $x \in X$ is rescaled to $V$ by setting $y \defeq \sfrac{\varepsilon}{\norm{x}_k} x$ and so one obtains the inequality $\norm{T(x)}_n \leq C_n (\sfrac{\norm{x}_k}{\varepsilon} + \norm{x}_{n+r})$. Define the new base $\tilde b$ in such a way to fulfil $\tilde{b} + r \geq k$ and $\tilde b \geq b$. Now the monotony of the family of seminorms implies the claimed inequality $\norm{T(x)}_n \leq C_n (\sfrac{1}{\varepsilon} + 1)\norm{x}_{n+r}$ for $n \geq \tilde b$.
\end{proof}
If one thinks of $T: X \to Y$ as a linear differential operator between function spaces, then $r$-tameness of $T$ implies a `loss of maximally $r$ derivatives'.

Besides estimates on maps, the Nash-Moser theorem also requires the space under consideration to be `sweet-tempered' in the sense that it provides a sufficiently Banach-like environment. In some respects (which will be made exact in the following), the prototype is the space $\Sigma(B)$ of exponentially strong falling sequences in a Banach space $(B, \normdot_{\text{B}})$. The grading and thus the Fréchet topology on $\Sigma(B)$ is given by seminorms
\begin{equation}
	\norm{(x_i)}_k \defeq \sum_{i=0}^\infty e^{k\cdot i} \, \norm{x_i}_{\text{B}}.	
\end{equation} 
A sequence $(x_i)$ in $B$ is an element of $\Sigma(B)$ if and only if $\norm{(x_i)}_k < \infty$ for all $k$.

\begin{defn}
	A graded Fréchet space $X$ is called \emphDef{tame} if there exists a Banach space $B$ such that $X$ is a tame direct summand in $\Sigma(B)$, that is, there exist tame linear maps $X \overset{\iota}{\rightarrow} \Sigma(B) \overset{\pi}{\rightarrow} X$ with $\pi \circ \iota = \id_X$.
\end{defn}
\begin{remark}
The intuitive idea behind tame Fréchet spaces is that the growth of the sequence $k \mapsto \normdot_k$ is controlled by the exponential map. For concreteness, consider the space $C^\infty(\R)$ of smooth functions with seminorms as in \cref{ex::locallyConvexSpaces:SmoothFunctionOnOpenSubset}, that is
\begin{equation}
	\norm{\phi}_{n, k} \defeq \sup_{\substack{x \in [-n,n] \\ i \leq k}}{\abs*{\difff{\phi}{x}{i}(x)}} 	\qquad \text{for } n,k \in \N.
\end{equation} 
Tameness would imply the existence of constants $C_{n,k}$ such that $C_{n,k} \cdot \norm{\phi}_{n,k}$ is rapidly decreasing. There obviously exist functions with exponential growth which thereby serve as counterexamples for such a tame estimate. Hence, $C^\infty(\R)$ is not a tame Fréchet space. Nonetheless, the space of smooth functions on a \emph{compact} manifold is tame (see \cref{sec::sectionSpace}).
\end{remark}

\begin{theorem}[Nash-Moser inverse function theorem] \label{prop::locallyConvexSpace:NashMoserInverseTheorem} %
	Let $X$ and $Y$ be tame Fréchet spaces and $f: X \supseteq U \to Y$ a tame smooth map defined on an open subset $U \subseteq X$.
	Assume that the derivative $\dif f$ has a tame smooth family $\Psi^f$ of inverses, that is, $\Psi^f: U \times Y \to X$ is a tame smooth map and the family $\Psi^f_x: Y \to X$ is inverse to $(\dif f)_x$ for all $x \in U$. Then the map $f$ is locally bijective and the inverse is a tame smooth map.
\end{theorem}
The proof can be found in \parencite[section III.1]{Hamilton1982} or \parencite[pp. 562-574]{KrieglMichor1997}. A detailed repetition of the arguments goes beyond the scope of this thesis. Instead the discussion focuses on the differences between the Nash-Moser theorem and the classical inverse function theorem. In this way the reader is familiarized with this invaluable tool and aware of its specialities. 

\begin{remarks}
	\item As in the Banach setting, the Nash-Moser theorem is proved with the help of the Newton fixed point iteration procedure. But in Fréchet spaces, due to the loss of derivatives, the iteration step from $x_n$ in some local Banach space $B_k$ results in $x_{n+1}$ in some `lower' building block $B_{k-r}$. The tameness of the involved maps allows one to estimate how big this jump is. Moreover, the technical prerequisite of tameness of the Fréchet spaces ensures the existence of so called smoothing operators. These permit to compensate the loss of derivatives and to approximate $x_{n+1}$ by an element in the original Banach space $B_k$, which then is used for the next iteration step. 

	\item Based on the seminal work of Nash, Moser and Hamilton, further refinements of the inverse function theorem can be found in the literature. Mostly they aim at extending the realm of applicability to a larger class of Fréchet spaces, see for example \parencite{Poppenberg1996}. Recently, \textcite{Ekeland2012} followed a completely different proof strategy, namely he relies on his variational principle instead of the Newton iteration procedure. The resulting theorem encompasses the Nash-Moser theorem as a special case but requires weaker assumptions on the involved Fréchet spaces. However, the current version only yields local surjectivity. 

	Furthermore, \parencite{Gloeckner2006} is a feasible resource for implicit function theorems for maps from arbitrary locally convex spaces to Banach spaces.  
\end{remarks}

The existence of a tame smooth family of inverses basically requires the differential to be invertible in a whole neighbourhood, and not only in a single point as for Banach spaces. This condition can be traced back to the fact that invertible operators are no longer open in the space of all linear operators, see \cref{ex::diffGroup:expMap}. Nonetheless, the presence of finite-dimensional spaces simplifies the situation:
\begin{proposition}[{\parencite[pp. 47ff]{Subramaniam1984}}] \label{prop::tameFrechet:extendFiniteDimSplitting}
	Let $A$ and $X$ be tame Fréchet spaces, $Y \subseteq X$ a closed subspace and $E$ a finite dimensional vector space. Moreover, let $\Phi: A \times (E \times Y) \to X$ be a tame smooth family of linear maps which decomposes into
	\begin{equation}
		\Phi_a (e, y) = \varphi_a(e) + y,
	\end{equation}
	where $\varphi: A \times E \to X$ is a tame smooth family of injective, linear maps. If $\Phi_{a_0}: E \times Y \to X$ yields a linear and topological isomorphism $X = \img \varphi_{a_0} \oplus Y$ for some $a_0 \in A$, then there exists an open neighbourhood $U \subseteq A$ of $a_0$ such that $\Phi_a$ is bijective for all $a \in U$ and its inverses form a tame smooth map $U \times X \to E \times Y$.  
\end{proposition}
The proof follows roughly the presentation in \parencite[pp. 47ff]{Subramaniam1984} and is carried out in several lemmas.  
\newpage
\begin{lemma} \label{prop::lcs:FiniteDimSubspaceTopComplemented} %
	Let $X$ be a Fréchet space and $E \subseteq X$ a finite-dimensional subspace. There exists a continuous, linear map $\epsilon: X \to E$ which splits the inclusion,
	\begin{equation}
		\begin{tikzcd}
			0 \arrow{r} & E \arrow{r}{\iota} & X. \arrow[bend  left=50]{l}{\epsilon}
		\end{tikzcd}
	\end{equation}
	In particular, $E$ is topologically complemented by $\ker \epsilon$. Conversely, every direct sum decomposition $X = E \oplus Y$ defines a natural projection
	\begin{equation}
		X  = E \oplus Y \to E, \qquad e + y \mapsto e,
	\end{equation}
	which possesses all the just described properties.
\end{lemma}
\begin{proof}
	Let $\set{e^i}$ be a basis in $E$ and let $\epsilon_i: E \to \R$ denote the dual basis. By the Hahn-Banach theorem \parencite[Theorem 3.6]{Rudin1973}, one can perceive $\epsilon_i$ as continuous functionals on the whole space. Now, 
	\begin{equation}
		\epsilon: X \to E, \qquad x \mapsto \epsilon_i(x) e^i
	\end{equation}
	is the desired projection.
\end{proof}
\begin{lemma}
	Let $A$ and $X$ be Fréchet spaces, $Y \subseteq X$ a closed subspace and $E$ a finite-dimensional vector space. Let $\varphi: A \times E \to X$ be a jointly continuous family of injective, linear maps. If $X = \img \varphi_{a_0} \oplus Y$ for some $a_0 \in A$, then there exists an open neighbourhood $U \subseteq A$ of $a_0$ such that $X = \img \varphi_{a} \oplus Y$ for every $a \in U$.
\end{lemma}
\newpage
\begin{proof}
	As a corollary to \cref{prop::lcs:FiniteDimSubspaceTopComplemented}, the space $X$ splits into $\img \varphi_{a_0} \oplus Y $ as well as into $ \img \varphi_{a} \oplus Y$ if and only if the projections $P_{a_0}$ and $P_{a}$ defined by\footnote{Note that $P_a$ is well-defined and continuous since $\varphi_{a}$ is a homeomorphism onto its image by the open mapping theorem.}
	\begin{equation}
		\begin{tikzcd}
			0 \arrow{r} & E \arrow{r}{\varphi_{a}} & \img \varphi_{a} \arrow{r}{\iota_{a}} & X \arrow[bend  right=50, swap]{ll}{P_{a} \defeq \varphi_{a}^{-1} \circ \epsilon_{a}} \arrow[bend  left=50]{l}{\epsilon_a}
		\end{tikzcd}
	\end{equation}
	have the same kernels, $\ker P_{a_0} = Y = \ker P_{a}$. Equivalently, the transition function
	\begin{equation}
		\begin{tikzcd}
			\tau_{a_0, a}: E \arrow{r}{\varphi_{a_0}} & \img \varphi_{a_0} \arrow{r}{\iota_{a_0}} 
				& X \arrow{r}{P_a} & E
		\end{tikzcd}
	\end{equation}
	is invertible. Since $GL(E)$ is open in $L(E,E)$ for the finite-dimensional space $E$ and the initial value at $a_0$ is the identity, $\tau_{a_0, a_0} = \id_E$, there exists an open subset around $a_0$ such that $\tau_{a_0, a}$ is bijective. 
\end{proof}
The previous lemma obviously implies that $\Phi(e, y) = \varphi(e) + y$ is a homeomorphism. Hence, one is left with showing tame smoothness of the inverse in order to complete the proof of \cref{prop::tameFrechet:extendFiniteDimSplitting}. By \parencite[Theorem 3.1.1]{Hamilton1982}, it actually suffices that the inverse is a tame map. But this directly follows by the remarks after \cref{defn::FrechetTame:TameNonLinearMap}, since $X \to E \oplus Y$ is the projection on a finite-dimensional space composed with maps between finite-dimensional spaces (in the $E$-factor). %

\chapter{Locally convex differential geometry} \label{cha::locallyconvexmanifolds:locally_convex_diffgeo}

This chapter explores to which degree the classical differential geometric methods extend to the infinite-dimensional realm. The presentation proceeds in the following order. After introducing manifolds and fibre bundles as locally trivial objects, the various possibilities of defining tangent vectors are analysed. The kinematic viewpoint turns out to be preferred, because every equivalence class of curves induces a derivation on the space of smooth functions but the inverse implication does not hold true in general. Vector fields are defined as smooth sections of the tangent bundle. However, the notion of a differential form does not generalize straightforwardly. As it will be shown, the cotangent bundle does not carry a natural smooth structure. Therefore, smoothness of differential forms is not automatically inherited from a smooth bundle structure and is introduced more directly. In later chapters product manifolds will play a major role and hence they are extensively discussed. Some basic facts of infinite-dimensional product manifolds can be found in \parencite{KrieglMichor1997} and, for matters beyond that, the finite-dimensional treatise \parencite{GreubHalperinEtAl1972} served as a guideline.

Armed with these basic notions, more advanced topics are addressed. In particular, the Nash-Moser theorem is carried over to manifolds and then applied in the study of submanifolds. The various concepts of submanifolds become even more elaborate in infinite dimensions, since the conditions on submanifold charts decouple from the notion of an immersion and embedding in the absence of a general inverse function theorem. At least, the category of tame Fréchet spaces provides some remedy. Subsequently, some ideas and methods of Riemannian geometry are studied. The main aim of this investigation consists in proving that every Riemannian manifold is topologically metrizable. Finally, the important example of function spaces serves as an illustration of the preceding concepts.  

The reader is assumed to be familiar with the theory of finite-dimensional manifolds. Therefore, the emphasis of the presentation does not lie in the motivation of the differential-geometric constructions. Instead, it will be analysed in which way and under which additional assumptions they can be lifted to this new setting. 

\section{Manifolds and fibre bundles}
The notion of manifolds and fibre bundles as locally trivial objects directly carries over from the finite-dimensional context, since the definition essentially relies only on the chain rule. Nevertheless, for the sake of completeness and to fix the notation, the statements are explicitly spelled out.

\begin{defn}
	A smooth \emphDef{(locally convex) manifold $M$} is a topological Hausdorff space that is locally homeomorphic to a locally convex space $E$. In more detail, every point $m \in M$ is contained in a homeomorphic chart $\kappa: M \supseteq U \to \kappa(U) \subseteq E$ from an open subset $U \subseteq M$ to an open subset $\kappa(U) \subseteq E$. Furthermore, all chart transitions $\kappa^{-1} \circ \tilde \kappa$ are required to be smooth as maps between open subsets of locally convex spaces. The collection of charts compatible in this sense constitutes an atlas on $M$.  

	A map $f: M \to N$ between manifolds is said to be \emphDef{smooth} if in charts $\rho$ and $\kappa$ the local representation $\rho \circ f \circ \kappa^{-1}$ is smooth. 
\end{defn} 
The constructions are independent from the particular chart due to the chain rule. In contrast to the finite-dimensional case, locally convex manifolds are not obliged to satisfy the second axiom of countability since this rigid requirement would confine the class of admissible manifolds too much. In particular, paracompactness and differentiable partitions of unity subordinate to a given open covering are not available in general. Thus attention has to be paid in using localisation arguments. 
\begin{remark}[A word about notation]
	In the following all chart representations are denoted simply by the same letter as the original object and without explicitly mentioning of the involved charts. For example, a map $f: M \to N$ between manifolds is locally a map $f: E \supseteq U \to F$ between locally convex spaces. The same remark later also applies to vector fields and differential forms.
\end{remark}

\begin{example}[Product manifolds] \label{ex::lcm:productManifold}
	Let $M$ and $N$ be manifolds with atlases $(U_\alpha, \kappa_\alpha: U_\alpha \to E)$ and $(V_\beta, \rho_\beta: V_\beta \to F)$, respectively. The direct product $M \times N$ is endowed with the product topology and the charts $\kappa_\alpha \times \rho_\beta$ model it on the product $E \times F$. It is easily seen that the chart transitions are indeed smooth. The manifold so constructed is called the product manifold of $M$ and $N$. The natural projections $\pr_M: M \times N \to M$ and $\pr_N: M \times N \to N$ are smooth by construction. 
\end{example}

\begin{defn}
	Let $P$ and $M$ be smooth manifolds, $\pi: P \to M$ a smooth submersion\footnote{A smooth submersion is a smooth map with surjective derivative at every point.} and $F$ a topological space. The tuple $(P, M, \pi, F)$ is called a smooth \emphDef{(locally convex) fibre bundle} with total space $P$, base $M$, projection $\pi$ and typical fibre $F$ if it is locally trivial. The latter condition means that for every point $p \in P$ there exists an open neighbourhood $U \subseteq M$ of the base point $\pi(p)$ such that $\pi^{-1}(U)$ is diffeomorphic to the product $U \times F$ in such a way that $\pi$ is identified with the projection on the first factor. Consequently, the following diagram is commutative.
	\begin{equation}
		\begin{tikzcd}[column sep=small, ampersand replacement=\&]
			\pi^{-1}(U) \arrow{rr}{\tau} \arrow[swap]{rd}{\pi} 	\&		\& U \times F \arrow{ld}{\pr_1} \\
																\& U 	\&
		\end{tikzcd}
	\end{equation}
	A \emphDef{smooth section} of $\pi: P \to M$ is a smooth map $\phi: M \to P$ fulfilling $\pi \circ \phi = \id_M$.

	A smooth map $\Psi: P_1 \to P_2$ between bundles $P_1 \to M$ and $P_2 \to M$ over the same base is called a \emphDef{vertical fibre bundle morphism} if it maps fibres onto themselves, or equivalently, induces the identity map on the base.

	A \emphDef{(locally convex) vector bundle} is a fibre bundle where a locally convex space constitutes the typical fibre and the local trivialisation respects the linear structure, that is, the induced map $\pi^{-1}(m) \to F$ is a vector space isomorphism for every base point $m \in M$. Morphisms in the category of vector bundles over a fixed base manifold are the vertical fibre bundle morphisms that induce linear maps between the fibres.  
\end{defn}

\section{Vector fields}
In finite dimensions a tangent vector at a point can be thought of in three equivalent ways. Kinematically, a tangent vector is understood as an equivalence class of curves. Upon representation in charts it yields an object that transforms in a special way under chart transitions. This transformation property can be used as an equivalent definition. Finally, an equivalence class $\equivClass{\gamma}$ of curves maps to a derivation on smooth functions,
\begin{equation}\label{eq::lcm:derivationFromCurve}%
	X^\gamma: C^\infty(M) \to \R, \qquad f \mapsto \diffAt{}{t}{t=0}(f \circ \gamma).
\end{equation}
Conversely, every derivation can be represented in such a way. 

The equivalence between the kinematic viewpoint and the definition by its behaviour under chart transitions only needs the existence of charts and thus carries over to the infinite-dimensional context. Equation \eqref{eq::lcm:derivationFromCurve} above still makes sense and results in a derivation, but the converse direction requires, among other things, the reflexivity of the model space\footnote{Every element $Y$ of the bidual $E''$ yields a derivation at $m \in M$ which maps $f \in C^\infty(M)$ to $\dualpair{Y}{(\dif f)_m}$. But $E \subseteq E''$, with equality only in the case of a reflexive space.} and hence does not hold in general. \Textcite{KrieglMichor1997} extensively discuss the usual constructions (exterior derivative, Lie derivative and so on) for differential forms in the context of kinematic tangent vectors as well as algebraic ones and come to the conclusion that only the former concept admits a full-fledged theory. Since the kinematic approach is also natural from a physicist's point of view, it is adopted for the subsequent discussion. 

\newpage
\begin{defn}
	A \emphDef{tangent vector} $X_m$ at a point $m \in M$ is an equivalence class of curves $X_m \equiv \equivClass{\gamma}_m$, where the two curves $\gamma$ and $\eta$ are called equivalent if $\gamma(0) = m = \eta(0)$ and
	\begin{equation} \label{eq::difGeo:tangentVectorEquivalentChartRep}
		\diffAt{}{t}{t=0}(\kappa \circ \gamma)(t) = \diffAt{}{t}{t=0}(\kappa \circ \eta)(t)
	\end{equation}
	holds in every chart $\kappa$ around $m$. The tangent space at $m$, denoted by $T_mM$, is defined as the set of all tangent vectors at $m$ with its natural vector space structure. 
\end{defn}
Independence of \eqref{eq::difGeo:tangentVectorEquivalentChartRep} from the chosen chart again traces back to the chain rule.
The union $TM \defeq \bigsqcup_{m \in M} T_mM$ is now endowed with a vector bundle structure.
Let $\pi_{TM}: TM \to M$ denote the canonical projection. $TM$ trivializes over an open chart domain $(U \subset M, \kappa)$ to $\kappa(U) \times E$ via the map
\begin{equation}
\tau_U: TM_{\restriction_U} \to \kappa(U) \times E, \qquad \equivClass{\gamma}_m \mapsto \left(\kappa(m), \diffAt{}{t}{t=0} (\kappa \circ \gamma)(t)\right),
\end{equation}
where $E$ is the model space of $M$. Changing the chart on $M$ to $(V, \rho)$ induces the transition map
\begin{equation} \label{eq::lcm:TangentBundleChartsTransitionMap}
	\tau_V \circ (\tau_U)^{-1} (x, X) = (\rho \circ \kappa^{-1} (x), (\rho \circ \kappa^{-1})_x (X)),
\end{equation}
which is smooth. Hence the maps $\tau_U$ can be used to define a differentiable structure on~$TM$. Furthermore, the local representative of $\pi_{TM}$ with respect to these bundle charts corresponds to the projection on the first factor and hence $\pi_{TM}$ is a smooth submersion. 
The transition maps \eqref{eq::lcm:TangentBundleChartsTransitionMap} are linear in the second component and therefore $TM$ is a vector bundle.

\begin{defn}
	The smooth vector bundle $TM$ is called the \emphDef{tangent bundle} of $M$. Its smooth sections are called $\emphDef{vector fields}$. The set of smooth vector fields is denoted by $\vectorf{M}$.

	Let $f:M \to N$ be a smooth map. The derivative $f': TM \to TN$ of $f$ is defined pointwise by
	\begin{equation}
	 	f'_m (\equivClass{\gamma}_m) \defeq \equivClass{f \circ \gamma}_{f(m)} \in T_{f(m)}N, \qquad \text{ for } \equivClass\gamma \in T_mM. 
	 \end{equation} 
\end{defn}
For a given vector field $X$ on $M$, the projection of the trivialisation on the $E$-factor yields a smooth map $\kappa(U) \to E$, which is called the local chart representation of $X$. Therefore, a vector field $X$ is locally identified with a smooth map $X: U \to E$. Note that in the absence of a solution theory for differential equations in locally convex spaces vector fields do not need to posses a (unique) local flow.

The next short-term objective is to equip the set of all vector fields with a Lie algebra structure. First, note that $\vectorf{M}$ becomes a real vector space with the natural, pointwise defined addition and scalar multiplication. Moreover, every vector field $X \in \vectorf{M}$ defines a derivation $C^\infty(M) \to C^\infty(M)$, which sends a smooth function $f: M \to \R$ to the composition $f' \circ X$. In fact, $\R$-linearity is inherited from the derivative operator $f \mapsto f'$ and the Leibniz rule is easily verified in a local chart. Since not every derivation can be seen as a tangent vector, the commutator of two vector fields defined in the usual way needs the additional proof to be a vector field again. This is accomplished by the following lemma.

\begin{lemma}[Lie algebra structure of $\vectorf{M}$]
	The commutator of two derivations $X$ and $Y$ defined by
	\begin{equation}
	 	\commutator{X}{Y}(f) = X(Y(f)) - Y(X(f)) \qquad \text{for } f \in C^\infty(M),
	\end{equation} 
	leaves the space of vector fields invariant and thus descends to a map $\commutator{\cdot}{\cdot}: \vectorf{M} \times \vectorf{M} \to \vectorf{M}$. The defined bracket is bilinear, anti-symmetric and fulfils the Jacobi identity, thus $(\vectorf{M}, \commutator{\cdot}{\cdot})$ is a Lie algebra.
\end{lemma}
\begin{proof}
	Once it is proven that the space of vector fields is closed under the above commutator, the claimed Lie algebra property follows from a straightforward calculation analogous to the finite-dimensional case. Since closedness under $\commutator{\cdot}{\cdot}$ is a local problem, it is enough to consider the chart representation $X, Y: U \to E$ of the two vector fields $X$ and $Y$, respectively. Setting
	\begin{equation}
		Z(x) \defeq (\dif Y)_x (X(x)) - (\dif X)_x \, (Y(x)) \qquad \text{for } x \in U
	\end{equation}
	defines the local representative of a vector field $Z$ on $M$. Now due to the identity $X_x(f) = (\dif f)_x X(x)$ the evaluation 
	\begin{equation}
		X(Y(f))_x = \dif (Y(f))_x \, X(x) = (\dif^{2} f)_x (Y(x), X(x)) + (\dif f)_x \left((\dif Y)_x \, X(x) \right)
	\end{equation}
	holds. This identity together with the symmetry of the higher derivatives \parencite[Proposition 1.3.16]{Gloeckner2005} shows that $Z$ fulfils the defining equation of the commutator. %
\end{proof}

\newcommand{\inod}{\mathrel{\rotatebox[origin=c]{90}{$\in$}}}

\begin{example}[Vector fields on product manifolds] \label{ex::locallyconvexmanifolds:tangentProductManifold} %
	Let $M$ and $N$ be manifolds. Recall from \cref{ex::lcm:productManifold} that the product $M \times N$ carries a natural manifold structure such that the projections $\pr_M$ and $\pr_N$ are smooth. The tangent bundle $T(M \times N)$ is identified with $TM \oplus TN$ by means of the vertical bundle isomorphism $\pr'_M \times \pr'_N$. The inverse is given pointwise by 
	\begin{equation}
		T_mM \oplus T_nN \to T_{m,n}(M \times N), \qquad (X_m, Y_n) \mapsto \iota'_{M, n} X_m + \iota'_{N, m} Y_n
	\end{equation}
	where the smooth injection $\iota_{M,n}: M \to M \times N$ is defined by $m \mapsto (m, n)$ and analogously $\iota_{N, m}$. In the following, these inclusions are suppressed in the notation so that $T_mM$ is viewed as a subspace of $T_{m,n} (M \times N)$ under the above identification. Likewise, the canonical Lie algebra homomorphisms $\iota'_M: \vectorf{M} \to \vectorf{M \times N}$ and $\iota'_N: \vectorf{N} \to \vectorf{M \times N}$  
	are implicitly understood when viewing vector fields on one factor as vector fields on the product manifold. Observe that vector fields on $M$ commute with vector fields on $N$. Conversely, every vector field $Z$ on $M \times N$ decomposes in the sum of vector fields $Z^M$ and $Z^N$ such that $Z^M_{m,n} \in T_m M$ and $Z^N_{m,n} \in T_n N$. While $X \in \vectorf{M}$, viewed as a vector field on $M \times N$, has the chart representation $X: M \supseteq U \to E \times \set{0}$, the local expression for $Z^M$ is a map $Z^M: M \times N \supseteq U \times V  \to E \times \set{0}$. Here $E$ denotes the modelling vector space of $M$.

	The commutator of vector fields can be modified in order to reflect the present product structure by defining for $Z_1, Z_2 \in \vectorf{M \times N}$ the brackets
	\begin{align}
	  	\prescript{M}{}{\pcommutator{Z_1}{Z_2}} &= \commutator{Z_1^M}{Z_2} + \commutator{Z_1}{Z_2^M} - \commutator{Z_1}{Z_2}^M \\
	  	\prescript{N}{}{\pcommutator{Z_1}{Z_2}} &= \commutator{Z_1^N}{Z_2} + \commutator{Z_1}{Z_2^N} - \commutator{Z_1}{Z_2}^N.
	\end{align}  
	They will be useful later for defining `partial' exterior derivatives. The following properties are readily observed (only the $M$-version is stated, but also apply to $N$):
	\begin{align}
		\prescript{M}{}{\pcommutator{\cdot}{\cdot}} & \text{ is $\R$-linear and antisymmetric}\\
		\prescript{M}{}{\pcommutator{f Z_1}{Z_2}} &= f \cdot \prescript{M}{}{\pcommutator{Z_1}{Z_2}} - Z_2^M(f) Z_1 \quad \text{for }f \in C^\infty(M \times N)\\
		\prescript{M}{}{\pcommutator{Z_1^M}{Z_2^M}} &= \commutator{Z_1^M}{Z_2^M}\\
		\prescript{M}{}{\pcommutator{Z_1^N}{Z_2^N}} &= 0\\
		\prescript{M}{}{\pcommutator{Z_1^N}{Z_2}} &= \commutator{Z_1^N}{Z_2^M}^N\\
		\prescript{M}{}{\pcommutator{Y}{Z_2}} &= 0\quad \text{for } Y \in \vectorf{N}\\
		\commutator{Z_1}{Z_2} &= \prescript{M}{}{\pcommutator{Z_1}{Z_2}} + \prescript{N}{}{\pcommutator{Z_1}{Z_2}}
	\end{align}  
	In particular, they show that $\prescript{M}{}{\pcommutator{\cdot}{\cdot}}$ does not only have a $M$-component. Consequently, $\prescript{M}{}{\pcommutator{\cdot}{\cdot}}$ should not be mistaken for $\commutator{\cdot}{\cdot}^M$.
\end{example}

\section{Differential forms}
In the previous section, vector fields were defined as smooth sections of the tangent bundle and therefore it is only natural to try to construct differential forms in an analogous way, namely as smooth sections of the dual bundle. Unfortunately, the cotangent bundle cannot be equipped with a suitable smooth structure. Hence, one has to switch to a more direct definition of differential forms. Nevertheless, all known features such as the wedge product, the exterior derivative and the contraction with vector fields can be carried over to the infinite-dimensional context. The construction is guided by the general principle to use algebraic relationships (for example Cartan's formula) in lieu of local descriptions or dynamical properties which rely on the flow of vector fields. 

The cotangent bundle $T^*M$ of a manifold $M$ is by definition the fibrewise dual to the tangent bundle and hence it should carry an atlas dual to the one of the tangent bundle. Let $(U, \kappa)$ be a chart on $M$ with model space $E$ and define $T^*U \defeq \bigsqcup_{m \in M} (T_mM)'$. For the time being $(T_mM)'$ does not carry a topology. The associated chart $\kappa^*: T^*U \to \kappa(U) \times E'$ is given in the second factor by $\alpha_m \mapsto (\kappa^{-1})^*\alpha_m \defeq \alpha_m \circ (\kappa')^{-1}_{\kappa(m)}$. A diffeomorphism $\psi: E \supseteq U \to V \subseteq E$ induces a chart transition $E' \ni \alpha \mapsto (\psi^*\alpha)_x = \alpha \circ \psi'_x \in E'$ in the second component. 
The following example shows that in general the map $U \times E' \to E', (x, \alpha) \mapsto \alpha \circ \psi'_x$ is not smooth. 
In particular, there exists no vector space topology on $E'$ such that $E'$ could serve as the model space of the cotangent bundle. 
Certainly, Banach spaces represent the borderline case since in this setting the above map is smooth for the norm topology.
\begin{counterexample}[No smooth structure on the cotangent bundle, {\parencite[Remark 1.3.9]{Neeb2006}}]
	Let $E$ be a locally convex space which is not normable. Consider the translation $\psi: E \to E, x \mapsto x+ \beta(x)x$ by a non-zero $\beta \in E'$. Since its derivative is given by $\psi'_x(h) = h + \beta(x) h + \beta(h) x$ (and analogously for higher derivatives), $\psi$ is smooth. Furthermore, the inverse map
	\begin{equation}
		\psi^{-1}(y) \defeq \frac{y}{\sfrac{1}{2} + \sqrt{\beta(y) + \sfrac{1}{4}}}	
	\end{equation}
	is defined on $\setc{y \in E}{\beta(y) > \sfrac{- 1}{4}}$, implying that the restriction of $\psi$ to some open neighbourhood of $0$ is a diffeomorphism. On the other hand, the map $U \times E' \to E', (x, \alpha) \mapsto \alpha \circ \psi'_x$ evaluates to 
	\begin{equation}
		\alpha \circ \psi'_x = (1 + \beta(x)) \alpha + \alpha(x) \beta.
	\end{equation}
	This expression cannot be continuous since the evaluation map $E' \times E \to \R$ is discontinuous in $0$ for every vector space topology on $E'$ for non-normable $E$, see \parencite{Maissen1963}\footnote{This inconvenience was one of the motivating ideas for the `convenient differential calculus' of \parencite{KrieglMichor1997}. There the map $U \times E' \to E', (x, \alpha) \mapsto \alpha \circ \psi'_x$ is still discontinuous, but now smooth. Thus the standard construction of differential forms as sections of the cotangent bundle runs through without modifications.}.
\end{counterexample}
Summarizing, the cotangent bundle is only a set-theoretic vector bundle without a differentiable structure. Hence, smoothness of differential forms has to be specified explicitly.

\begin{defn}
	Let $M$ be a manifold and $V$ a locally convex space. Denote by $\Lambda^k (M, V) \defeq \bigsqcup_{m \in M} \Lambda^k ((T_mM)', V)$ the set-theoretic bundle of $V$-valued exterior $k$-forms, where $\Lambda^k (E', V)$ is the space of $k$-linear, antisymmetric, $V$-valued maps on the locally convex space $E$. 
	A set-theoretic section $\alpha$ of $\Lambda^k (M, V)$ is called a \emphDef{$V$-valued differential form} if for every chart $\kappa: M \supseteq U \rightleftarrows U \subseteq E$ the induced chart representation
	\begin{equation}
		\alpha: U \times E^k \to V, \qquad (x, X^{(1)}, \dotsc ,X^{(k)}) \mapsto \alpha_{\kappa^{-1}(x)} \left((\kappa^{-1})'_x X^{(1)}, \dotsc ,(\kappa^{-1})'_x X^{(k)}\right)
	\end{equation}
	is a smooth map. The set of all $V$-valued $k$-forms on $M$ is denoted by $\diffform{k}{M}{V}$. %
\end{defn}

The rest of this section is concerned with the question of how the well-known operations on differential forms carry over to the infinite-dimensional setting. First of all a word of warning is in order: 
\begin{remark}
	Let $\alpha$ be a $k$-form and $X^{(1)}, \dotsc, X^{(k)}$ vector fields. The composition $\alpha \circ (X^{(1)}, \dotsc, X^{(k)})$ is a smooth function $M \to \R$. Thereby, every differential form can be seen as an alternating, $C^\infty(M)$-multilinear map $\alpha: \vectorf{M} \times \dotsb \times \vectorf{M} \to C^\infty(M)$. In finite dimensions this property characterizes differential forms completely and is useful for defining operations on them. However, the proof of the converse direction needs a smooth partition of unity subordinate to a given open cover in order to extend tangent vectors to smooth vector fields on the entire manifold and hence is not applicable in the present case\footnote{Due to the absence of smooth partitions of unity, a tangent vector at $m \in M$ can only be extended to a smooth vector field on some open set around $m$.}. Summarizing, new differential forms cannot be constructed purely from their action on vector fields.
\end{remark}

In the following $\alpha$ and $\beta$ are differential forms and all occurring $X^{(i)}_m$ and $Y^{(j)}_n$ are tangent vectors at $m \in M$ and $n \in N$, respectively.
\begin{enumerate}
	\item 
		The wedge product of $\alpha \in \diffform{k}{M}{V_1}$ and $\beta \in \diffform{r}{M}{V_2}$ with respect to a jointly continuous, bilinear map $B: V_1 \times V_2 \to V_3$ is the $k+r$-form $\alpha \wedge \beta \in \diffform{k+r}{M}{V_3}$ defined pointwise by
		\begin{multline}
			(\alpha \wedge \beta)_m (X^{(1)}_m, \dotsc, X^{(k+r)}_m) \defeq \frac{1}{k!r!} \sum_{\sigma \in \operatorname{perm}} \sgn(\sigma) \cdot \\
			B\left(\alpha_m(X^{\sigma(1)}_m, \dotsc, X^{\sigma(k)}),\beta_m(X^{\sigma(k+1)}_m, \dotsc, X^{\sigma(k+r)})\right).
		\end{multline}
	\item 
		Let $X$ be a vector field on $M$ and $\alpha$ a $k$-form. The equation
		\begin{equation}
			(X \contr \alpha)_m (X^{(1)}_m, \dotsc ,X^{(k-1)}_m) \defeq \alpha_m (X_m, X^{(1)}_m, \dotsc ,X^{(k-1)}_m)
		\end{equation}
		defines a $k-1$-form $X \contr \alpha$, which is called the interior product. %
	\item 
		The pullback $\phi^* \alpha \in \diffform{k}{N}{V}$ of the differential form $\alpha \in \diffform{k}{M}{V}$ to the manifold $N$ by the smooth map $\phi: N \to M$ is given pointwise by 
		\begin{equation}
			(\phi^* \alpha)_n (Y^{(1)}_n, \dotsc, Y^{(k)}_n) \defeq \alpha_{\phi(n)} (\phi'_n Y^{(1)}_n, \dotsc, \phi'_n Y^{(k)}_n)
		\end{equation}
	\item
		Normally, the exterior derivative is defined by its local behaviour but these arguments do not carry over to non-paracompact manifolds. Thus, a more algebraic definition is given. Let $\alpha \in \diffform{k}{M}{V}$ be a $k$-form and $X^{(0)}_m, \dotsc, X^{(k)}_m$ be tangent vectors at $m \in M$ which are extended locally to smooth vector fields $X^{(0)}, \dotsc, X^{(k)}$ around $m$. The exterior derivative $\dif \alpha$ of $\alpha$ is the $k+1$-form determined by 
		\begin{equation}\label{eq::lcm:exteriorDerivativeDef}\begin{split}
			(\dif \alpha)_m(X^{(0)}_m, \dotsc, X^{(k)}_m) \defeq \sum_{i=0}^k (-1)^i X^{(i)}_m \left(\alpha(X^{(0)}, \dotsc, \widehat {(i)}, \dotsc, X^{(k)}) \right)\\
																 +\sum_{i<j}^k (-1)^{i+j} \alpha_m (\commutator{X^{(i)}}{X^{(j)}}, X^{(0)}, \dotsc, \widehat {(i)}, \dotsc, \widehat {(j)}, \dotsc,  X^{(k)}),
		\end{split}\end{equation}
		where the hat stands for omission of the corresponding entry. \Autoref{lemma::locallyconvexmanifolds:diffformWelldefined} below verifies that this equation indeed defines a $k+1$-form which does not depend on the choice of the extensions of the tangent vectors to vector fields. 
	\item
		Let $\alpha \in \diffform{k}{M}{V}$ be a $k$-form and $X^{(1)}_m, \dotsc, X^{(k)}_m$ be tangent vectors extended locally as above. The Lie derivative $\difLie_X \alpha \in \diffform{k}{M}{V}$ with respect to $X \in \vectorf{M}$ is given by
		\begin{multline}
			(\difLie_X \alpha)_m (X^{(1)}_m, \dotsc, X^{(k)}_m) \defeq X_m \left(\alpha(X^{(1)}, \dotsc X^{(k)}) \right) \\
						+ \sum_{i=1}^k (-1)^i \alpha_m (\commutator{X}{X^{(i)}}, X^{(1)}, \dotsc, \widehat {(i)}, \dotsc,  X^{(k)}).
		\end{multline}
		Well-definedness is easily verified once Cartan's formula is established, see \eqref{eq::locallyConvexManifolds:exteriorDiffProperties}. If the local flow of the vector field $X$ exists, then this definition is clearly equivalent to the usual one, $\difLie_X \alpha = \diffAt{}{t}{t=0} (\flow^X_{-t})^* \alpha$.
\end{enumerate}

The standard rules for calculations carry over with in parts lengthy but straightforward verifications, which are left to the reader.
\begin{subequations} \label{eq::locallyConvexManifolds:exteriorDiffProperties}
\begin{align+} 
  	\phi^*(\alpha \wedge \beta) &= \phi^* \alpha \wedge \phi^* \beta \\ 
  	\dif^{2} \alpha &= 0 \\ 
	\dif(\phi^* \alpha) &= \phi^* (\dif \alpha) \\  
	\difLie_X \alpha &= \dif (X \contr \alpha) + X \contr \dif \alpha
\end{align+}
\end{subequations}
Finally, well-definedness of the exterior derivative is demonstrated.
\begin{lemma}\label{lemma::locallyconvexmanifolds:diffformWelldefined} %
	Let $\alpha \in \diffform{k}{M}{V}$ be a $k$-form and $X^{(0)}_m, \dotsc, X^{(k)}_m$ be tangent vectors at $m \in M$ that are extended locally to smooth vector fields $X^{(0)}, \dotsc, X^{(k)}$. The exterior derivative $\dif \alpha$ defined by \cref{eq::lcm:exteriorDerivativeDef} is a $k+1$-form. The construction is independent from the specific choice of the extension of the tangent vectors to local vector fields.
\end{lemma}
\begin{proof}
	The proof consists in demonstrating the following local representation for the exterior derivative
	\begin{equation}\label{eq::locallyconvexmanifolds:localExteriorDifferential}
		(\dif \alpha)_m(X^{(0)}_m, \dotsc, X^{(k)}_m) = \sum_{i=0}^k (-1)^i (\dif_1 \alpha)_m(X^{(i)}_m; X^{(0)}_m, \dotsc, \widehat {(i)}, \dotsc, X^{(k)}_m).
	\end{equation}
	This expression immediately implies anti-symmetry, smoothness and independence from the choice of the extension. A note about notation: The differential form $\alpha$ is locally represented as a function $U \times E^k \to \R$ and thus $\dif_1 \alpha: (U \times E^k) \times E \to \R$. Hence, in \cref{eq::locallyconvexmanifolds:localExteriorDifferential} the last $k$ tangent vectors should actually be written as subscripts since they represent the `point' where the partial derivative is evaluated. 
	To establish the identity \eqref{eq::locallyconvexmanifolds:localExteriorDifferential} consider first the term $X^{(i)}_m \left(\alpha(X^{(0)}, \dotsc, \widehat {(i)}, \dotsc, X^{(k)}) \right)$, which by the product rule evaluates to
	\begin{align}
		\dif \left(\alpha(X^{(0)}, \dotsc, \widehat {(i)}, \dotsc, X^{(k)}) \right)_m & (X^{(i)}) 
			= (\dif_1 \alpha)_m(X^{(i)}_m; X^{(0)}_m, \dotsc, \widehat {(i)}, \dotsc, X^{(k)}_m)\\ 
			&+ \sum_{j=0}^{i-1} \alpha(X^{(0)}_m, \dotsc, (\dif X^{(j)})_m (X^{(i)}_m), \dotsc, \widehat {(i)}, \dotsc, X^{(k)}_m)\\
			&+ \sum_{j=i+1}^{k} \alpha(X^{(0)}_m, \dotsc, \widehat {(i)}, \dotsc, (\dif X^{(j)})_m (X^{(i)}_m), \dotsc, X^{(k)}_m).
	\end{align}
	Here, linearity of $\alpha$ in the $E$-components is utilised to identify the partial differential in these directions with $\alpha$ itself. Recalling the local expression $(\dif X^{(j)})_m (X^{(i)}_m) - (\dif X^{(j)})_m (X^{(i)}_m)$ for the commutator of $X^{(j)}$ and $X^{(i)}$, one easily sees that the last two sums together cancel the double sum in \cref{eq::lcm:exteriorDerivativeDef}. The remaining part is exactly \cref{eq::locallyconvexmanifolds:localExteriorDifferential}.
\end{proof}

\begin{example}[Differential forms on product manifolds] \label{ex::locallyconvexmanifolds:diffFormsProductManifold}
	Let $M$ and $N$ be manifolds modelled on locally convex spaces $E$ and $F$, respectively. Recall from \autoref{ex::locallyconvexmanifolds:tangentProductManifold} the identification $T(M \times N) = TM \oplus TN$ and the accompanying sum representation $Z = Z^{(M)} + Z^{(N)}$ of vector fields on $M \times N$. For the sake of simplicity, the dual decomposition of differential forms is first illustrated for $1$-forms. 

	For $\mu \in \diffform{1}{M \times N}$ set $\mu^{1,0}_{m,n}(Z_{m,n}) \defeq \mu(Z^{(M)}_{m,n})$ and $\mu^{0,1}_{m,n}(Z_{m,n}) \defeq \mu(Z^{(N)}_{m,n})$. This defines $1$-forms $\mu^{1,0}$ and $\mu^{0,1}$ on $M \times N$ because their local expression is the restriction of $\mu: U \times (E \times F) \to \R$ to $U \times (E \times \set{0})$ and $U \times (\set{0} \times F)$, respectively. Clearly, $\mu^{1,0}$ vanishes identically if restricted to $TN$ and this property is taken to be the defining condition for the subspace $\diffformBi{1}{0}{M \times N} \subseteq \diffform{1}{M\times N}$ (and analogously for $\diffformBi{0}{1}{M \times N}$). Now linearity implies $\mu = \mu^{1,0} + \mu^{0,1}$ and thus $\diffform{1}{M\times N} = \diffformBi{1}{0}{M \times N} \oplus \diffformBi{0}{1}{M \times N}$.

	The preceding remarks generalize straightforwardly to higher order differential forms. Let $\diffformBi{p}{q}{M \times N}$ denote the subspace of $(p+q)$-forms that vanishes whenever not evaluated on $p$ tangent vectors of $M$ and $q$ ones on $N$. Then $k$-forms decompose into the bigrading %
	\begin{equation}
		\diffform{k}{M\times N} = \bigoplus_{p+q=k} \diffformBi{p}{q}{M \times N}.
	\end{equation}
	The $(p,q)$-component of a form $\alpha \in \diffform{k}{M\times N}$ will be denoted by $\alpha^{p,q}$.	Note that the $(p,q)$-forms are sections of the bundle $\Lambda^p(M) \otimes \Lambda^q(N)$ in finite dimensions. As topological considerations of the space $\Lambda^k(M \times N)$ were avoided so far and, moreover, the tensor product of infinite-dimensional topological vector spaces is an involved theory in its own, this approach is not further investigated. However, it should generalize at least on the level of vector spaces without topology. 

	In analogy to vector fields, differential forms on $M$ are identified with their image under the pullback map $\pi^*_M: \diffform{p}{M} \to \diffformBi{p}{0}{M \times N}$. In fact, these are exactly those $(p,0)$-forms that do not depend on the base point $n \in N$. Conversely, if $\mu$ is a $(p,q)$-form on $M \times N$ and $X^{(1)}, \dotsc, X^{(p)}$ are vector fields on $M$, then the contraction $\mu(X^{(1)}, \dotsc, X^{(p)})$ is a family of $q$-forms on $N$, parameterized by the points on $M$. Similar statements hold if the roles of $M$ and $N$ are interchanged.

	The concept of partial derivatives is generalized to the case of product manifolds by defining the following \emphDef{partial exterior derivative}
	\begin{multline}
			(\dif^M \mu)_{m, n}(Z^{(0)}_{m, n}, \dotsc, Z^{(p+q)}_{m, n}) \defeq \sum_{i=0}^{p+q} (-1)^i Z^{(i), M}_{m,n} \left(\mu(Z^{(0)}, \dotsc, \widehat {(i)}, \dotsc, Z^{(p+q)}) \right)\\
																 +\sum_{i<j}^{p+q} (-1)^{i+j} \mu_{m,n} \left(\prescript{M}{}{\pcommutator{Z^{(i)}}{Z^{(j)}}}, Z^{(0)}, \dotsc, \widehat{(i)}, \dotsc, \widehat{(j)}, \dotsc,  Z^{(p+q)} \right),
	\end{multline}
	for $\mu \in \diffformBi{p}{q}{M \times N}$ and tangent vectors $Z_{m, n}^{(0)}, \dotsc, Z_{m, n}^{(p+q)}$ locally extended to vector fields on $M \times N$. An analogous equation holds for $\dif^N$. 

	Let $\Dif$ denote the total exterior derivative on $M \times N$. The following properties are easily verified (except for the cochain property $(d^M)^2 = 0$ which follows from a tedious but straightforward calculation).
	\begin{lemma}
		\begin{thmenumerate}*
			\item $\dif^M: \diffformBi{p}{q}{M \times N} \to \diffformBi{p+1}{q}{M \times N}$ \\ $\dif^N: \diffformBi{p}{q}{M \times N} \to \diffformBi{p}{q+1}{M \times N}$
			\item $\Dif = \dif^M + \dif^N$.
			\item $(\dif^M)^2 = 0 = (\dif^N)^2$ and so $\dif^M \dif^N + \dif^N \dif^M = 0$. \label{lem::lcm:AntiCommutativPartialExteriorDifferential}
			\item Let $X^{(1)}, \dotsc, X^{(p)}$ and $Y^{(0)}, \dotsc, Y^{(q)}$ be vector fields on $M$ and $N$, respectively. The partial derivative $\dif^N$ corresponds to the usual de Rham differential $\dif$ on $N$ in the following sense:
			\begin{equation}\begin{split}
				(\dif^N &\mu)(X^{(1)}, \dotsc, X^{(p)}, Y^{(0)}, \dotsc, Y^{(q)}) \\
						&= \sum_{v=0}^{q} (-1)^{p+v} Y^{(v)}\left(\mu(X^{(1)}, \dotsc, X^{(p)}, Y^{(0)}, \dotsc, \widehat{(v)}, \dotsc, Y^{(q)}) \right)\\
						&\quad+\sum_{v<w}^{q} (-1)^{p + v + w} \mu (X^{(1)}, \dotsc, X^{(p)}, \commutator{Y^{(v)}}{Y^{(w)}}, \dotsc, \widehat {(v)}, \dotsc, \widehat {(w)}, \dotsc,  Y^{(q)})\\
						&= (-1)^p \dif \left(\mu(X^{(1)}, \dotsc, X^{(p)}) \right),
			\end{split}\end{equation}
			where $\mu(X^{(1)}, \dotsc, X^{(p)})$ is viewed as a parameterized form on $N$ as described above. \qedhere %
		\end{thmenumerate}
	\end{lemma}
	In the same way the partial Lie derivatives are defined by
	\begin{multline}
		(\difLie^M_Z \mu)_{m,n} (Z^{(1)}_{m,n}, \dotsc, Z^{(p+q)}_{m,n}) \defeq Z^M_{m,n} \left(\mu(Z^{(1)}, \dotsc Z^{(p+q)}) \right) \\
					+ \sum_{i=1}^{p+q} (-1)^i \mu_{m,n} \left(\prescript{M}{}{\pcommutator{Z}{Z^{(i)}}}, Z^{(1)}, \dotsc, \widehat {(i)}, \dotsc,  Z^{(p+q)} \right)
	\end{multline}
	and analogously for $\difLie^N$. The following properties are immediate from the definition and the properties of the adapted commutator $\prescript{M}{}{\pcommutator{\cdot}{\cdot}}$: 
	\begin{enumerate}
		\item $\difLie^M_Z \mu = \dif^M(Z \contr \mu) + Z \contr \dif^M \mu$
		\item $\difLie_Z \mu = \difLie^M_Z \mu + \difLie^N_Z \mu$
		\item $(\difLie^M_{Z^N} \mu)(Z^{(1)}, \dotsc, Z^{(p+q)}) = \sum_{i=1}^{p+q} (-1)^i \mu (\commutator{Z^N}{Z^{(i), M}}^N, Z^{(1)}, \dotsc, \widehat {(i)}, \dotsc,  Z^{(p+q)})$\\ In particular, $\difLie^M_{Z^N} \mu$ has bigrading $(p+1, q-1)$ and in general vanishes only if evaluated for a vector field $Z = Y$ on $N$.
		\item The total Lie derivative decomposes into 
			\begin{equation}
				\difLie_Z \mu = \difLie^N_{Z^M} \mu + (\difLie^M_{Z^M} \mu + \difLie^N_{Z^N} \mu) + \difLie^M_{Z^N} \mu,
			\end{equation}
			where the summands have bigrading $(p-1, q+1)$, $(p, q)$ and $(p+1, q-1)$, respectively. \qedhere %
	\end{enumerate}	
\end{example}

\section{Nash-Moser inverse function theorem revisited}
The Nash-Moser theorem discussed in \autoref{sec:locallyconvexspaces:nashMoserInverseTheorem} generalizes to the global context once one restricts to an appropriate subcategory of Fréchet manifolds.
\begin{defn}
	A Fréchet manifold is called \emphDef{tame} if the modelling space is a tame Fréchet space and the chart transitions are tame maps. A map between tame Fréchet manifolds is said to be \emphDef{tame smooth} if its chart representation is a tame smooth map.
\end{defn}

\begin{theorem}[Nash-Moser inverse function theorem] %
	Let $f: M \supseteq U \to N$ be a tame smooth map between tame Fréchet manifolds defined on an open subset $U \subseteq M$. Assume that the derivative $f'_m: T_mM \to T_{f(m)}N$ is bijective for all points $m \in U$ and that the map $\Psi^f: f^*TN \to TM$ defined by the inverses is tame smooth. Then, the map $f$ is locally bijective and the inverse is a tame smooth map, that is, $f$ is a tame local diffeomorphism. 
\end{theorem}  
\begin{proof}
	The assertion is of local nature and thus follows from the inverse function theorem \ref{prop::locallyConvexSpace:NashMoserInverseTheorem} applied to the chart representation of $f$. 
\end{proof}

\section{Submanifolds}
A submanifold is a subset of a manifold that carries itself a suitable differentiable structure. In finite dimensions there exist mainly two ways to define submanifolds: First, one can specify how the submanifold-charts are inherited from the ambient manifold. Second, the inclusion map and the variety of different properties one can demand on it are brought into the focus of consideration. Both characterisations are intimately linked due to the inverse function theorem (or, more precisely, the constant rank theorem). The situation becomes more diverse when one moves to the infinite-dimensional setting, because formerly equivalent definitions now separate into distinct ones. 

This section primarily focuses on what is often denoted as `embedded' submanifolds. These are subsets mapped under charts to a closed subspace of the modelling vector space. Since in finite dimensions every closed subspace is complemented, one equivalently gets an embedded submanifold by setting some chart-coordinates to zero. Clearly, this equivalence does not carry over.

\begin{defn} %
  	A \emphDef{submanifold} is a subset $S \subseteq M$ of a manifold $M$ such that there exist a covering of $S$ by charts $\kappa_\alpha: U_\alpha \to E$ of $M$ and a closed subspace $G \subseteq E$ fulfilling
  	\begin{equation} \label{eq::submanifold:definitionChart}
  		\kappa_\alpha(U_\alpha \cap S) = \kappa_\alpha(U_\alpha) \cap G.
  	\end{equation}
  	Hence the restrictions of the $\kappa_\alpha$'s to $U_\alpha \cap S$, shortly denoted by $\kappa_{\alpha {\restriction_S}}$, constitute an atlas of $S$ endowing it with a manifold structure\footnote{To see that the collection $\kappa_{\alpha {\restriction_{S}}}$ indeed defines a manifold structure, endow the subset $S \subseteq M$ with the relative topology. Thus it is a Hausdorff topological space and $U_\alpha \cap S$ are open sets by definition. Moreover, as restrictions of homeomorphisms, the maps $\kappa_{\alpha {\restriction_S}}$ are homeomorphisms themselves and by \eqref{eq::submanifold:definitionChart} their image is open. The chart transitions, being restrictions of smooth maps, are smooth themselves.}. Charts of this type are called \emphDef{submanifold charts}. If in addition $G$ is topologically complemented\footnote{A topological complement of a closed subspace $G$ in a topological vector space $E$ is a closed subspace $F$ such that $E = G \oplus F$. The existence of algebraic complements is always assured, but the natural isomorphism $G \oplus F \ni (x,y) \mapsto x+y \in E$ need not be a homeomorphism.}, then one speaks of a \emphDef{splitting submanifold}. In this case
  	\begin{equation}
  		\kappa_\alpha(U_\alpha \cap S) = \kappa_\alpha(U_\alpha) \cap (G \times \set{0}).
  	\end{equation}

  	Let $\pi: P \to M$ be a vector bundle. A (splitting) submanifold $Z \subseteq P$ is called a \emphDef{(splitting) subbundle} of $P$ if there exist a vector bundle $\rho: Q \to N$ and an injective morphism $\phi: Q \to P$ with image $\phi(Q) = Z$. 
\end{defn} 
As remarked above, the split condition is superfluous in finite dimensions (as well as for Hilbert manifolds). It is easy to see that the vector bundle  $\rho: Q \to N$ in the definition of a subbundle is uniquely determined up to a vector bundle isomorphism.  

For a submanifold $S \subseteq M$ the natural inclusion map $\iota: S \to M$ is smooth as one verifies by inspection in charts. In fact, let $\kappa: U \to E$ be an arbitrary chart of $M$ and as above denote by $\kappa_{\restriction_S}$ the induced chart on $S$. The chart representation $\kappa \circ \iota \circ \kappa^{-1}_{\restriction_S}: \kappa(U) \cap G \to \kappa(U)$ is smooth since it is the restriction of the linear, continuous inclusion $G \to E$ to an open subset. This argument also implies that the tangent map $\iota'_s$ is an injective, (closed) topological embedding for every point $s \in S$ (with complemented image if $S$ is a splitting submanifold). Furthermore, $\iota$ itself is an embedding since $S$ carries the subspace topology. 

In the finite-dimensional context, the converse is also well-known, that is: the image of an injective immersion $\iota$ is an embedded submanifold if $\iota$ is also a topological embedding. Since the proof of this fact relies on the constant rank theorem one cannot hope to generalize it to arbitrary locally convex manifolds. Nonetheless, the Nash-Moser theorem provides the right stage for a similar proposition.
The following definition captures the required notation.
\begin{defn} %
	Let $S$ and $M$ be tame Fréchet manifolds. A tame smooth map $\iota: S \to M$ is called a \emphDef{tame immersion} if 
	\begin{enumerate}
		\item the tangent map $\iota'_s: T_sS \to T_{\iota(s)}M$ is injective for all $s \in S$ and the image $T\iota \defeq \bigsqcup_{s \in S} \img (\iota'_s)$ is a subbundle of $\iota^* TM$,
		\item the bundle map of inverses $\Psi^\iota: T\iota \to TS$ of $\iota'$ is tame smooth. 
	\end{enumerate}
	If additionally $T\iota$ splits as a subbundle of $\iota^* TM$, then it is called a \emphDef{splitting tame immersion}.
\end{defn}
\begin{remark}
	The requirement of $T\iota$ being a subbundle seems to be redundant at first glance since in finite dimensions the injectivity of a vector bundle morphism already ensures the existence of a subbundle structure. This statement though relies on the extension of a linearly independent local frame to a local basis frame, which in turn goes back to the openness of the subset of invertible operators in the space of arbitrary linear operators. The latter result does not hold beyond the finite-dimensional regime, see \autoref{ex::diffGroup:expMap}.	%
\end{remark}
The following proposition expands a theorem of \parencite[p. 542]{AbbatiCirelliEtAl1989} slightly, where it is stated additionally without an explicit proof.  
\begin{proposition} \label{prop::submanifold:tameImmersionImpliesSplitSubmanifold}
	Let $S$ and $M$ be tame Fréchet manifolds. If $\iota: S \to M$ is a splitting tame, injective immersion, then for every $s \in S$ there exist an open neighbourhood $V \subseteq S$ of $s$, a chart $\kappa: U \to E$ of $M$ and a topological decomposition $E = G \oplus F$ such that:
	\begin{enumerate}
		\item The image $\iota(V)$ is contained in $U$.
		\item $\kappa(\iota(V))$ is open in $G$.
		\item The induced chart $\kappa \circ \iota_{\restriction_V}: V \to G$ is compatible with the atlas of $S$.
	\end{enumerate}
	Thus $\iota_{\restriction V}$ is a tame diffeomorphism onto a splitting submanifold of $M$.	If $\iota$ is in addition a topological embedding, then the image $\iota(S)$ is a splitting submanifold of $M$.
\end{proposition}
\begin{proof}
	Fix a point $s \in S$. Since $T\iota$ is a splitting subbundle of $\iota^* TM$, there exist charts $\rho: S \supseteq V \to G$ around $s$ and $\kappa: M \supseteq U \to E$ containing $\iota(V)$ such that, relative to these charts, $\iota'$ is pointwise the inclusion of $G$ into $E$ (up to some neglected isomorphism). Moreover, $G$ is topologically complemented in $E$, say $E = G \oplus F$. Hence the chart representation of $\iota$ maps the open subset $W \equiv \rho(V) \subseteq G$ to the product $G \times F$. Denote by $g: W \to G$ and $f: W \to F$ the composition with the projection on the first and second factor, respectively. Furthermore, define the linear extension $I: W \times F \to G \times F$ of $\iota$ by $(x, y) \mapsto (g(x), y + f(x))$. Its derivative then clearly evaluates to
	\begin{equation}
		(\dif I)_{x,y} = \begin{pmatrix} (\dif g)_x & 0 \\ (\dif f)_x & id_F \end{pmatrix}.
	\end{equation}  
	The bundle map of inverses $\Psi^\iota: T\iota \to TS$ translates in the local setting to a tame family $\Psi^g: W \times G \to G$ inverse to $\dif g$. As can be easily verified, the linear operators
	\begin{equation}
		\Psi^I_{x,y} \defeq \begin{pmatrix} \Psi^g_x & 0 \\ - (\dif f)_x \circ \Psi^g_x & \id_F \end{pmatrix}
	\end{equation}
	give rise to a tame family $\Psi^I: (W \times F) \times (G \times F) \to G \times F$ of inverses of $(\dif I)$. By the inverse function theorem, $I$ is a local diffeomorphism around $(\rho(s), 0)$. Hence there exist open subsets $O_G \subseteq G$, $O_F \subseteq F$, $O_E \subseteq E$ such that $O_G \times O_F$ is diffeomorphic to $O_E$ (by $I$) and the following diagram commutes.
	\begin{equation}
		\begin{tikzcd}[column 2/.style={anchor=base west}]
			O_G \arrow{r}{\iota} \arrow[hook, end anchor=north west]{rd} &[1cm] |[anchor=base west]| \iota(O_G) \subseteq O_E\\
			 		{}										& |[anchor=base west]| O_G \times \set{0} \arrow[swap, end anchor={[xshift=4ex]south west}, start anchor={[xshift=4ex]north west}]{u}{I} \subseteq O_G \times O_F 
		\end{tikzcd}
	\end{equation}
	After shrinking $V$ in an appropriate way this argumentation shows that $\iota(V)$ is splitting submanifold of $M$, and thus the first claim of the proposition is proven. 

	Now assume that $\iota$ is also a topological embedding. Choose an open covering $V_\alpha$ of $S$ and, for every $\alpha$, a chart $\kappa_\alpha: U_\alpha \to E$ of $M$ with properties (i) - (iii) as ensured by the previous considerations. For every $\alpha$, there exists an open subset $O_\alpha \subseteq M$ with $\iota(V_\alpha) = O_\alpha \cap \iota(S)$ since $\iota(S) \subseteq M$ carries the relative topology. Now $\kappa_\alpha$ restricted to $U_\alpha \cap O_\alpha$ provides the desired submanifold chart.
\end{proof}
Note that the above proof relies on the topological sum decomposition of the model space and hence a similar result cannot be expected for non-splitting tame immersions. Submanifolds permit the following feasible statement about smoothness of maps. 
\begin{proposition}[Submanifolds are initial] \label{prop::lcm:SubmanifoldIsInitial}
	Let $S \subseteq M$ be a submanifold and let $\iota$ denote the inclusion $S \hookrightarrow M$. For every smooth map $f: N \to M$ with $f(N) \subseteq S$ the restriction in range $\iota^{-1} \circ f: N \to S$ is smooth. 
\end{proposition}
The proof proceeds as in finite dimensions \parencite[Lemma~2.3.23]{Gloeckner2005}. %
This chapter is summarized in the following hierarchy, which starts with the most specific concept and ends at the most general level:
\begin{equation}
	\text{splitting submanifold } \leq \text{ submanifold } \leq \text{ initial submanifold } \leq \text{ injective immersion}.
\end{equation}

\section{Riemannian geometry} \label{cha::lcm:gradedRiemannianGeometry}
This section discusses how the concepts of finite-dimensional Riemannian geometry carry over to Fréchet manifolds. Since the results are only essential for the slice theorem~\labelcref{prop::liegroup:sliceTheorem}, the reader might prefer to skip this chapter upon first reading and come back to it by the time Riemannian metrics are needed. The work of \textcite[pp. 51ff]{Subramaniam1984} serves as the primary source, but the current presentation goes slightly beyond it by also discussing the exponential map.

In the Banach space setting, the theory of Riemannian geometry splits into two branches. Strong Riemannian metrics induce a topology on the tangent spaces equivalent to the original one, while for weak Riemannian metrics the induced topology is coarser. 
The well-consolidated building of finite-dimensional Riemannian geometry remains largely unchanged for strong metrics but is confronted with serious problems if weak metrics are considered.
The reason for this unsatisfactory state lies in the impossibility of defining Christoffel symbols by the Koszul formula in the case of a weak metric. 
Therefore, one often takes refuge in the strong case. However, Fréchet manifolds cannot be strong Riemannian in the classical sense as the discussion about dual spaces shows (cf. \autoref{prop::locallyConvexSpaces:dualOfFrechetNotMetrizable}). The idea is now to circumvent some of these problems by using not only one metric but a whole collection of them which are `strong' in the sense that they together induce an equivalent topology on the tangent space.

\begin{defn}
	Let $M$ be a Fréchet manifold. A \emphDef{graded Riemannian metric} is a family of functions $g^k_m: T_mM \times T_mM \to \R$ with the following properties (for all $m \in M$): 
	\begin{enumerate}
		\item The functions $g^k_m: T_mM \times T_mM \to \R$ are semi-inner products, that is symmetric, bilinear and positive semi-definite functions\footnote{Hence, a semi-inner product resembles an inner product except for $g^k_m(X, X)$ can be zero even for non-zero $X$.}.  
		\item The induced seminorms $\norm{\cdot}^k_m = \sqrt{g^k_m(\cdot, \cdot)}$ constitute a directed fundamental system generating a topology on $T_mM$ equivalent to the original one\footnote{In the following, no powers of seminorms are required and thus the notation $\norm{\cdot}^k$ should not lead to confusions.}.
		\item $g^k$ varies smoothly with $m$ in the sense that for every chart $\kappa: M \supseteq U \to E$ of $M$ the induced chart representation $g^k: \kappa(U) \times E \times E \to \R$ is a smooth map (where the usual identifications in charts are implied). \qedhere
	\end{enumerate}	 
\end{defn}

The definition demands only a weak relationship of the collection $g^k$ at two nearby points. To see this, consider the local representation $g^k: U \times E \times E \to \R$. For all $x, y \in U$, the families of semi-inner products $g^k_x$ and $g^k_y$ induce an equivalent topology on $E$. Hence for a given $k$ and $K > 1$ there exists $n$ such that $\normdot^n_y \leq K \normdot^k_x$ (cf. \cref{prop::locallyConvexSpaces:seminormCharateristEquivalentTopology}). Note that the index $n$ might depend on $k$ as well as on the chosen points. Often more control is necessary.
\begin{defn} \label{defn::Riemann:LocallyEquivalentMetric} %
	A graded Riemannian metric is called \emphDef{locally equivalent} if for every $m \in M$ and $K>1$ there exists a chart $\kappa: U \to E$ around $m$ such that locally 
	\begin{equation}
  		\frac{1}{K} \normdot^k_y \leq \normdot^k_x \leq K \normdot^k_y
	\end{equation}
	holds for all $k$ and every $x,y \in U$. If there exists a chart around every point in which all $g^k$ are constant, then the metric is called \emphDef{locally flat}. %
\end{defn}

In the usual approach to the theory of Riemannian geometry one now proves existence and uniqueness of the Levi-Civita connection by using the Koszul formula. The whole arsenal of techniques that comes along with a connection allows one to speak of geodesics and finally to define the exponential map. Equipped with these tools it is straightforward to prove that the path-length metric induces an equivalent topology on the manifold. Considering the lack of an inverse function theorem and solution theory for differential equations the ordinary route does not lead to the desired result in Fréchet spaces. Instead the last statement is directly investigated. %

Analogous to the finite-dimensional case, one defines the \emphDef{length of a smooth curve} $\gamma:\R \supseteq [a, b] \to M$ with respect to the $k$-th component as
\begin{equation}
	l^k (\gamma) \defeq \int_a^b \norm{\dot \gamma(t)}^k_{\gamma(t)} \dif t.
\end{equation}
If two points $p$ and $q$ lie in the same connected component of $M$, then their distance $d^k(p,q)$ results from taking the infimum of $l^k$ over all $C^1$-curves connecting $p$ and $q$. Finally, define the length metric as
\begin{equation}
	d(p,q) \defeq \sum_{k=1}^\infty 2^{-k} \frac{d^k(p,q)}{1 + d^k(p,q)}.  	
\end{equation}    
This is indeed a metric as the following proposition demonstrates.
\begin{proposition}[{\parencite[p. 52]{Subramaniam1984}}] \label{prop::riemannianGeometry:locEquivalentMetricInducesLengthMetric}
	Let $M$ be a Fréchet manifold and $g^k$ a locally equivalent, graded Riemannian metric. The map $d(p,q)$ is a metric on each connected component of $M$ which is compatible with the original manifold topology.
\end{proposition}
\begin{proof}
	All properties of a metric are evident, except positivity. Fix some point $p \in M$ and consider an open subset $U$ around it. One can reduce the problem to the case where the other endpoint $q$ and the connecting curve $\gamma$ completely lie in $U$. To see this, assume $\img(\gamma) \nsubseteq U$. By continuity of $\gamma$, there exists $c \in [a,b]$ such that for times up to $c$ the curve takes values in $U$ and $\gamma(c) \neq \gamma(a)$. But $l^k(\gamma) \geq l^k(\gamma_{\restriction_{[a,c]}})$ implies that it is enough to show positivity for $\gamma_{\restriction_{[a,c]}}$.

	Thus, in particular, one can consider a chart $\kappa: M \supseteq U \rightleftarrows U \subseteq E$ around $p$. Under the usual identifications, $\gamma$ is a curve $\gamma: [a,b] \to U \subseteq E$ connecting $p$ with some different point $q$, and $g^k: U \times E \times E \to \R$. Due to the local equivalence property one can assume that
	\begin{equation}
  		\frac{1}{K} \normdot^k_p \leq \normdot^k_x \leq K \normdot^k_p
	\end{equation}
	holds for all $x \in U$. The length of $\gamma$ can be estimated from below by %
	\begin{equation}\label{eq::riemannianGeometry:lengthMetricEstimateLengthOfCurve} \begin{split} 
		l^k(\gamma) &= \int_a^b \norm{\dot \gamma(t)}^k_{\gamma(t)} \dif t \geq \frac{1}{K} \int_a^b \norm{\dot \gamma(t)}^k_p \dif t \\
					&\geq \frac{1}{K} \norm*{\int_a^b \dot \gamma(t) \dif t}^k_p = \frac{1}{K} \norm{\gamma(b) - \gamma(a)}^k_p.
	\end{split}\end{equation}
	Since fundamental systems of seminorms split points, there exists a $k$ such that $\norm{\gamma(b) - \gamma(a)}^k_p$ is non-zero and thus the corresponding $l^k$ does not vanish. Therefore $d(p,q) > 0$ whenever $p\neq q$ and $d$ is a metric.

	Finally, it is left to show that the topology of $d$ coincides with the manifold topology. Charts on $M$ are homeomorphisms and hence test the topology of $M$. In particular, $q_n \to p$ in $M$ if and only if $q_n \to p$ in some chart, if and only if $\norm{q_n - p}^k_p \to 0$ for all $k$. The later equivalence holds as $\normdot^k_p$ is a compatible fundamental system of seminorms. On the other hand, $q_n$ converges to $p$ with respect to $d$ if and only if it converges relative to all $d^k$. Thus equivalence of both topologies can be checked for each $k$ separately. 
	\begin{itemize}
		\item Assume $d^k(p, q_n) \to 0$. Then for some $\varepsilon > 0$ and sufficiently large $n$ one has $d^k(p, q_n) < \varepsilon$. The points $q_n$ eventually lie in some chart $U$ around $p$, since otherwise all connecting curves between $p$ and $q_n$ have lengths greater or equal to $\varepsilon$ (compare to the argumentation at the beginning of the proof). Equation \eqref{eq::riemannianGeometry:lengthMetricEstimateLengthOfCurve} implies $\norm{q_n - p}^k_p \leq K d^k(p, q_n)$, which converges to $0$.

		\item Conversely, assume $q_n \to p$ in $M$. It implies $\norm{q_n - p}^k_p \to 0$ in some chart. The path-length can be estimated by considering the linear curve $\sigma(t) = p + t (q_n - p)$. Then by the local equivalence of seminorms
		\begin{equation}
			d^k(p, q_n) \leq l^k(\sigma) = \int_0^1 \norm{\dot \sigma(t)}^k_{\sigma(t)} \dif t = \int_0^1 \norm{q_n - p}^k_{\sigma(t)} \dif t \leq K \norm{q_n - p}^k_p. 
		\end{equation}
	\end{itemize}
\end{proof}

\begin{defn}
	An \emphDef{$r$-exponential map} for a graded Riemannian metric $g^k$ is a smooth map $\exp: TM \subseteq U \to M$ defined on an open subset $U$ of the zero section in $TM$ such that for every $X \in T_mM$ the associated curve $[0,1] \ni t \mapsto \lambda_X(t) = \exp(t X)$ fulfils $\diffAt{}{t}{t=0} \lambda_X(t) = X$ and is the $r$-shortest conjunction between its endpoints, that is
	\begin{equation}
		l^r(\lambda_X) = d^r(m, \exp(X)). 
	\end{equation}
\end{defn}
Since differential equations in Fréchet spaces do not necessarily have a unique solution, there could exist none or many exponential maps conforming to the above definition. Furthermore, $\exp$ is in general not a local diffeomorphism and thus, in particular, the image of $\exp_m$ fails to be an open neighbourhood of $m$. It is subject to further investigations how the exponential maps corresponding to different values of $r$ are related to each other. Examples suggest that the existence of an $0$-exponential map is an indicator for the existence of all higher exponential maps. This phenomenon is closely related to the `regularity of geodesics' \parencite{EbinMarsden1970}, that is, geodesics relative to one Riemannian metric $g^r$ transform under suitable extra conditions to geodesics with respect to another $g^{r'}$. %
\section{Example: Spaces of mappings} \label{sec::sectionSpace}
From a mathematical point of view, function spaces are of particular interest, since they often provide (counter-) examples for deep, global analytic questions. Therefore, progress in the general theory is closely interwoven with concrete problems in function spaces. Moreover, maps between finite-dimensional manifolds constitute by definition the objects of interest in classical field theory.

This section is devoted to the manifold structure of the section space of a finite-dimensional fibre bundle. Mostly, the presentation is limited to compact base manifolds, but the general case will also be commented on. The manifold structure on the section space provides the basis to investigate important maps and their derivatives, for example the evaluation map or the pushforward. The work of \textcites[Example II.1.4]{Neeb2005}[Example 1.1.4 and 1.1.5]{Hamilton1982}[section 2.3]{Wang2012}[section 2.2]{BrunettiFredenhagenEtAl2012} constitute the literature underlying this section.

As a motivating example, first consider the space $C^\infty(U, \R^l)$ of smooth maps from an open subset $U \subseteq \R^r$ to $\R^l$. Slightly generalizing\cref{ex::locallyConvexSpaces:SmoothFunctionOnOpenSubset,ex::lcs:smoothFunctionOnOpenSubsetFrechet} the seminorms
\begin{equation} \label{eq::functionSpace:seminormSmoothFunctionOnOpenSubset}
	\norm{\phi}_{K, k} \defeq \sup_{\substack{x \in K \\ \abs{I} \leq k}}{\norm*{\difff{\phi}{x}{I}(x)}_{\R^l}} 	\qquad \text{for } \phi \in C^\infty(U, \R^l), k \in \N, K \subseteq U \text{ compact}
\end{equation}   
define a Fréchet structure on $C^\infty(U, \R^l)$. The resulting topology is called \emphDef{$C^\infty$ compact-open topology}. 

Now let $\pi: E \to M$ be a finite-dimensional vector bundle over a compact manifold $M$. The section space of $E$, denoted by $\secspace{E}$, carries a natural linear structure defined by fibrewise addition and scalar multiplication. To endow it with a Fréchet topology one proceeds in a similar way as in the above example but replaces the derivative in \eqref{eq::functionSpace:seminormSmoothFunctionOnOpenSubset} with an appropriate generalization. In particular, the following alternatives all achieve this and provide the same topology:
\begin{itemize}
 	\item Chart representation,
 	\item Covariant derivative,
 	\item Jet techniques.
\end{itemize} 

Going back to the chart representation is straightforward since the local model was already discussed above. In fact, the topology on $\secspace{E}$ is the initial topology with respect to the maps
\begin{equation}
	\secspace{E} \to \prod_{j} C^\infty(U_j, \R^l)
\end{equation}
which send a section to its local representative. Here $i$ indexes a countable and trivializing atlas $(U_j, \kappa_j)$ of $M$. Explicitly seminorms on $\secspace{E}$ are given by the same formula as in \eqref{eq::functionSpace:seminormSmoothFunctionOnOpenSubset} if all objects are replaced by their local expressions.
A more geometric approach uses higher covariant derivatives. Let $\nabla^E: \secspaceEx{E} \to \secspaceEx{T^*M \otimes E}$ and $\nabla^M: \secspaceEx{TM} \to \secspaceEx{T^*M \otimes TM}$ be torsion-free connections on $E$ and $TM$, respectively. They combine via the Leibniz identity to a torsion-free connection $\nabla: \secspaceEx{\otimes_k T^*M \otimes E} \to \secspaceEx{\otimes_{k+1} T^*M \otimes E}$. Exemplary for $k=1$, one has
\begin{equation}
	\nabla_X (\alpha \otimes \phi) = \nabla^{M, *}_X \alpha \otimes \phi + \alpha \otimes \nabla^E_X \phi,
\end{equation}
where $\nabla^{M,*}$ denotes the dual connection to $\nabla^M$. Now recursively define the higher covariant derivatives:
\begin{align}
	\nabla^0 \phi &= \phi\\
	\nabla^1_X \phi &= \nabla^E_X \phi\\
	\nabla^2_{X, Y} \phi &= \nabla_X (\nabla^1 \phi)(Y) = \nabla^1_X (\nabla^1_Y \phi) - \nabla^1_{\nabla^M_X Y} \phi\\[-1.5ex]
	&\;\; \vdots\\[-2ex] %
	\nabla^i_{X_1, \dotsc, X_i} \phi &= \nabla_{X_1} (\nabla^{i-1}_{X_2, \dotsc, X_i} \phi) - \sum_{r=2}^i \nabla^{i-1}_{X_2, \dotsc, \nabla^M_{X_1} X_r ,\dotsc, X_i} \phi.
\end{align} 
A Riemannian metric on $M$ and a fibre metric on $E$ are the last, necessary ingredients. Abbreviate
\begin{equation}
	\norm{\nabla^i \phi(m)} \defeq \sup_{\substack{X_r \in T_mM \\ \norm{X_r} \leq 1}} \norm{\nabla^i_{X_1, \dotsc, X_i} \phi (m)}
\end{equation}
and define seminorms
\begin{equation}
	\norm{\phi}_{K, k} \defeq \sup_{\substack{m \in K \\ i \leq k}}{\norm{\nabla^i \phi(m)}}
\end{equation}
for $\phi \in \secspace{E}$, some compact subset $K \subseteq M$ and $k \in \N$. This endows $\secspace{E}$ with a Fréchet topology which is independent from the previously used connections and metrics.

Finally, the just described topology is identified with the initial topology with respect to the infinite jet extension
\begin{equation}
	j^\infty: \secspace{E} \ni \phi \mapsto j^\infty\phi \in C(M, J^\infty E),
\end{equation} 
where the space of continuous maps $C(M, J^\infty E)$ carries the compact-open topology\footnote{The sets $\secspace{U}_{K} = \setc{g: M \to J^\infty E}{g(K) \subseteq U}$ for compact $K \subseteq M$ and open $U \subseteq J^\infty E$ constitute a sub-basis for the $C^0$-compact-open topology on $C(M, J^\infty E)$.} and the family $J^k  E$ endows $J^\infty  E$ with the inverse limit topology.  

\begin{remarks}
	\item The topology discussed above is tame. For a proof of this fact, the reader in referred to the literature \parencite[Corollary 1.3.9.]{Hamilton1982}.

	\item The Sobolev inequalities show that $\norm{\phi}_{K, k}$ could equivalently be replaced by the local Sobolev seminorms
	\begin{equation}
	 	\norm{\phi}^S_{K, k} \defeq \sup_{i \leq k} \int_K \norm{\nabla^i \phi(m)} \dif \mu(m),
	\end{equation} 
	where $K$ ranges over compact subsets with non-empty interior and $\dif \mu$ denotes the volume element induced from a Riemannian metric on $M$. 

	\item %
		Often a direct description of an open set around a given $\phi \in \secspace{E}$ is convenient. For the $C^0$ compact-open topology, a sub-basis for the neighbourhood around $\phi$ consists of sets of the form
		\begin{equation}
			\secspace{U}_{K} \defeq \setc{\psi \in \secspace{E}}{\psi(K) \subseteq U}
		\end{equation}
		for compact $K \subseteq M$ and open subsets $U \subseteq E$, which contain the image of $\phi$. Since the $C^\infty$ topology is finer, $\secspace{U}_{K}$ remains open. Because sets of this type occur often, the calligraphic notation also applies to other, arbitrary subsets of $E$. If no compact subset is given in the subscript, then the whole manifold $M$ is implied. That is, $\secspace{O} \equiv \setc{\psi \in \secspace{E}}{\psi(M) \subseteq O}$ for some subset $O \subseteq E$.

	\item The results easily extend to arbitrary fibre bundles $E$ (over compact $M$) by reducing the problem to a fibre-preserving tubular neighbourhood. This will render the section space of a fibre bundle an infinite-dimensional tame Fréchet manifold which locally resembles the section space of a vector bundle. 

	First, fix a section $\phi \in \secspace{E}$ and construct a tubular neighbourhood, that is, a diffeomorphism $\exp_\phi: V_\phi \to U_\phi$ from an open neighbourhood $V_\phi$ of the zero section in $\phi^* VE$ to an open neighbourhood $U_\phi$ of the image of $\phi$ in $E$. Now define
	\begin{align}
		\secspace{U}_\phi &\defeq \setc{\psi \in \secspace{E}}{\psi(M) \subseteq U_\phi}\\
		\secspace{V}_\phi &\defeq \setc{\varphi \in \secspaceEx{\phi^* VE}}{\varphi(M) \subseteq V_\phi}
	\end{align}
	and notice that $\secspace{V}_\phi$ is an open subset of the tame Fréchet space $\secspaceEx{\phi^* VE}$. The topology as well as the manifold structure results from the charts
	\begin{equation}
		\secspace{U}_\phi \to \secspace{V}_\phi, \quad \psi \mapsto \exp^{-1} \circ\, \psi.
	\end{equation}
	Since $\exp$ can be chosen to preserve fibres \parencite[Lemma 3.2.1]{Wang2012} the chart transitions originate from fibre-preserving maps and thus are tame smooth by \autoref{prop::sectionSpace:pushforwardTameSmooth} below.

	\item The generalisation to non-compact base manifolds $M$ is not straightforward. The problems prominently arise if the topology is derived from the jet extension $\secspace{E} \to C(M, J^\infty E)$. If $M$ is not compact, then the class of suitable and in some sense natural topologies of $C(M, J^\infty E)$ widens to a whole list and each of them comes with its own advantages and perils.

	The following sets will constitute the prototypes of neighbourhoods for the most important topologies on $C(M,N)$. Let $d$ be a metric on $M$, $S \subseteq M$ a subset and $\epsilon$ a positive-valued, continuous function on $M$. For $f \in C(M,N)$ define the set
	\begin{equation}
		\secspace{N}^d(f, S, \epsilon) = \setc{g \in C(M,N)}{d(f(x), g(x)) \leq \epsilon(x) \text{ for all } x \in S},
	\end{equation}
	which contains all functions $g$ that on $S$ are near to $f$. The most significant topologies on $C(M,N)$ differ by the specific choice of parameters.
	\begin{itemize}
		\item Whitney topology: The neighbourhood system comprises $\secspace{N}^d(f, M, \epsilon)$.
		\item Uniform topology: The neighbourhood system comprises $\secspace{N}^d(f, M, \epsilon)$, where $\epsilon$ is a constant function.   
		\item Compact-open topology: The neighbourhood system comprises $\secspace{N}^d(f, K, \epsilon)$, where $\epsilon$ is a constant function and $K \subseteq M$ a compact subset. 
	\end{itemize}
	In applications, one discovers that the Whitney topology is often too strong and the compact-open topology too coarse. The uniform topology has the disadvantage to depend on the chosen metric $d$. They all coincide for compact $M$.%

	Moreover, one cannot hope to end in the category of tame Fréchet spaces in general and thus interesting applications requiring the inverse function theorem are precluded. Even such a common Fréchet space as $C^\infty(\R, \R)$ fails to be tame.

	In their totality these problems and inconveniences make it necessary to retreat to the compact case. \qedhere
\end{remarks}

The rest of this section is concerned with maps between section spaces which will play a major role in later chapters. The \emphDef{evaluation map} 
\begin{equation}
	\ev: \secspace{E} \times M \to M, \quad (\phi, m) \mapsto \phi(m)
\end{equation}
distinguishes function spaces from other infinite-dimensional spaces. Its smoothness and derivative are captured in the following proposition.

\begin{proposition}[{\parencite[Proposition I.2]{NeebWagemann2007}}] \label{prop::sectionSpace:smoothEval}
	Let $E \to M$ be a finite-dimensional fibre bundle over a compact manifold $M$ and denote its section space by $\secspace{E}$. The evaluation map $\ev: \secspace{E} \times M \to M$ is smooth and its derivative is given by
	\begin{equation}
		(\ev)'_{\phi, m} (\Xi, X) = (\ev_m)'_\phi \, \Xi + \phi'_m X, \qquad \text{for } \Xi \in T_\phi \secspace{E}, X \in T_mM,
	\end{equation}
	where $\ev_m: \secspace{E} \to M$ is defined in the obvious way. %
\end{proposition}
\begin{proof}
	Since the statement is local in nature, it is enough to consider the case of a vector bundle $E$.%

	By compactness of $M$, local uniform convergence and uniform convergence on compacta coincide. Thus, to address the question of continuity, one has to show that $\ev: C^\infty(U, V) \times U \to V$ is a continuous map, where $U \subseteq M$ is open and $V$ a finite-dimensional vector space. For this purpose, it is sufficient to consider only the $C^0$-compact-open topology on $C^\infty(U, V)$ since the $C^\infty$-topology is finer. Let $O \subseteq V$ be an open subset. Openness of the inverse image $\ev^{-1}(O)$ will be shown by constructing an open neighbourhood around every point $(f, x) \in \ev^{-1}(O)$. As $f$ is continuous, $f^{-1}(O)$ is an open neighbourhood around $x$. $M$ is locally compact and thus there exists an open neighbourhood $W$ of $x$ such that its closure $K \equiv \bar{W}$ is compact and lies in $f^{-1}(O)$. Hence $\secspace{O}_K = \setc{g: M \to N}{g(K) \subseteq O}$ is open and $\secspace{O}_K \times W$ is an open neighbourhood of $(f, x)$ contained in $\ev^{-1}(O)$.

	Once the tangent space $T_\phi \secspace{E}$ is identified with $\secspace{E}$ itself, the proposed formula for the derivative follows from the chain rule. In fact, one has
	\begin{equation}
		(\ev)': \secspace{E} \times \secspace{E} \times TM \to TM, \quad (\phi, \Xi, X_m) \mapsto (\ev_m)'_\phi \, \Xi + \phi'_m X_m,
	\end{equation}
	which is continuous by the same argumentation as above if the second term $\phi' X_m$ is interpreted as the evaluation of $\phi': TM \to TE$ on $X_m$. By induction $\ev$ is a smooth map.
\end{proof}

\begin{corollary} \label{prop::sectionSpace:tangentSpace}
	Let $E \to M$ be a finite-dimensional fibre bundle over a compact manifold $M$. The assignment of $\Xi$ to the map $m \mapsto (\ev_m)'_\phi \, \Xi$ identifies $T_\phi \secspace{E}$ with $\secspaceEx{\phi^* VE}$. Under this isomorphism the derivative of the evaluation map simplifies to 
	\begin{equation}
		(\ev)'_{\phi, m} (\Xi, X) = \Xi(m) + \phi'_m X, \qquad \text{for } \Xi \in \secspaceEx{\phi^* VE}, X \in T_mM. 
	\end{equation}
\end{corollary}

\begin{proposition}[{\parencite[Theorem 2.2.6]{Hamilton1982}}] \label{prop::sectionSpace:pushforwardTameSmooth} %
	Let $E \to M$ and $F \to M$ be finite-dimensional fibre bundles over a compact manifold $M$ and $U \subseteq E$ an open subset. Denote the section spaces by $\secspace{E}$ and $\secspace{F}$, respectively. A smooth fibre-preserving map $f: E \supseteq U \to F$ yields a $0$-tame smooth map 
	\begin{equation}
		\secmap{f}_*: \secspace{E} \supseteq \secspace{U} \to \secspace{F}, \quad \phi \mapsto f \circ \phi,
	\end{equation}
	defined on the open subset $\secspace{U} = \setc{\phi \in \secspace{E}}{\phi(M) \subseteq U}$. Denote the fibre tangent map of $f$ by $Vf: VE \supseteq VU \to VF$. The derivative of $\secmap{f}_*$ at $\phi$ equals the pushforward of $Vf$ composed with $\phi$, that is
	\begin{equation}
		(\secmap{f}_*)'_\phi = \secmap{Vf}_* \circ \phi:\; \secspaceEx{\phi^* VE} \to \secspaceEx{(f \circ \phi)^* VF}, \qquad \Xi \mapsto f'_{\phi(\cdot)} \Xi. 
	\end{equation}
\end{proposition}
\begin{proof}
	A similar discussion as in the beginning of the proof of \autoref{prop::sectionSpace:smoothEval} reduces the whole problem to the case of a vector bundle. %

	For continuity, it is enough to consider the local representation in trivializing charts on $M$. Thus one identifies $f$ and $\secmap{f}_*$ with maps
	\begin{equation} \label{eq::sectionSpace:pushforwardLocalExpression} \begin{split}
		&f: V \times O \to V \times \R^r, \quad (x, v) \mapsto (x, \tilde{f}(x, v)),\\
		&\secmap{f}_*: C^\infty(V, \R^l) \supseteq \secspace{O} \to C^\infty(V, \R^r), \quad \phi \mapsto \tilde{f} \circ (\id_V \times \phi),
	\end{split}\end{equation}
	where $V \subseteq \R^n$ and $O \subseteq \R^l$ are open subsets (and $l$ and $r$ denote the fibre dimensions of $E$ and $F$, respectively). The induced map $\tilde{f}: V \times O \to \R^r$ is smooth. Since the $C^\infty$ topology is initial with respect to the infinite jet extension and the following diagram commutes,
	\begin{equation}
		\begin{tikzcd}
				C^\infty(V, \R^l) \arrow{r}{\secmap{f}_*} \arrow{d}{j^k} 	& 	C^\infty(V, \R^r) \arrow{d}{j^k}\\
				C(V, J^k (V, \R^l)) \arrow{r}{j^k f}					&	C(V, J^k (V, \R^r)) 
		\end{tikzcd}
	\end{equation}
	$\secmap{f}_*$ is continuous if and only if the $k$-jet prolongation of $f$ is continuous for all $k$. This argument reduces the problem to the continuous case $\secmap{f}_*: C(V, \R^l) \supseteq \secspace{O} \to C(V, \R^r)$. There the claim is a direct consequence of the decomposition of $\secmap{f}_*$ in $\phi \mapsto \id_V \times \phi$ and $\psi \mapsto \tilde{f} \circ \psi$, since both maps are continuous as apparent from the definition of the compact-open topology.  

	Consider now the asserted identity $(\secmap{f}_*)' = \secmap{Vf}_*$. The proof consists in showing that the difference quotient\footnote{ $(\secmap{f}_*)^{[1]}$ is well-defined since $\phi(M)$ and $\Xi(M)$ are compact and thus for small $t$ the image of $\phi + t \Xi$ lies in $U$.}
	\begin{equation}
		(\secmap{f}_*)^{[1]} (\phi, \Xi, t) \defeq \frac{\secmap{f}_* (\phi + t \Xi) - \secmap{f}_* (\phi)}{t}
	\end{equation}
	equals $(\secmap{Vf}_* \circ \phi)(\Xi)$ for $t=0$, since the identity $(\secmap{f}_*)'_\phi \Xi = (\secmap{f}_*)^{[1]} (\phi, \Xi, 0)$ holds by \autoref{prop::locallyConvexSpace:charaterisationDifferentialByDifferentialQuotient}. Here the tangent vector $\Xi$ is seen as an element of either $\secspace{E}$ or $\secspaceEx{\phi^* VE}$, as appropriate. It is enough to check $(\secmap{f}_*)^{[1]} (\phi, \Xi, 0) = (\secmap{Vf}_* \circ \phi)(\Xi)$ pointwise since both sides are elements of $\secspaceEx{(f \circ \phi)^* VF}$. After evaluation at a point, the claim follows directly from the same \autoref{prop::locallyConvexSpace:charaterisationDifferentialByDifferentialQuotient} but applied to the map $f$.
	
	Hence the derivative is again a pushforward and therefore is continuous due to the above argumentation. By induction, $\secmap{f}_*$ is smooth. Finally, the pushforward is tame smooth if it is a tame map by the same reasoning. To see the latter, the tame estimate 
	\begin{equation}
		\norm{\secmap{f}_* (\psi)}_{K, k} \leq C (1 + \norm{\psi}_{K, k})\quad \text{for all } k\in \N 
	\end{equation}
	will be shown on a neighbourhood of a fixed section $\phi$ and for some compact subset $K \subseteq M$. Let $C \subseteq E$ be a compact neighbourhood of the image $\phi(M)$ and let $\psi \in \secspace{C}$. By compactness of $M$, the bundle $E$ and $F$ trivialize over a finite open cover $V_s$ of $M$. On every $V_s$, the local expression \eqref{eq::sectionSpace:pushforwardLocalExpression} yields
	\begin{equation}
		\left(\diff{}{x}\right)^k (\secmap{f}_* \psi)(x) = \sum_{j + \abs{I} = k} \partial_1^j \partial_2^{\abs{I}} \tilde{f} (x, \psi(x))\, \difff{\psi}{x}{I_1} \dotsm \difff{\psi}{x}{I_n}. 
	\end{equation}
	Due to compactness of $C$, the map $\tilde{f}(x, \psi(x))$ and its derivatives are bounded for all $\psi \in \secspace{C}$. Hence 
	\begin{equation}
		\norm{\secmap{f}_* (\psi)}_{K, k} \leq \norm*{\difff{}{x}{k} (\secmap{f}_* \psi)}_{K, 0} \leq D \sum_{\abs{I} \leq k} \norm{\psi}_{K, I_1} \dotsm \norm{\psi}_{K, I_n}.
	\end{equation}
	The interpolation formula $\norm{\psi}_i^{k} \leq D \norm{\psi}_k^{i} \norm{\psi}_0^{k-i}$ for some $D > 0$\footnote{Contrary to good behaviour, all occurring constants are denoted by the same letter $D$ although they need not coincide.} from \parencite[p. 143f]{Hamilton1982} is needed to derive the estimate $\norm{\psi}_i \leq D \norm{\psi}_k^{i/k}$. Note that the superscripts are actually powers of seminorms and are not to be confused with the notation $\norm{\cdot}^k$ of \autoref{cha::lcm:gradedRiemannianGeometry}. Combining these inequalities leads, for $k > 0$, to
	\begin{equation}
	 	\norm{\secmap{f}_* (\psi)}_{K, k} \leq D \sum_{\abs{I} \leq k} \norm{\psi}_{K, I_1} \dotsm \norm{\psi}_{K, I_n} \leq D \norm{\psi}_{K, k}.
	\end{equation} 
	Since $\norm{\secmap{f}_* (\psi)}_{K, 0}$ equals some constant, this formula yields the desired estimate on $V_s$ with a constant $D_s$, which possible depends on $s$. The covering $V_s$ is finite and hence $D_s$ can be replaced by a global constant.  
\end{proof}
In particular, the previous proposition yields a functor from the category of fibre bundles over a fixed manifold with smooth fibre-preserving maps as morphisms to the category of tame manifolds and tame smooth maps.

\chapter{Locally convex Lie groups} \label{sec::lieGroups}
This chapter introduces the concept of infinite-dimensional Lie groups. While many elements of the usual, finite-dimensional theory carry over, there is one big obstacle: The exponential map need not exist or be a diffeomorphism around $0$. Hence, in general, there is no connection between the local description in terms of Lie algebra and global features of the theory. To avoid the partly pathological complications that come with such a disconnection of Lie algebraic and Lie group theoretic aspects, this work mostly restricts itself to the discussion of Lie groups with a locally diffeomorphic exponential map (these are called locally exponential Lie groups). After studying the basic properties of Lie groups, the focus shifts to different types of Lie subgroups. The last section is devoted to actions of Lie groups on manifolds and, in particular, a general slice theorem is proven.

The beginnings of infinite-dimensional Lie theory can be found for example in \parencite{Hamilton1982}, where the attention is restricted to Fréchet Lie groups. In their forthcoming book, \textcite{GlocknerNeeb2010} give an elaborate discussion of these aspects for arbitrary locally convex manifolds (see also the review paper \parencite{Neeb2006}). In some ways, the first two sections of this chapter can be seen as an embedding of the Fréchet results in this more general setting (and is therefore based upon the just given references). But to the best knowledge of the author, there is no general study of Lie group actions beyond the Banach category. Some specific examples are discussed in literature (e.g. \parencite{AbbatiCirelliEtAl1989,EbinMarsden1970}). The most systematic and integrated discussion known to the author is the PhD thesis of \textcite{Subramaniam1984}. However, even in this work the slice theorem is only stated for elliptic actions between section spaces. \Autoref{sec::lieGroupAction} collects and combines these ideas into a more general treatment.

\section{Lie groups and the exponential map}
A Lie group combines the algebraic structure of a group and the smooth structure of a manifold in a compatible way. 
\begin{defn} %
	A \emphDef{(locally convex) Lie group $G$} is a group carrying an additional smooth locally convex manifold structure such that multiplication and inversion 
	\begin{alignat}{2}
		&\mult: G \times G \to G,& \qquad &(g, h) \mapsto g \cdot h \\
		&\inv: G \to G,& &g \mapsto g^{-1} 
	\end{alignat}
	are smooth maps. A morphism of Lie groups is a smooth group homomorphism.

	If the manifold is a tame Fréchet manifold and the above maps are tame smooth, $G$ is called a \emphDef{tame Fréchet Lie group}.
\end{defn}
Keeping one factor in the multiplication map fixed induces diffeomorphisms 
\begin{equation}
	\La_g (h) \defeq g \cdot h, \qquad \Ra_g (h) \defeq h \cdot g,
\end{equation}
which are called left and right translation.

As usual, the tangent space $\liea{g} \equiv T_e G$ at the identity element is identified with the space of left invariant vector fields on $G$. More explicitly, this isomorphism assigns to every $A \in \liea{g}$ the vector field $A_g \defeq L'_g A$ and conversely every left invariant vector field gives an element of $\liea{g}$ upon evaluation at the identity. Note that $L'_g$ is actually a shorthand notation for $(L_g)'_e$. In the following, the evaluation of a tangent map at the identity is implicitly understood in order to streamline the notation. The commutator of vector fields induces a bracket operation $\commutator{\cdot}{\cdot}$ on $\liea{g}$, which in a chart is determined by the (Taylor expansion up to second order of) multiplication and thus is continuous (see \parencite[Remark II.1.8]{Neeb2006} for details). The concept of a Lie algebra formalises the preceding remark. 
\begin{defn} %
	A \emphDef{(locally convex) Lie algebra} is a locally convex space $\liea{g}$ endowed with a continuous, bilinear, alternating map $\commutator{\cdot}{\cdot}: \liea{g} \times \liea{g} \to \liea{g}$ such that the Jacobi identity
	\begin{equation}
		\commutator{\commutator{A}{B}}{C} + \commutator{\commutator{B}{C}}{A} + \commutator{\commutator{C}{A}}{B} = 0
	\end{equation}
	holds for all $A,B,C \in \liea{g}$.	A morphism of Lie algebras is a continuous, linear map which preserves the commutator. %
\end{defn} 
\begin{example}[Diffeomorphism group] \label{ex::diffGroup}
	Let $M$ be a finite-dimensional, compact manifold. This examples illustrates the tame Fréchet Lie group structure on the group of diffeomorphisms $G = \DiffGr(M)$, which will be modelled on the space of vector fields $\liea{g} = \vectorf{M}$.

	The subset of diffeomorphisms is open in $C^\infty(M,M)$ \parencite[Theorem 1.7]{Hirsch1997}, thus it is a submanifold. By \autoref{prop::sectionSpace:tangentSpace}, the tangent space at the identity is given as $T_{\id_M} \DiffGr(M) = \secspaceEx{\id_M^*TM} = \vectorf{M}$, which is a tame Fréchet space. The tame smoothness of the natural group operations, composition and inversion, takes more effort and the reader is referred to the literature \parencite[Theorem 2.3.5]{Hamilton1982}. %
	For later use, the tangent of the composition map $\comp: C^\infty(M,M) \times C^\infty(M,M) \to C^\infty(M,M), (f,g) \mapsto f \circ g$ is stated
	\begin{equation} \label{eq::diffGroup:tangentComposition}
		(\comp)'_{f,g} (X,Y) = f'Y + X \circ g
	\end{equation}
	and it is remarked that the first (second) term equals the derivative of the left (right) translation on $\DiffGr(M)$.

	In order to clarify the specific meaning of a vector field $X \in \vectorf{M}$ it will be denoted by $\secmap{X}$ if regarded as a left invariant vector field on $\DiffGr(M)$. The flow of $\secmap{X}$ is given by $\flow^\secmap{X}_t: f \mapsto f_t \defeq f \circ \flow^X_{t}$, because $\diff{}{t} f_t = f'X$ equals the left transport of $\secmap{X}$ to $f \in \DiffGr(M)$. With the aid of this information the Lie bracket on the Lie algebra of $\DiffGr(M)$ computes to
	\begin{align}
		\commutator{\secmap{X}}{\secmap{Y}}_{\id_M}	& \stackrel{\hphantom{\eqref{eq::diffGroup:tangentComposition}}}{=}
								\diffAt{}{t}{t=0} (\flow^{\secmap{X}}_{-t})' {\secmap{Y}}_{\flow^{\secmap{X}}_{t}(\id_M)}\\
													& \stackrel{\hphantom{\eqref{eq::diffGroup:tangentComposition}}}{=} 
								\diffAt{}{t}{t=0} (\flow^{\secmap{X}}_{-t})' {\secmap{Y}}_{\flow^X_{t}} 	\\
													& \stackrel{\hphantom{\eqref{eq::diffGroup:tangentComposition}}}{=} 
								\diffAt{}{t}{t=0} (\flow^{\secmap{X}}_{-t})' ( (\flow^X_t)' Y ) \\
													& \stackrel{\eqref{eq::diffGroup:tangentComposition}}{=} 
								\diffAt{}{t}{t=0} (\flow^{\secmap{X}}_{-t} \circ \La_{\flow^{X}_{t}})' Y 	\\
													& \stackrel{\hphantom{\eqref{eq::diffGroup:tangentComposition}}}{=}
								\diffAt{}{t}{t=0} (\flow^{X}_{t} \circ \Ra_{\flow^{X}_{-t}})' Y 	\\
													& \stackrel{\eqref{eq::diffGroup:tangentComposition}}{=} 
								\diffAt{}{t}{t=0} (\flow^{X}_{t})' Y_{\flow^{X}_{-t}}  = - \commutator{X}{Y},	
	\end{align}
	where $(\flow^{\secmap{X}}_{-t} \circ \La_{\flow^{X}_{t}})(f) = \flow^X_t \circ f \circ \flow^X_{-t} = (\flow^{X}_{t} \circ \Ra_{\flow^{X}_{-t}})(f)$ was used. That is, the Lie algebra structure is the negative of the usual one. 
	Summarizing, $\DiffGr(M)$ is a tame Fréchet Lie group. This result for the smooth diffeomorphisms is in sharp contrast to their finite-differentiable counterparts $\DiffGr^k(M)$, which are Banach manifolds but for which the composition is not smooth (this is another example where the loss of derivative manifests itself, see \eqref{eq::diffGroup:tangentComposition}). Thus $\DiffGr^k(M)$ is no Lie group. This observation implies that the theory of Banach Lie groups and their actions on Banach manifolds cannot be applied and forces one to consider `pseudo' Lie groups with non-smooth actions if one wants to stay in the Banach category (see for example \parencite{Gay-BalmazRatiu2006}).    
\end{example}
The exponential map connects information from the Lie algebra to properties of the Lie group. Thereby, it is a valuable tool in the study of Lie groups. However, in the locally convex setting, an exponential map does not need to exist or, even if it exists, it is not necessary a local diffeomorphism at $0$. To see this, recall that the exponential map is defined by the flow of the associated left invariant vector field and that the flow of a vector field is not ensured to exist beyond Banach manifolds. Likewise, the local diffeomorphism property follows from the inverse mapping theorem applied to $\exp'_0 = \id_{\liea{g}}$. Nevertheless, in applications one can often prove the existence of a locally diffeomorphic exponential map and thus it is worthwhile to study this subclass of Lie groups.
\begin{defn} \label{defn::lieGroup:expMap} %
	Let $G$ be a Lie group with Lie algebra $\liea{g}$. An \emphDef{exponential map} is a smooth map $\exp: \liea{g} \to G$ such that for every $A \in \liea{g}$ the associated curve $\gamma_A(t) \defeq \exp(tA)$ is a one-parameter group with $\diffAt{}{t}{t=0} \gamma_A(t) = A$. If there exists an exponential map, then $G$ is called a \emphDef{Lie group with exponential map}. If $\exp$ additionally is a local diffeomorphism around $0$, then $G$ is called a \emphDef{locally exponential Lie group}.
\end{defn} 
Since the solution of a differential equation in locally convex spaces need not be uniquely determined by its initial value, a priori there could exist many exponential maps for a given Lie group. But with the aid of the Lie group structure, one can prove uniqueness of the exponential map (see \parencite[remarks after Definition II.5.1]{Neeb2006}).   

\begin{example}[Exponential map of the diffeomorphism group] \label{ex::diffGroup:expMap}
	Let $G = \DiffGr(M)$ be the diffeomorphism group of a finite-dimensional, compact manifold $M$. Its Fréchet Lie group structure was examined in \autoref{ex::diffGroup}. The Lie algebra is identified with the space of smooth vector fields on $M$. Recall that every vector field on a compact manifold has a complete flow \parencite[Corollary IV.2.4]{Lang1998}. Thus the exponential map
	\begin{equation}
		\exp: \vectorf{M} \to \DiffGr(M), \qquad X \mapsto \flow^X_1 (\cdot)
	\end{equation}
	is well defined\footnote{This also implies that the exponential map of $\DiffGr(M)$ for a non-compact manifold $M$ is only defined on the subset of complete vector fields and thus is no `true' exponential map in the sense of \autoref{defn::lieGroup:expMap}.}. It is indeed the exponential map of $\DiffGr(M)$, since the scaling property of flows implies $\exp(t \secmap{X}) = \flow^X_t$ and thus the required one-parameter property. One can show that $\exp$ is smooth \parencite[p. 455]{Kriegl1997}.

	Nevertheless, $\exp$ is not locally surjective as was first observed by \parencite{Freifeld1968, Kopell1971}. For concreteness this result is shown here only for the case $M = S^1$. The idea is to state properties of diffeomorphisms in the image of $\exp$ and then construct a specific diffeomorphism not possessing these properties but lying arbitrarily close to the identity map.

	\emph{Claim: Every fixed point free diffeomorphism of the form $f = \exp(X)$ for $X \in \vectorf{S^1}$ is conjugate to some rotation $R_\vartheta$ about a fixed angle $\vartheta$.}\\
	To prove $f = g^{-1} \circ R_\vartheta\, \circ g$, it is enough to show that every nowhere vanishing vector field $X$ is $g$-conjugate to the constant vector field $Y_\vartheta: S^1 \ni \phi \mapsto \vartheta$. In fact, $g'X = Y_\vartheta \circ g$ implies $g \circ \flow^X_t = \flow^{Y_\vartheta}_t \circ \,g$, but $\flow^{Y_\vartheta}_t(\phi) = \phi + t \vartheta$ and so $\exp(Y_\vartheta) = R_\vartheta$. To construct such a diffeomorphism $g$, write the vector field $X$ as $X_\phi = X(\phi) \difp{}{\phi}$ for a non-zero function $X: S^1 \to \R$ and define
	\begin{equation}
		g: S^1 \to S^1, \qquad \phi \mapsto \vartheta \int_0^\phi \frac{1}{X(\tilde\phi)} \dif \tilde\phi.
	\end{equation}
	Now, $g'(\phi)X(\phi) = \vartheta$ is a constant vector field as desired. 

	In particular, if a diffeomorphism of the form $f = g^{-1} \circ R_\vartheta\, \circ g$ has \emph{one} periodic orbit $\set{x_k}_{k=1}^n$ of length $n$, than \emph{all} points are periodic with the same length. To see this, construct the sequence $y_{k+1} = f(y_k), y_1 = y$ for an arbitrary starting point $y \in S^1$. Since $f$ decomposes into $g^{-1} \circ R_\vartheta \circ g$ the recursion relation translates into $g(y_{k+1}) = g(y_k) + \vartheta$. Hence, $g(y_n) - g(y_1)= n \vartheta = g(x_n) - g(x_1) = 0$ and thus $y_n = y$.  

	As outlined above, the existence of a fixed point free diffeomorphism not conjugate to a rotation but lying close to the identity implies that $\exp$ is not locally surjective. For this purpose consider the map 
	\begin{equation}
		f: S^1 \to S^1, \qquad \theta \mapsto \theta + \frac{2 \pi}{n} + \varepsilon \sin^2(nx).  	
	\end{equation}
	For every neighbourhood of $\id_{S^1}$ one can find $n$ and $\varepsilon > 0$ such that $f$ belongs to that neighbourhood. Note that all potential fixed points are complex valued and thus $f$ is fixed point free. Furthermore, the only periodic orbit of length $n$ consists of the points $\theta = \sfrac{2 k \pi}{n}$ for $k=1, \dotsc, n$. Since these are isolated points, $f$ cannot be conjugate to a rigid rotation by the above remark.  

	The inverse function theorem cannot be applied in the present case since the derivative of $\exp$ is not injective in an open neighbourhood of $0$. In particular, $GL( \vectorf{S^1})$ is not open in $L( \vectorf{S^1}, \vectorf{S^1})$.
\end{example}
\newpage
\begin{example}[Current groups, \parencite{CirelliMania1985}] \label{ex::lieGroup:currentGroup}
	Let $\pi: F \to M$ be a fibre bundle over a compact, finite-dimensional manifold $M$ whose typical fibre is a fixed Lie group $G$ (not necessary compact as often assumed in similar circumstances). Denote by $\secspace{F}$ the space of smooth sections of $F$ endowed with the compact-open tame Fréchet manifold structure discussed in \autoref{sec::sectionSpace}.

	The group operations $\secmap{mult}$ and $\secmap{inv}$ are induced from the fibre preserving smooth maps
	\begin{alignat}{2}
		&\mult: F \times_M F \to F,& \qquad 	&(m, a,b) \mapsto (m, a \cdot b)\\
		&\inv: F \to F, 		   &			&(m, a) \mapsto (m, a^{-1}).
	\end{alignat}
	Here $F \times_M F$ %
	denotes the fibre product over $M$ of $F$ with itself and the canonical tame identification $\secspaceEx{F} \times \secspaceEx{F} = \secspaceEx{F\times_M F}$ is used in the definition of $\secmap{mult}$. By \autoref{prop::sectionSpace:pushforwardTameSmooth} the group operations are smooth tame and therefore $\secspace{F}$ is a tame Fréchet Lie group, which is called a $\emphDef{current group}$.

	Denote the fibre of $F$ over $m$ by $G_m$, its Lie algebra by $\lieaE{G_m}$ and the induced Lie algebra bundle by $\lieaE{F}$. With this notation and \autoref{prop::sectionSpace:tangentSpace}, the tangent space of $\secspace{F}$ at the identity $\secmap{e} \in \secspace{F}$ consists of sections of the bundle $\lieaE{F}$, since 
	\begin{equation}
		(\secmap{e}^*(VF))_m = V_{\secmap{e}(m)}F = T_{\secmap{e}(m)} G_m = \lieaE{G_m}.
	\end{equation}
	The Lie bracket on the Lie algebra $\lieaE{\secspace{F}}$ of $\secspace{F}$ is continuous as it is induced from the fibre-wise commutator. Verifying that the so defined bracket agrees with the canonical one on $T_{\secmap{e}} \secspace{F}$ is a matter of a straightforward calculation, which is left to the reader.

	The exponential map on the fibre $\exp_m: \lieaE{G_m} \to G_m$ gives rise to a smooth fibre-preserving bundle map $\exp: \lieaE{F} \to F$ which in turn induces a tame smooth exponential map $\secmap{exp}: \lieaE{\secspace{F}} \to \secspace{F}$. Being maps between finite-dimensional spaces, every $\exp_m$ is a local diffeomorphism at $0$. Hence there exists an open $0$-neighbourhood $V_m \subseteq \lieaE{G_m}$ on which $\exp_m$ is a diffeomorphism onto an open neighbourhood $U_m \subseteq G_m$ of the identity element $e_m$. Denote the inverse map by $\log_m$ and combine the $U_m$ to an open subset $U = \bigsqcup_{m \in M} U_m$ around the image of the identity section. Again by \autoref{prop::sectionSpace:pushforwardTameSmooth}, the induced map $\secmap{log}: \secspace{F} \supseteq \secspace{U} \to \lieaE{\secspace{F}}$ is a smooth tame map defined on the open subset $\secspace{U} = \setc{\phi \in \secspace{F}}{\phi(M) \subseteq U}$ and is clearly inverse to $\secmap{exp}$. In conclusion, all current groups are locally exponential Lie groups. 
\end{example}

\section{Lie subgroups}

\begin{defn} %
	Let $G$ be a Lie group with Lie algebra $\liea{g}$. A \emphDef{Lie subgroup} of $G$ is a subgroup $H \subseteq G$ which is also a  submanifold. That is, there exist a closed subspace $\liea{h} \subseteq \liea{g}$ and a chart $\kappa: U \to \liea{g}$ around the identity such that $\kappa(U \cap H) = \kappa(U) \cap \liea{h}$. If $H$ additionally is a splitting submanifold, then $\liea{h}$ is topologically complemented and $H$ is called a \emphDef{splitting Lie subgroup}. 
\end{defn}
\begin{remarks}
	\item Since $H$ is an initial submanifold (see \cref{prop::lcm:SubmanifoldIsInitial}), the restrictions $\mult_H: H \times H \to H$ and $\inv_H: H \to H$ of the group operations to $H$ are smooth maps. Thus $H$ is a Lie group itself. A simple calculation shows that the model space $\liea{h}$ is actually a closed Lie subalgebra and therefore it is also the Lie algebra of $H$ as expected. %

	\item Every Lie subgroup is closed in the ambient space. The proof is adopted from the Banach case. 
	Choose a chart $\kappa: U \to \liea{g}$ around $e$ with the submanifold-property as in the above definition. $\liea{h}$ is by definition a closed subspace and therefore renders $\kappa(U) \cap \liea{h}$ locally closed. By continuity, $U \cap H$ is also locally closed and thus it can be assumed to be closed in $U$ after shrinking $U$. Now let $h_\alpha$ be a net converging to $h \in G$. Since the elements $h^{-1} h_\alpha$ converge to $e$, they lie in $U$ for sufficiently large $\alpha$. By a possible adaptation of $U$, the double net $h_\beta^{-1} h_\alpha = (h^{-1} h_\beta)^{-1} (h^{-1} h_\alpha)$ lies in $U$ and thus also in $U \cap H$. The latter set is closed in $U$ and therefore performing the limit with respect to $\beta$ one obtains $h^{-1} h_\alpha$ as an element of $H$. Hence $h \in H$. 

	In finite dimensions every closed subgroup of a Lie group is automatically a Lie subgroup. This is not true in the infinite-dimensional setting, although one can prove that every finite-dimensional, locally compact subgroup of a locally exponential Lie group is a Lie subgroup, see \parencite[Theorem 7.3.14]{GlocknerNeeb2010} or \autoref{prop::lieGroup:finiteDimSubgroupIsStrongSplittingLieSubgroup} for a stronger result in the context of Fréchet manifolds.
\end{remarks}

In finite dimensions, for every Lie subgroup $H \subseteq G$, the set of left cosets $G/H \defeq \setc{gH}{g \in G}$ carries a manifold structure turning the natural projection $\pi: G \to G/H$ into a principal $H$-bundle. Recall that the local diffeomorphism 
\begin{equation}
	\mu: \liea{g} =  \liea{k} \oplus \liea{h} \to G, \qquad (X,Y) \mapsto \exp(X) \exp(Y)
\end{equation}
plays an important role in the construction of charts on $G/H$. This concept needs some refinement for locally convex manifolds since in that case $\mu$ fails to be a local diffeomorphism (and the exponential map does not need to exist in the first place).

\begin{proposition}[{\parencite[Proposition 7.1.21]{GlocknerNeeb2010}}] \label{prop::subliegroup:strongSplitting}
	Let $G$ be a Lie group and $H \subseteq G$ a subgroup. For a Lie group structure on $H$, the following are equivalent:
	\begin{thmenumerate}
		\item $H$ is a splitting Lie subgroup and there exists a unique smooth structure on the left cosets $G/H$ such that the natural projection $\pi: G \to G/H$ defines a smooth principal $H$-bundle structure. In this case, a map $f: G/H \to M$ to some manifold $M$ is smooth if and only if the induced map $\hat f = f \circ \pi: G \to M$ is smooth. \label{prop::subliegroup:strongSplitting_leftCosetPrincipalBundleAndUniversalProperty} %
		\item The inclusion $\iota: H \to G$ is a morphism of Lie groups and there exist an open subset $V$ around $0$ in some locally convex space $\liea{k}$ and a smooth map $\sigma:V \to G$ with $\sigma(0)=e$ such that
		\begin{equation}
		 	\mu: V \times H \to G, \qquad (X, h) \mapsto \sigma(X) h
		\end{equation} 
		is a diffeomorphism onto an open subset of G. In this case, $\mu( V \times H)$ is a tube around $H$ in $G$.
	\end{thmenumerate}
	The result also holds restricted to the category of tame Fréchet manifolds. If one (and thus both) of these conditions is satisfied, then $H$ is called a \emphDef{principal Lie subgroup}\footnote{In \parencite{GlocknerNeeb2010} such subgroups are merely named split Lie subgroups, but as the requirement of being a splitting submanifold is not enough to ensure the above properties a special name is chosen in this thesis.}.   
\end{proposition}
\begin{proof}
	The proof will be given only for the $C^\infty$-setting, but also holds true if all occurrences of `smooth' are replaced by `tame smooth'. 

	Assume that (i) is satisfied. There exists a chart $\rho: U \to \rho(U) \equiv V \subseteq \liea{k}$ on $G/H$ centred at $\equivClass{e}$ such that $\pi$ admits a smooth section $\phi: U \to G$. For the smooth chart representation $\sigma \defeq \phi \circ \rho^{-1}: V \to G$, the map $\mu(X,h) \defeq \sigma(X) h$ is a diffeomorphism onto its image. Indeed, smoothness follows from writing $\mu$ as the composition of $\sigma$ and the group multiplication. The inverse corresponds to the local trivialisation $\pi^{-1}(U) \to U \times H$ over $U$ (composed with the chart $\rho$).    

	Consider now the reverse direction. The following comments give an outline of the proof. With the aid of the diffeomorphism $\mu$, a chart on $G/H$ around the identity coset is constructed which identify $\mu(V \times H)/H$ homeomorphically with $V$. This chart is now transported to the whole space $G/H$ by left translation, rendering the quotient space as a manifold. With respect to the induced atlas, $\pi$ is smooth and possesses $\sigma$ as a local section. This conceptual framework is now described in more detail. 

	Let $\sigma: \liea{k} \supseteq V \to G$ and $\mu: V \times H \to G$ as stated in the proposition. Since $\mu$ is a diffeomorphism onto its image, the subgroup $H \subseteq G$ is also a submanifold. The tangent $\mu'_{X, h}$ maps $\liea{k} \times \liea{h}$ linearly and topologically isomorphic to $\liea{g}$, thus $H$ is a splitting Lie subgroup. The charts on $G/H$ are constructed by considering the smooth map $\pi \circ \sigma: V \to \pi(W)$, where $W \subseteq G$ is the open image of $\mu$. By definition, $\pi \circ \sigma$ is surjective and the following short argument shows it is injective as well. Let $X$ and $\tilde{X}$ be two elements mapped to the same equivalence class $gH$. Then there exists $h \in H$ such that $\sigma(X) = \sigma(\tilde{X}) h$ holds. Since $\mu$ is a diffeomorphism, this relation implies $X = \tilde{X}$. Finally, $W / H = \pi(W)$ is open in $G/H$ and by $\pi \circ \sigma$ homeomorphic to $V$ (as $\mu$ is a homeomorphism) and thus its inverse $\rho: W/H \to V$ defines a chart at the identity coset. Left translation induces a homeomorphism $\check{\La}_{\equivClass{a}}$ on $G/H$ through which the charts are translated to the whole quotient space, $\rho_{\equivClass{a}} = \rho \circ \check{\La}_{\equivClass{a^{-1}}}: \check{\La}_{\equivClass{a}}(W/H) \to V$. In order to verify smoothness of chart transitions introduce the smooth map $\pr_V: W \to V$ by making the following diagram commutative.
	\begin{equation}
		\begin{tikzcd}
			V \times H \arrow{r}{\mu} \arrow[swap]{dd}{\pr_1}		&[5ex] W \arrow[swap]{ddl}{\pr_V} \arrow{rd}{\pi} 		&[4ex] \\
	    															&														& W/H \\
			V \arrow[swap]{r}{\sigma}								& G \arrow[swap]{uu}{\Ra_h} \arrow[swap]{ru}{\pi} 		&
		\end{tikzcd}
	\end{equation}
	Now the chart transition from $\rho_{\equivClass{a}}$ to $\rho_{\equivClass{b}}$ is expressed as
	\begin{equation}
		\rho_{\equivClass{b}} \circ \rho_{\equivClass{a}}^{-1} = \rho \circ \check{\La}_{\equivClass{b^{-1}}} \circ \check{\La}_{\equivClass{a}} \circ \pi \circ \sigma = \rho \circ \pi \circ \La_{b^{-1} a} \circ\, \sigma = \pr_V \circ \La_{b^{-1} a} \circ\, \sigma,
	\end{equation}
	where the necessary restrictions to appropriated domains is understood. Thus, these charts yield a smooth atlas modelling $G/H$ on the space $\liea{k}$. %

	Moreover, note that the chart representation $\rho_{\equivClass{a}} \circ \pi$ of $\pi$ equals $\pr_V \circ \La_{a^{-1}}$ and thus $\pi$ is a smooth submersion. The bundle $G \to G/H$ is locally trivial because it allows smooth sections of the form $W/H \ni \sigma(X)H \mapsto \sigma(X) \in G$ (and translations thereof).

	Finally, every map $f: G/H \to M$ factors locally through a section $\phi$ of $\pi$ and the induced map $\hat{f}: G \to M$. Hence, the smoothness of $f$ and $\hat{f}$ are equivalent. This also implies that the smooth structure on the quotient $G/H$ is the unique one fulfilling the required properties. 
\end{proof}

In the case of Lie groups with exponential maps a local product structure around the identity is already sufficient since it can be transported to the whole subgroup.
\begin{proposition} \label{prop::lcliegroup:strongSplittingLieSubgroupByLocalDiffeomorphism}
	Let $G$ be a Lie group with smooth exponential map and $H \subseteq G$ a splitting Lie subgroup. Denote the complement of $\liea{h}$ in $\liea{g}$ by $\liea{k}$ and let $V \subseteq \liea{k}$ be an open $0$-neighbourhood. If the map 
	\begin{equation}
	 	\mu: V \times H \to G, \qquad (X, h) \mapsto \exp(X) h
	\end{equation} 
	is a local diffeomorphism at $(0, e)$, then $H$ is a principal Lie subgroup.
\end{proposition}
\begin{proof}
	Since $\mu$ is a local diffeomorphism at $(0, e)$ there exists a neighbourhood $U_H \subseteq H$ around $0$ such that the restriction $\mu_{\restriction V \times U_H}$ is a diffeomorphism onto an open $e$-neighbourhood $U_G \subseteq G$ (after potentially shrinking $V$). Due to the identity $\mu(X, a h) = \mu(X, a) h$ the map $\mu$ is a local diffeomorphism (at every point $(X,h)$). 

	Thus, in order to apply \cref{prop::subliegroup:strongSplitting}, it is enough to show injectivity of $\mu$. Because $U_H$ is open in $H$, there exists an open $e$-neighbourhood $W_G \subseteq G$ such that $U_H = H \cap W_G$. By  shrinking $V$ further, the set $\exp(-V) \exp(V)$ can be assumed to lie completely in $W_G$. Now let $(X,a)$ and $(Y,b)$ be two points with the same image under $\mu$. Then 
	\begin{equation}
		\exp(-Y) \exp(X) = b a^{-1}.
	\end{equation}
	The left side lies in $W_G$ by assumption and the right side is an element of $H$, thus both expressions are contained in $U_H$. On the other hand, $\mu$ is bijective on $V \times U_H$ and hence the calculation $\mu(Y, b a^{-1}) = \mu(Y,b) a^{-1} = \mu(X,a) a^{-1} = \mu(X,e)$ shows the desired result $X=Y$ and $a=b$. 
\end{proof}

Since the derivative $\mu'_{0,e}: \liea{k} \times \liea{h} \to \liea{g}$ is just the direct sum isomorphism $\liea{k} \oplus \liea{h} = \liea{g}$, the inverse function theorem provides the necessary local diffeomorphism in order to apply \autoref{prop::lcliegroup:strongSplittingLieSubgroupByLocalDiffeomorphism} to Banach Lie groups (thus all splitting Lie subgroups of a Banach Lie group are principal). An application of \autoref{prop::tameFrechet:extendFiniteDimSplitting} yields the following important consequences for tame Fréchet Lie groups. 
\begin{corollary} \label{prop::lieGroup:finiteDimSubgroupIsStrongSplittingLieSubgroup} 
	Let $G$ be a tame Fréchet Lie group with tame smooth exponential map which is a local diffeomorphism around $0$. Then every finite-dimensional closed subgroup $H \subseteq G$ is a principal Lie subgroup. 
\end{corollary}

\section{Lie group actions} \label{sec::lieGroupAction} 
This section discusses the important concept of a Lie group action. It will be shown that, as in the finite-dimensional case, proper actions possess orbits which are (embedded) submanifolds. The core statement consists in a general slice theorem.

The main focus will be on tame Fréchet Lie groups since the desired results rely heavily on the Nash-Moser inverse function theorem. While the general theory of infinite-dimensional Lie theory gained much attention recently, the same thing cannot be said about Lie group actions. Clearly, there are some discussions of concrete actions of special Lie groups (especially the diffeomorphism group is to be named), but the principal textbooks and references \parencite{GlocknerNeeb2010,Kriegl1997,Hamilton1982} for infinite-dimensional Lie theory do not mention this topic beyond the very basic definitions. The only publication known to the author that comes close to the scope of this section is the PhD thesis of \textcite{Subramaniam1984}. Yet, even there, the main results are solely obtained for section spaces and elliptic actions. Nevertheless, many statements in this section are straightforward generalisations of \parencite{Subramaniam1984} making it the main foundation besides some classical, finite-dimensional literature \parencite{Abraham1980,Lang1998}.

\begin{defn} %
	Let $G$ be a Lie group and $M$ a manifold. A smooth map $\Upsilon: G \times M \to M$ is called a smooth \emphDef{left action} if 
	\begin{equation}
		\Upsilon_e = \id_M \quad \text{and} \quad \Upsilon_g \circ \Upsilon_h = \Upsilon_{gh}
	\end{equation}
	hold for the induced maps $\Upsilon_g: M \to M, m \mapsto \Upsilon(g,m)$. In other words, the assignment $g \mapsto \Upsilon_g$ is an abstract group homomorphism of $G$ to $\DiffGr(M)$. If it is instead an anti-homomorphism, then $\Upsilon$ is called a right action.   

	Furthermore, define $\Upsilon_m: G \to M$ by $ g \mapsto \Upsilon(g,m)$. Its image $G \cdot m \defeq \Upsilon_m(G) \subseteq M$ is the \emphDef{orbit} of the action at $m$. The inverse image $G_m \defeq \Upsilon^{-1}_m(m) = \setc{g \in G}{\Upsilon(g, m) = m}$ is called the \emphDef{stabilizer subgroup} of $m$.

	For every $A \in \liea{g}$ the \emphDef{Killing vector field} $A^*$ on $M$ is pointwise defined by
	\begin{equation}
		A^*_m \defeq \Upsilon'(0_m, A_e) = \Upsilon'_m A. 
	\end{equation}
\end{defn}
If $G$ has a smooth exponential map, then every Killing vector field $A^*$ has a flow, namely $(m, t) \mapsto \Upsilon(\exp(tA), m)$.

The special class of proper actions is convenient since their orbits are splitting submanifolds, as will be shown later. The definition of proper maps and equivalent characterisations for compactly generated Hausdorff spaces are discussed in \autoref{cha::topology:properMapCompactlyGeneratedSpace}. 
\begin{defn}
	A smooth left action $\Upsilon$ of a Lie group $G$ on a manifold $M$ is called \emphDef{proper} if its extension 
	\begin{equation}
		\Upsilon_{\text{diag}}: G \times M \to M \times M, \qquad (g,m) \mapsto (\Upsilon_g(m), m) 	
	\end{equation} 
	is a proper map.
\end{defn}
Proper actions possess useful properties but checking properness directly based on the definition can be intricate. The following proposition gives practical alternatives, which directly follow from \autoref{prop::topo:properMapEquivCharacterization}. %
\begin{proposition} \label{prop:lieGroupAction:properEquivalent}
	Let $\Upsilon: G \times M \to M$ be a smooth left action. The following statements are equivalent to properness of $\Upsilon$ for adequate conditions on the involved spaces:

	If $M$ is a compactly generated Hausdorff space:
	\begin{enumerate} \label{prop:lieGroupProperAction:closedOribtMapCompactStabilizer}
		\item $\Upsilon_{\text{diag}}$ is a closed map and every stabilizer subgroup is compact. 
	\end{enumerate}
	If $G$ and $M$ are even metric spaces:
	\begin{enumerate}[resume]
		\item Let $g_i$ and $m_i$ be arbitrary sequences in $G$ and $M$, respectively. If $m_i$ and $\Upsilon(g_i, m_i)$ converge, then the sequence $g_i$ has a convergent subsequence. \qedhere
	\end{enumerate}
\end{proposition}

The metric on $M$ and $G$ is often induced from a locally equivalent, graded Riemannian metric, see \cref{prop::riemannianGeometry:locEquivalentMetricInducesLengthMetric}. In fact, one can prove metrizability of a wide class of Lie groups.

\begin{corollary}[{\parencite[p. 53]{Subramaniam1984}}] \label{prop::lcLieGroup:metrisableZeroTame}
	Every $0$-tame Fréchet Lie group has a complete left (or right) invariant metric.
\end{corollary}
\begin{proof}
	Note that the proof of \cref{prop::riemannianGeometry:locEquivalentMetricInducesLengthMetric} does not rely on the inner product structure, instead, only the induced locally equivalent seminorms are required. Hence, the same arguments prove this corollary once such a system of seminorms is constructed.

	The idea consists in translating a given directed fundamental system of seminorms $\normdot_k$ on the Lie algebra to the whole group. Now $0$-tameness of $\La_g$ (or $\Ra_g$) implies that the so-constructed seminorms are locally equivalent. 
\end{proof}

\begin{proposition} \label{prop:lieGroupProperAction:properties} %
	Let the Lie group $G$ act properly on a manifold $M$ via $\Upsilon: G \times M \to M$. Assume that the topology of $M$ is compactly generated. Then the following holds for every point $m \in M$:
	\begin{enumerate}
		\item The orbit map $\Upsilon_m: G \to M$ is proper and the orbit $G \cdot m$ is closed in $M$.
		\item $G / G_m$ and $M / G$ with the quotient topology are Hausdorff.
		\item The orbit map descends to a map $\check\Upsilon_m: G/G_m \to M$, which makes the following diagram commutative and is a homeomorphism onto $G \cdot m$. 	
		\begin{equation}
			\begin{tikzcd}
				G \arrow{r}{\Upsilon_m} \arrow[swap]{d}{\pi_{G_m}} & M \\
				G/G_m \arrow[swap]{ru}{\check\Upsilon_m} & 
			\end{tikzcd}
		\end{equation}
	\end{enumerate}
\end{proposition}
\begin{proof}
	\begin{thmenumerate}*
		\item Properness of $\Upsilon_m$ follows by considering $\Upsilon_{\text{diag}}(\cdot, m) = \Upsilon_m(\cdot) \times \set{m}$ for fixed $m$. Now by \cref{prop::topo:properMapEquivCharacterization_closedMapAndInverseImageOfPointIsCompact} the image $\Upsilon_m(G) = G \cdot m$ is a closed subset of $M$. 
		
		\item For the Hausdorff property of the quotient spaces recall that the codomain $Y$ of a surjective, continuous, open map $f:X \to Y$ is Hausdorff if and only if $R_f = \setc{(x_1, x_2) \in X \times X}{f(x_1) = f(x_2)}$ is closed in $X \times X$, e.g. \parencite[Lemma 3.2]{Dieck1987}. In the present case the maps 
		\begin{equation}
			\pi_{G_m}: G \to G/G_m \quad \text{and} \quad \pi_M: M \to M / G
		\end{equation}
		are surjective and continuous by definition, as well as open since they are quotient maps with respect to (continuous) group actions. %
		Furthermore, 
		\begin{align}
			R_{\pi_{G_m}} &= \setc{(g,h) \in G \times G}{g G_m = h G_m} = \left((g,h) \mapsto g^{-1}h \right)^{-1}(G_m), \\
			R_{\pi_M} &= \setc{(m, \tilde{m}) \in M \times M}{G \cdot m = G \cdot \tilde{m}} = \Upsilon_{\text{diag}}(G,M)
		\end{align}
		are closed as the inverse image of the closed subset $G_m$ under a continuous map and as the image under a closed map, respectively.
		
		\item By the universal property of quotient spaces there exists a unique continuous and injective map $\check\Upsilon_m$ such that $\check\Upsilon_m \circ \pi_{G_m} = \Upsilon_m$. Moreover, the so-defined map is proper since for a compact subset $K \subseteq M$ the inverse image $\check\Upsilon_m^{-1}(K) = \pi_{G_m} (\Upsilon_m^{-1}(K))$ is compact as the projection of a compact subset. But every closed, continuous, injective map is a homeomorphism onto its image. \qedhere
	\end{thmenumerate}
\end{proof}

The previous \lcnamecref{prop:lieGroupProperAction:properties} collects the most important topological consequences of proper actions. Now, interactions with the smooth structure are investigated.
\begin{proposition} \label{prop::lieGroup:orbitClosedSubmanifold}
	Let $G$ be a tame Fréchet Lie group, $M$ a tame Fréchet manifold and $\Upsilon$ a tame smooth, proper action. Suppose the following conditions are fulfilled for some point $m \in M$:
	\begin{enumerate}
		\item The stabilizer $G_m$ is a principal tame Fréchet Lie subgroup of $G$.
		\item There exists a closed subspace $F_m \subseteq T_mM$ such that $T_mM = \img (\Upsilon_m)'_e \oplus F_m$ is a tame isomorphism.
		\item Locally and for $g \in G$ near $e$, the derivative $(\Upsilon_m)'_g$ followed by the projection on $\img (\Upsilon_m)'_e$ has a tame smooth family of inverses.
	\end{enumerate}
	Then the orbit $G \cdot m$ is a closed splitting submanifold of $M$.  
\end{proposition}
\begin{proof}
	Since G acts properly, \autoref{prop:lieGroupProperAction:properties} implies that the action descends to a homeomorphism $\check\Upsilon_m: G/G_m \to G \cdot m$. In particular, the orbit $G \cdot m$ is closed in $M$. Due to \autoref{prop::submanifold:tameImmersionImpliesSplitSubmanifold} it is enough to show that $\check\Upsilon_m$ is a splitting tame, injective immersion:
	\begin{itemize}
		\item Tame smoothness: Since $G_m$ is a principal Lie subgroup the universal property of \cref{prop::subliegroup:strongSplitting_leftCosetPrincipalBundleAndUniversalProperty} yields smoothness of $\check\Upsilon_m$. Moreover, $\check\Upsilon_m$ is tame smooth since it locally factors through $\Upsilon_m$ and a tame smooth section of $\pi_{G_m}$.

		\item Injectivity of $\check\Upsilon_m$: clear.

		\item Injectivity of $\check\Upsilon'_m$: Due to the identity $(\Upsilon_m)'_g = (\check\Upsilon_m)'_{g G_m} \circ (\pi_{G_m})'_g$ it is enough to show $\ker (\Upsilon_m)'_g \subseteq \ker (\pi_{G_m})'_g$ for all $g \in G$. By equivariance of $\Upsilon_m$ and $\pi_{G_m}$ it actually suffices to consider the case $g=e$. Since $G_m$ is a principal Lie subgroup, there exists a smooth map $\sigma: V \to G$ from an open subset $V \subseteq \liea{k}$ of some locally convex vector space $\liea{k}$ such that $\mu: V \times G_m \to G, \mu(X, h) = \sigma(X)h$ is a diffeomorphism onto an open subset of the identity in $G$. Thus a smooth curve $\gamma: [0,1] \to G$ with $\gamma(0) = e$ induces smooth curves $\gamma_V$ and $\gamma_{G_m}$ in $V$ and $G_m$, respectively, for sufficiently small times $t$. That is, $\gamma(t) = \mu(\gamma_V(t), \gamma_{G_m}(t)) = \sigma(\gamma_V(t)) \gamma_{G_m}(t)$. Now 
		\begin{align}
			(\Upsilon_m)'_e \equivClass{\gamma} &= \diffAt{}{t}{t=0} (\Upsilon_m \circ \gamma)(t)\\
												&= \diffAt{}{t}{t=0} \Upsilon\Bigl(\sigma(\gamma_V(t)) \gamma_{G_m}(t), m\Bigr) \\
												&= \diffAt{}{t}{t=0} \Upsilon\Bigl(\sigma(\gamma_V(t)), m\Bigr).
		\end{align}
		By definition $\sigma(\gamma_V(t))$ does not lie in the stabilizer of $m$ except for the case where it is the identity element and hence $(\Upsilon_m)'_e \equivClass{\gamma}$ vanishes only if $\gamma_V(t)$ is constant for small $t$. On the other hand, $\pi_{G_m}$ corresponds in these product coordinates to the projection on the $V$-component and thus $(\pi_{G_m})'_e \equivClass{\gamma} = \dot\gamma_V$. This shows the desired relation $\ker (\Upsilon_m)'_e \subseteq \ker (\pi_{G_m})'_e$ (actually equality holds). 

		\item Bundle structure of $T \check\Upsilon_m$ and tame inverse: In contrast to the general case, the splitting $T_mM = \img (\Upsilon_m)'_e \oplus F_m$ at the point $m$ integrates to a global splitting subbundle with the aid of the group structure. As for injectivity above, locally one can work with $\Upsilon$ instead of $\check\Upsilon$. Furthermore, by equivariance it is enough to construct the subbundle-chart only for some identity neighbourhood $U \subseteq G$. If the pullback-bundle $\Upsilon_m^* TM$ is trivialized by $\tau: (\Upsilon_m^* TM)_{\restriction_U} \to U \times T_mM$ in such a way that $\tau(T \Upsilon_m) = U \times \img (\Upsilon_m)'_e$ holds\footnote{Recall the notation $T \iota$ for the image-bundle $\bigsqcup_{s \in S} \img (\iota'_s) \subseteq TM$ of an immersion $\iota: S \to M$.}, then $T \check\Upsilon_m$ is a splitting subbundle of $\check\Upsilon_m^* TM$ by hypothesis (ii). Define
		\begin{equation}
			\tau(g, X) \defeq (g, (\Upsilon_{g^{-1}})'_{\Upsilon(g,m)} X) \qquad \text{for } g \in U, X \in T_{\Upsilon(g,m)}M.
		\end{equation}
		Due to the equivariance relation $\Upsilon_{g^{-1}} \circ \Upsilon_m = \Upsilon_m \circ \La_{g^{-1}}$, vectors of the form $X = (\Upsilon_m)'_g Y$ for $Y \in T_gG$ are mapped to $(\Upsilon_m)'_e (\La_{g^{-1}})'Y$ under $\tau$. Therefore, $\tau(T \Upsilon_m) = U \times \img (\Upsilon_m)'_e$ as desired and $T \check\Upsilon_m$ is a splitting subbundle modelled on the typical fibre $\img (\Upsilon_m)'_e$. Relative to this trivialization the bundle map $\Psi^{\check\Upsilon_m}: T \check\Upsilon_m \to T (G/G_m)$ is just the inverse map to
		\begin{equation}
			\begin{tikzcd}
				T_gG \arrow{r}{ (\Upsilon_m)'_g} & T_{\Upsilon(g,m)}M \cong \img (\Upsilon_m)'_e \oplus F_m \arrow{r}{\pr_1} & \img (\Upsilon_m)'_e.
			\end{tikzcd}
			\vspace{-2ex}
		\end{equation}
		Hence by assumption (iii) the family $\Psi^{\check\Upsilon_m}$ is smooth tame. \qedhere
	\end{itemize}
\end{proof}

Under the same assumptions as in the previous proposition the bundle
\begin{equation}
	N O \defeq \bigcup_{g \in G} \set{\Upsilon(g,m)} \times (\Upsilon_g)'_m F_m
\end{equation}
above the orbit $O = G\cdot m$ is the smooth normal bundle. That is, $N O$ is a smooth tame subbundle of $TM_{\restriction O}$ fulfilling $TM_{\restriction O} = TO \oplus NO$. This claim is now substantiated. Denote the projection $NO \to O$ by $\pi_N$ and notice that it is a tame smooth submersion as the restriction of the natural projection $TM \to M$ to $NO$. By the previous proposition, the orbit $O$ is a splitting submanifold of $M$ and can be identified with $G/G_m$ via the tame diffeomorphism $\check\Upsilon_m$. Under this identification, there exists a local section $\sigma: U \to G$ of $\Upsilon_m$, defined on an open subset $U \subseteq O$. Then, every point $\tilde{m} \in U$ can be written as $\tilde{m} = \Upsilon(\sigma(\tilde{m}), m)$ and thereby the map
\begin{equation}
	U \times F_m \ni (\tilde{m}, X) \mapsto (\tilde{m}, (\Upsilon_{\sigma(\tilde{m})})'_m X) \in N_{\tilde{m}}O
\end{equation}
yields a trivialisation of $NO$. Varying $U$ over the whole orbit endows $NO$ with bundle charts. The above expression for the charts, together with the fact that $NO \to TM_{\restriction_O}$ is fibrewise a splitting inclusion, renders $NO$ a subbundle. Finally, $N_{\Upsilon(g,m)}O = (\Upsilon_g)'_m N_mO$ and $T_{\Upsilon(g,m)}O = \img (\Upsilon_{\Upsilon(g,m)})'_e = (\Upsilon_g)'_m T_mO$ imply 
\begin{equation}
	T_{\Upsilon(g,m)}M = (\Upsilon_g)'_m (T_mO \oplus N_mO) = T_{\Upsilon(g,m)}O \oplus N_{\Upsilon(g,m)}O.
\end{equation}
\section{Slice theorem}
Slices provide a valuable tool to investigate group actions, since they reduce a $G$-action on a manifold $M$ to an action of the stabilizer subgroup on some invariant submanifold. The following definition stems from \parencite[Definition 1.1]{IsenbergMarsden1982} and is a slight variation of the standard finite-dimensional version. The latter is too rigid in infinite dimensions.
\begin{defn}
	Let $\Upsilon: G \times M \to M$ be a smooth action of a Lie group $G$ on a manifold $M$. A \emphDef{slice} at $m \in M$ is a submanifold $S \subseteq M$ containing $m$ such that
	\begin{enumerate}
		\item $S$ is invariant under the induced action of $G_m$. 
		\item $(\Upsilon_g S) \cap S \neq \emptyset$ implies $g \in G_m$.
		\item There exists a local section $\chi: G/G_m \supseteq U \to G$ defined in a neighbourhood $U$ of the identity coset such that the map
		\begin{equation}
			\chi^S: U \times S \to M, \qquad (\equivClass{g}, s) \mapsto \Upsilon (\chi(\equivClass{g}), s)
		\end{equation}
		is a diffeomorphism onto a neighbourhood $V \subseteq M$ of $m$. \qedhere
	\end{enumerate}
\end{defn}
A slice need not exist in general and if it exists it is not necessarily unique. They provide a trade-off between being local in $M$ and testing global aspects of the group action. The former is displayed by choosing the submanifold $S$ rather small. In point (ii) one runs through the whole group to ensure that the orbit $G \cdot m$ intersects $S$ only at $m$. Furthermore, condition (iii) reveals that a `larger' stabilizer subgroup results in a more expanded slice.

The important result of \textcite{Palais1961} proves the existence of slices for proper actions in the finite-dimensional realm. In the present Fréchet setting the statement needs some refinement and additional conditions. The following theorem is inspired by \parencite[Chapter~3]{Subramaniam1984} but in this generality it represents original work\footnote{The slice theorem of \textcite{Subramaniam1984} is restricted to section spaces and only applies to elliptic actions.}.
\begin{theorem}[Slice theorem] \label{prop::liegroup:sliceTheorem}
	Let $G$ be a tame Fréchet Lie group and $M$ a tame Fréchet manifold. A tame smooth proper $G$-action $\Upsilon$ has a slice at the point $m_0$ if the following conditions are fulfilled (the first three are the same as in \cref{prop::lieGroup:orbitClosedSubmanifold}):
	\begin{enumerate}
		\item The stabilizer $G_{m_0}$ is a principal tame Fréchet Lie subgroup of $G$.
		\item There exists a closed subspace $F_{m_0} \subseteq T_{m_0}M$ such that $T_{m_0}M = \img (\Upsilon_{m_0})'_e \oplus F_{m_0}$ is a tame isomorphism.
		\item Locally and for $g \in G$ near $e$, the derivative $(\Upsilon_{m_0})'_g$ followed by the projection on $\img (\Upsilon_{m_0})'_e$ has a tame smooth family of inverses.
		\item $M$ carries a $G$-invariant, locally equivalent, graded Riemannian metric $g^k$ such that the $l$-exponential map exists for some $l$. Furthermore, assume that the restriction to the normal bundle, $\exp: NO \to M$, is an equivariant local diffeomorphism at every point of the zero section. Here the orbit through $m_0$ is denoted by $O$. \qedhere
	\end{enumerate}
\end{theorem}
\begin{proof}
	First, set up the stage by recalling some previous results. By \autoref{prop::lieGroup:orbitClosedSubmanifold} the orbit $O \equiv G \cdot m_0$ is a closed splitting submanifold and 
	\begin{equation}
		N O = \bigcup_{g \in G} \set{\Upsilon(g,m)} \times (\Upsilon_g)'_m F_m
	\end{equation}
	constitutes the normal bundle over $O$. Furthermore, the graded metric $g^k$ induces pointwise seminorms which combine to a $G$-invariant, compatible metric $\rho_m$ on $T_mM$. \Cref{prop::riemannianGeometry:locEquivalentMetricInducesLengthMetric} yields a length-metric $d$ on $M$ which is compatible with the manifold topology. $d$ inherits the $G$-invariance of $g^k$ as an inspection of the proof of \cref{prop::riemannianGeometry:locEquivalentMetricInducesLengthMetric} reveals.

	The desired slice will be constructed as the image under the exponential map. To make this work, the local diffeomorphism $\exp: NO \to M$ actually has to be a diffeomorphism on a neighbourhood of the zero-section onto an open subset of $M$ containing the orbit. It is sufficient to show injectivity of $\exp$ on some open subset of $NO$ around the zero-section. 
	Since the exponential map is a local diffeomorphism, one can without loss of generality assume that it is injective on sets of the form
	\begin{equation}
		U_\delta(m) \defeq \setc{X_p \in NO}{d(m,p) < \delta, \rho_m(0, X_p) < \delta}, 	
	\end{equation} 
	where $\delta > 0$ and $m \in O$. It is now claimed that $\exp$ is injective on the open subset $U_\varepsilon \defeq \setc{X_p \in NO}{\rho_p(0, X_p) < \varepsilon}$ for some $\varepsilon > 0$. Suppose the contrary. Then there exist two sequences $(p_i, X_i) \neq (q_i, Y_i)$ in $NO$ such that $\rho_{p_i}(0, X_i)$ and $\rho_{q_i}(0, Y_i)$ converge to zero and such that $\exp(p_i, X_i) = \exp(q_i, Y_i)$ for all $i$. A contradiction is derived by noticing that the points $p_i$ and $q_i$ eventually lie close to each other. Hence $(p_i, X_i)$ and $(q_i, Y_i)$ are finally contained in some $U_\delta$ on which $\exp$ is injective. 
	In order to show $d(p_i, \exp(p_i, X_i)) \to 0$, choose a sequence $g_i \in G$ such that $\Upsilon(g_i, p_i) = m_0$ for all $i$ (such group elements clearly exist since $p_i$ was assumed to lie on the orbit). By equivariance, $\rho_{m_0}(0, (\Upsilon_{g_i})'X_i) = \rho_{p_i}(0, X_i) \to 0$ holds and implies that $(\Upsilon_{g_i})'X_i$ converges to $0$. Thus,
	\begin{equation}\begin{split}
		0 = d(m_0, m_0) \longleftarrow \, &d(m_0, \exp(m_0, (\Upsilon_{g_i})'X_i)) \\
					&= d( \Upsilon_{g_i} p_i, \Upsilon_{g_i} \exp(p_i, X_i)) = d(p_i, \exp(p_i, X_i)).
	\end{split}\end{equation}
	Now, $d(p_i, q_i) \leq d(p_i, \exp(p_i, X_i)) + d(\exp(q_i, Y_i), q_i)$ converges to $0$ and thus the points $p_i$ and $q_i$ are close to each other for sufficiently large $i$. Since $X_i$ and $Y_i$ are nearly zero, $(p_i, X_i)$ and $(q_i, Y_i)$ are eventually contained in $U_\delta (p_i)$ for some $\delta > 0$. But $\exp$ is injective on $U_\delta (p_i)$ for sufficiently small $\delta$ and hence $(p_i, X_i) = (q_i, Y_i)$ in contradiction with the original assumption.

	Finally, the splitting submanifold $\exp(T_{m_0}M \cap U_\varepsilon)$ is the desired slice $S$ at $m_0$. Indeed all defining properties hold true:
	\begin{enumerate}
		\item The exponential map as well as the normal bundle are equivariant and $U_\varepsilon$ is invariant under the action of $G$. Hence $S$ is invariant under the action of the stabilizer $G_{m_0}$.
		\item By equivariance, $(\Upsilon_g S) \cap S \neq \emptyset$ implies $g \in G_{m_0}$.
		\item Let $\chi: G/G_{m_0} \supseteq V \to G$ be a local section around the identity coset. Every $s \in S$ is of the form $\exp(X_{m_0})$ for some $X_{m_0} \in T_{m_0}M$. Therefore, the map $\chi^S: V \times S \to M, \chi^S(\equivClass{g}, s) \defeq \Upsilon (\chi(\equivClass{g}), s)$ can be rewritten as
		\begin{equation}
			\chi^S(\equivClass{g}, s) = \Upsilon_{\chi(\equivClass{g})} s = \Upsilon_{\chi(\equivClass{g})} \exp(X_{m_0}) = \exp((\Upsilon_{\chi(\equivClass{g})})' X_{m_0}),
		\end{equation}
		which, as a composition of diffeomorphisms, is itself a diffeomorphism. \qedhere
	\end{enumerate}
\end{proof}
\chapter{Classical field theory} \label{sec:cft}

This chapter develops a covariant formulation of classical field theory, which shares many features with finite-dimensional symplectic geometry. In particular, the covariant field-theoretic analogues of the symplectic form and of the momentum map play a prominent role. The approach followed in this work is inspired by the variational bicomplex (see \parencite{Anderson1992,Zuckerman1987} and particularly \parencite{DeligneFreed1999}). However, no reference to jet bundle techniques is needed in the new formulation presented in this chapter and only the useful concept of the bicomplex remains. The mathematical rigorousness is provided by the differential geometric framework which was presented in the previous chapters.

\section{Lagrangian dynamics}
The variational principle represents the natural starting point for a Lagrangian description of classical field theory and thus is shortly reviewed now. At this motivational stage, no mathematical rigour is implied and indeed the first task constitutes in providing a mathematically sound framework for these manipulations. Let $M$ be a compact, finite-dimensional manifold representing spacetime and $L$ a first-order Lagrangian functional density on some function space $\secspace{F}$ over $M$. The equation of motion arises by varying the action
\begin{align}
 	0 = \diF \int_M L(\phi) 	&= \int_M \left(\diFF{L}{\phi} \diF \phi + \diFF{L}{(\dif \phi)} \diF (\dif \phi) \right) \\
 								&= \int_M \left[ \left(\diFF{L}{\phi} - \dif \diFF{L}{(\dif \phi)}\right) \diF \phi  + \dif\left( \diFF{L}{(\dif \phi)} \diF \phi \right) \right].
\end{align}
The integral over the total derivative vanishes by specifying suitable boundary conditions and hence the Euler-Lagrange equation are encoded in the first term. Note that the integral only converges for compact $M$ or by supplying additional conditions on the behaviour of the fields $\phi$ at infinity (or at the boundary). Since both possibilities influence the physical predictions of the model,
a framework which works without such additional assumptions is aspired. Therefore it stands to reason to completely omit integration from the derivation of the Euler-Lagrange equation. In fact, a closer inspection of the above calculation shows that the value of the integral is nowhere explicitly required and only the technique of integration by parts is essential. The latter can be emulated by subtracting the total derivative of $\theta \defeq \diFF{L}{(\dif \phi)} \diF \phi$. Indeed
\begin{equation}\begin{split}
 	\diF L - \dif \theta &= \diFF{L}{\phi} \diF \phi + \diFF{L}{(\dif \phi)} \diF (\dif \phi) - \dif\left(\diFF{L}{(\dif \phi)}\right) \wedge \diF \phi - \diFF{L}{(\dif \phi)} \dif (\diF \phi) \\ &= \left( \diFF{L}{\phi} - \dif\left(\diFF{L}{(\dif \phi)}\right)\right) \wedge \diF \phi
\end{split}\end{equation} 
is the equation of motion. This observation is fundamental for the following exposition and is the starting point for the rigorous formulation.

As a last ingredient, the structure of the function spaces involved has to be clarified. This was already prepared in \autoref{sec::sectionSpace}. In particular, let $M$ be a $n$-dimensional, not necessarily compact manifold and $F \to M$ a fibre bundle over $M$. Endow the space of sections $\secspace{F}$ with the already discussed locally convex manifold structure, which is tame Fréchet if the base manifold $M$ is compact. The product $\secspace{F} \times M$ will be of profound importance for the following discussion and thus the reader is encouraged to recall the properties of vector fields and differential forms on product manifolds, see \autoref{ex::locallyconvexmanifolds:tangentProductManifold} and \autoref{ex::locallyconvexmanifolds:diffFormsProductManifold}. To conform with the above notation, the partial exterior differential with respect to the function space will we denoted by $\diF$ whereas $\dif$ indicates differentiation in the direction of $M$\footnote{Recall that the partial exterior differential with respect to the second factor of the product manifold $\secspace{F} \times M$ can be identified with the normal de Rham differential on $M$ after evaluating all $\secspace{F}$-components. This fact justifies the notation $\dif$ for the partial exterior with respect to the $M$-factor.}. Note that these two differentials anti-commute as a consequence of \cref{lem::lcm:AntiCommutativPartialExteriorDifferential}, in contrast to the above assumed commutativity. Finally, denote the total derivative by $\Dif = \dif + \diF$. Now all the necessary concepts and notations are assembled for the following abstraction of the variational principle.

\begin{defn}
	Let $M$ be a $n$-dimensional manifold and $\secspace{F}$ the section space of a fibre bundle over $M$. Let $(\diF\,, \dif\,)$ denote the partial differentials in the bicomplex of forms $\diffformBi{p}{q}{\secspace{F} \times M}$. A \emphDef{total Lagrangian} $\mathcal{L} = L + \theta$ is a $n$-form on $\secspace{F} \times M$ consisting of a $(0,n)$-form $L$ called \emphDef{Lagrangian density}, and a $(1,n-1)$-form $\theta$ called \emphDef{variational form}. The \emphDef{Euler-Lagrange form} is defined as the $(1,n)$-form\footnote{The different sign convention for $\dif \theta$ is a consequence of the anti-commutativity of $\dif$ and $\diF$.} 
	\begin{equation}
		E_{\mathcal{L}} \defeq (\Dif \mathcal{L})^{1,n} = \diF L + \dif \theta.
	\end{equation}
	The space of solutions $\secspace{S}_{\mathcal{L}}$ consists of those $\phi \in \secspace{F}$ on which $E_{\mathcal{L}}$ vanishes, that is, $(E_{\mathcal{L}})_{\phi, m} = 0$ for all $m \in M$. If only such solutions are considered, one speaks of an on-shell analysis. In contrast, off-shell accounts also for non-solutions.
\end{defn}

In this context, the manifold $M$ plays the role of spacetime (or solely time) and the infinite-dimensional manifold $\secspace{F}$ represents the kinematically realisable states. Thus $F$ is often called the configuration bundle. The solution space $\secspace{S}_{\mathcal{L}} \subseteq \secspace{F}$ corresponds to dynamically admissible states. To simplify notation, the space of $(p,q)$-forms on $\secspace{F} \times M$ is often shortly denoted by $\diffformBi{p}{q}$. 

Before this calculus is demonstrated on concrete examples, a remark about the origins of and history behind this definition is in order.
\begin{remark}
	In the eighties, the so-called inverse problem of variational calculus gained wide attention. The aim of this research program was to clarify which differential equations originate from a variational description. It was mainly expedited by \textcite{Tulczyjew1980, Vinogradov1978, Takens1979, Anderson1992}. Their attention focused primarily on the infinite jet bundle, which is well-founded by a theorem of \textcite{Peetre1959} stating that every local differential operator factors through some jet bundle. One recognizes that $J^\infty F$ possesses a natural connection and a closely related differential bicomplex. \textcite{Zuckerman1987} discovered that the pullback of this bicomplex under the infinite jet prolongation
	\begin{equation}
		j^\infty: \secspace{F} \times M \to J^\infty F, \qquad (\phi, m) \mapsto j^\infty_m \phi  	
	\end{equation}
	equals the bicomplex structure induced by the product $\secspace{F} \times M$. This observation served as the starting point of the analysis of classical field theory by \textcite{DeligneFreed1999}. The above definition generalizes these ideas to the case where locality and jet bundles are not part of the underlying assumptions. 
\end{remark}

\begin{example}[Classical mechanics]
	Classical mechanics is described by choosing $M = \R$ as the time axis and the trivial bundle $F= Q \times \R \to \R$. As usually, the manifold $Q$ represents the configuration space and sections of $F$ are identified with smooth curves $t \mapsto q(t)$ in $Q$. In order to conform with the usual notation, one has to consider the infinite jet bundle $J^\infty (Q \times \R)$.~\footnote{In the present setting, the jet bundle $J^\infty (Q \times \R)$ decomposes into $T^\infty Q \times \R$, where the infinite tangent manifold $T^\infty Q$ is defined as the inverse limit of the sequence of finite-order tangent bundles.} Choosing fibred coordinates $(q^i, t)$ on $Q \times R$ induces coordinates $(q^i_{(k)}, t)$ on the jet bundle. They are explicitly defined by
	\begin{equation}
		q^i_{(0)}(j^\infty_t q) \defeq q^i(t) \qquad \text{and} \qquad q^i_{(k)}(j^\infty_t q) \defeq \difpp{q}{t}{k}(t) \text{ for } k\geq 1,
	\end{equation}
	where $q \in \secspace{F}$ is a smooth curve in $Q$. The pullback of these coordinate functions by the infinite jet prolongation yields a coordinate frame for local differential forms on $\secspace{F} \times M$. By local forms one understands forms which are pulled back from $J^\infty F$. If the total derivative is defined as %
	\begin{equation}
		\diff{}{t} = \difp{}{t} + \sum_{k=0}^\infty q^i_{(k+1)} \difp{}{q^i_{(k)}},
	\end{equation}
	then for a local function $f = f(q^i_{(k)}, t)$ the differential evaluates to
	\begin{equation}
		\dif f = \diff{f}{t} \dif t \quad \text{and} \quad \diF f = \sum_{k=0}^\infty \difp{f}{q^i_{(k)}} \diF q^i_{(k)}.
	\end{equation}
	In particular, the following identities hold
	\begin{align}
		&\dif t = \Dif t								& &\diF t = 0 \\
		&\dif q^i_{(k)} = q^i_{(k+1)} \dif t 	& &\diF q^i_{(k)} = (\Dif - \dif\,) q^i_{(k)} = \Dif q^i_{(k)} - q^i_{(k+1)} \dif t.
	\end{align} 
	If applicable, the alternative notation $q^i, \dot{q}^i, \ddot{q}^i, \dotsc$ is used in favour of $q^i_{(k)}$.

	Now let $L$ be a smooth time dependent function on the tangent bundle $TQ$ and promote it to a Lagrangian density $L \in \diffformBi{0}{1}$ by identifying $L(q^i, v^i, t) \rightleftharpoons L(q^i, \dot{q}^i, t) \dif t$. Furthermore, the variational form is in local coordinates defined by 
	\begin{equation}
		\theta = \difp{L}{\dot{q}^i} \diF q^i.
	\end{equation}
	Now the Euler-Lagrange form evaluates to
	\begin{equation} \begin{split}
		E_{\mathcal{L}} &= \diF L + \dif \theta \\
						&= \left(\difp{L}{q^i} \diF q^i + \difp{L}{\dot{q}^i} \diF \dot{q}^i\right) \wedge \dif t + \diff{}{t}\left(\difp{L}{\dot{q}^i} \right) \dif t \wedge \diF q^i + \difp{L}{\dot{q}^i} \dif \diF q^i \\
						&= \left(\difp{L}{q^i} \diF q^i + \difp{L}{\dot{q}^i} \diF \dot{q}^i\right) \wedge \dif t + \diff{}{t}\left(\difp{L}{\dot{q}^i} \right) \dif t \wedge \diF q^i - \difp{L}{\dot{q}^i} \diF (\dot{q}^i \dif t)\\
						&= \left( \diff{}{t}\left(\difp{L}{\dot{q}^i} \right)  - \difp{L}{q^i} \right) \dif t \wedge \diF q^i.
	\end{split}\end{equation}
	Note that the variational form $\theta$ is chosen in such a way that the term proportional to $\diF \dot{q}^i$ is compensated and thus indeed plays the role of integration by parts. Furthermore, it has the same coordinate representation as the usual canonical (or tautological) $1$-form on $T^*Q$ pulled back to the tangent bundle by the Legendre-transformation. This observation allows for the conclusion that the differential of $\theta$ should correspond to the symplectic form, and indeed
	\begin{equation}
		\omega = \diF \theta = \difpm{L}{q^j}{\dot{q}^i} \diF q^j \wedge \diF q^i + \difpm{L}{\dot{q}^j}{\dot{q}^i} \diF \dot{q}^j \wedge \diF q^i  	
	\end{equation}  
	`equals' the symplectic form on the tangent bundle. This relation will serve as the guideline for defining a counterpart to the symplectic structure in the field theoretic context. However, note that $\omega$ is far from being nondegenerate since it is only sensitive to first-order variations.
\end{example}
The previous example straightforwardly generalizes to first-order field theories by considering a general spacetime manifold $M$ instead of only the time axis $\R$ as the base manifold and replacing the index $k$ by an appropriate multi-index. However, this viewpoint will not be further investigated. Instead, the relativistic aspects of the theory are emphasized by giving a fully covariant example.  
\begin{example}[Complex Klein-Gordon field] \label{ex::classicalFieldTheory:KleinGordonEulerLagrange}
	Let $(M, \eta)$ be $4$-dimensional space time with its Minkowski metric and $F = \C \times M$ the trivial bundle over $M$. Its sections $\phi \in \secspace{F}$ are identified with complex-valued functions on $M$. Moreover, extend the Hodge star operator from $M$ to the product manifold $\secspace{F} \times M$ by ignoring the $\secspace{F}$-component as described in \autoref{sec::differentialgeometry:hodgeStarProductManifold}. Note that the commutative law $\alpha \wedge \star \beta = \beta \wedge \star \alpha$ is modified to
	\begin{equation} \label{eq::interal_cft:wedgeHodgeCommutative}
		\alpha \wedge \star \beta = (-1)^{\#\alpha \, \cdot \, \#\beta \,-\, \#_M \alpha \, \cdot \, \#_M \beta} \beta \wedge \star \alpha,
	\end{equation}
	where the sharp denotes the total degree of the differential form and the subindexed version refers to the $M$-degree (see \autoref{prop::differentialgeometry:hodgeStarProductManifoldWedgeCommutative} for the derivation).

	The Lagrangian system for the complex Klein-Gordon field is specified by
	\begin{align}
		L 		&= \frac{1}{2} \dif \bar\phi \wedge \star \dif \phi - \star \frac{\mu^2}{2} \phi \bar\phi,\\
		\theta 	&= \frac{1}{2} (\diF \phi \wedge \star \dif \bar \phi + \diF \bar \phi \wedge \star \dif \phi).
	\end{align}
	It should be remarked that, here and in the following, evaluation at some appropriate points of $\secspace{F} \times M$ is understood and thus $\phi$ actually stands for the evaluation map $(\phi, m) \mapsto \phi(m)$. The exterior differentials are given by:
	\begin{align}
		\diF L 		&= \frac{1}{2} (\diF \dif \bar\phi \wedge \star \dif \phi - \dif \bar\phi \wedge \star \diF \dif \phi) - \frac{\mu^2}{2} (\diF\bar\phi \wedge \star \phi + \diF\phi \wedge \star \bar\phi),\\
		\dif \theta	&= \frac{1}{2} (\dif \diF \phi \wedge \star \dif \bar \phi - \diF \phi \wedge \dif \, \star \dif \bar \phi + \dif \diF \bar \phi \wedge \star \dif \phi - \diF \bar \phi \wedge \dif \, \star \dif \phi).
	\end{align}
	Thereby, the terms containing mixed derivatives cancel each other upon summation (using the above mentioned property \eqref{eq::interal_cft:wedgeHodgeCommutative} of the Hodge dual) and thus the Euler-Lagrange form results in the Klein-Gordon equation:
	\begin{align}
	 	E_{\mathcal{L}} &= \diF L + \dif \theta \\
	 					&= - \frac{1}{2} \left( \diF \phi \wedge (\dif \, \star \dif \bar \phi + \star \mu^2 \bar\phi) + \diF \bar \phi \wedge (\dif \, \star \dif \phi + \star \mu^2 \phi) \right).
	\end{align} 
	Moreover, the $\diF\,$-derivative of the variational form yields the usual expression for the symplectic form,
	\begin{equation}
		\omega = \diF \theta = - \frac{1}{2} (\diF \phi \wedge \star \diF \dif \bar \phi + \diF \bar \phi \wedge \star \diF \dif \phi). 
	\end{equation} 
\end{example}

\section{Symplectic systems}
The examples in the previous chapter suggest strongly that the functional derivative of the variational form results in the symplectic density. But instead of solely considering $\omega = \diF \theta$, it will be convenient to focus attention on the total differential of the Lagrangian, 
\begin{equation}
	\Dif \mathcal{L} = \omega + E_{\mathcal{L}}.
\end{equation} 
In more general terms, the equation of motion need not arise from a Lagrangian description and hence Euler-Lagrange operator is not necessarily of the form $E_{\mathcal{L}} = \diF L + \dif \theta$.
\begin{defn}
	A \emphDef{dynamical symplectic structure} is a $(n+1)$-form $\Omega$ on $\secspace{F} \times M$ such that it is $\Dif\,$-closed, i.e. $\Dif\Omega=0$, and decomposes into 
	\begin{equation}
		\Omega = \omega + E.
	\end{equation}  
	Here $\omega$ is a $(2,n-1)$-form called \emphDef{subsymplectic density} and $E$ is a $(1,n)$-form encoding the \emphDef{equation of motion}. The solution locus $\secspace{S}_\Omega \subseteq \secspace{F}$ is defined by the condition $E_{\restriction \secspace{S}_\Omega \times M} = 0$. 
\end{defn}
\newpage
\begin{remarks}
	\item Note that no nondegeneracy conditions are required and thus, strictly speaking, the notion `symplectic' is not legitimate. Nevertheless, many important constructions of classical symplectic geometry carry over and this close relationship is underlined by using the same `language'. Moreover, weakly nondegenerate forms are the only reasonable concept beyond Banach manifolds. However, such a weak condition yields only a small advantage and is even too restrictive in concrete examples. The notion `subsymplectic' spells out that $\omega$ is solely a part of the more fundamental symplectic form $\Omega$.

	\item The requirement of being closed is equivalent to the following conditions,
		\begin{equation}
		 	0 = \begin{cases}
		 		\diF \omega,\\
		 		\dif \omega + \diF E.
		 	\end{cases}
		\end{equation}
	The first line expresses the usual closedness of the subsymplectic density and the second condition can be regarded as the conservation of symplecticity on solutions. 

	\item In particular, every Lagrangian $\mathcal{L}$ defines a dynamical symplectic structure by $\Omega = \Dif \mathcal{L}$ and in this sense is an `exact' system. Note that the total Lagrangian can be replaced by $\mathcal{L} \mapsto \mathcal{L} + \Dif \lambda$ for $\lambda \in \diffformBi{0}{n-1}$ without changing the symplectic form $\Omega$. Under such a transformation the components change as follows:
	\begin{align}
		L &\mapsto L + \dif \lambda,\\
		\theta &\mapsto \theta + \diF \lambda.
	\end{align}
	In other words, the usual indeterminacy of $L$ is absorbed in $\theta$.

	\item In the standard approach to classical mechanics, the solution operator $E$ does not occur because the tangent bundle can be interpreted as the parameterization of the solution space by initial conditions and hence amounts to a purely on-shell analysis. For field theories one cannot hope to get an analogous description of the solution locus since in general $\secspace{S}$ fails to be a submanifold of $\secspace{F}$.%
	Moreover, an off-shell analysis promises another advantage when considering its application to quantization. The Feynman path integral approach vividly illustrates that quantum corrections arise by including paths which do not fulfil the equations of motion. As the framework presented in this work handles solutions and non-solutions on an equal footing it provides a promising starting point for quantization. 

	Furthermore, the previous ideas suggest a different ansatz to handle the path integral, namely, to completely circumvent the integral in a similar way as the action was avoided at the beginning of this chapter. In fact, the value of the Feynman integral often plays an inferior role and the main interest lies in its derivatives, that is, in the correlation functions. Hence, the concept of the bicomplex $\diffformBi{\bcdot}{\bcdot}$ might avoid the highly nontrivial construction of a measure on the function space $\secspace{F}$. It is subject to further work to pursue these ideas. \label{rem::cft:offShellAsStartingPointOfQuantization}
\end{remarks}

The rest of this section is concerned with the transition of the basic symplectic theory to this setting.
\begin{defn}
	Let $\alpha \in \diffform{n-1}{\secspace{F} \times M}$. A vector field $X_\alpha$ on $\secspace{F} \times M$ is called a \emphDef{Hamiltonian vector field} for $\alpha$ if
	\begin{equation} \label{eq::cft:hamiltonianVectorField}
		X_\alpha \contr \Omega + \Dif \alpha = 0. 
	\end{equation}
\end{defn}
Without a non-degeneracy condition on the symplectic form $\Omega$ the Hamiltonian vector field might not exists or is not uniquely defined. Thus the Poisson brackets are not available. Note that in the present case the $(n-1)$-forms play the role of observables. %
Writing out \cref{eq::cft:hamiltonianVectorField} in components yields the following system: 
\begin{subequations} \label{eq::cft:HamiltionianVFComponentWise}
\begin{align+} 
  	0 &= (X_\alpha \contr E)^{0,n} + \dif \alpha^{0,n-1},\\
	0 &= (X_\alpha \contr E)^{1,n-1} + (X_\alpha \contr \omega)^{1,n-1} + \dif \alpha^{1,n-2} + \diF \alpha^{0,n-1},\\
	0 &= (X_\alpha \contr \omega)^{2,n-2} + \dif \alpha^{2,n-3} + \diF \alpha^{1,n-2},\\
	0 &= \diF \alpha^{2,n-3}.
\end{align+}
\end{subequations}
In the case of classical mechanics ($n=1$) all but the second equation are trivial and the latter reduces to $0 = (X_\alpha \contr \omega)^{1,0} + \diF \alpha^{0,0}$ on-shell. As the examples in \cref{sec::cft:symmetriesMomentumMap} will show, the components of $\alpha$ besides $\alpha^{0,n-1}$ have to be interpreted as a relict stemming from integration by parts. 

Hamiltonian vector fields possess the important property that they preserve $\Omega$, that is,
\begin{equation}
	\difLie_{X_\alpha} \Omega = \Dif (X_\alpha \contr \Omega) + X_\alpha \contr \Dif \Omega = 0.
\end{equation}
Thereby, they form a special subclass of objects which leave the symplectic structure invariant.
\begin{defn}
	A \emphDef{symplectic map} is a map $\psi: (A, \Omega) \to (B, \Xi)$ between dynamical symplectic systems such that $\Omega = \psi^* \Xi$. An action $\Upsilon: G \times (\secspace{F} \times M) \to \secspace{F} \times M$ of a Lie group is called \emphDef{symplectic} if $\Upsilon_g$ is a symplectic map for all $g \in G$.

	A vector field $X$ on $\secspace{F} \times M$ is called a \emphDef{symplectic vector field} if $\difLie_X \Omega = 0$. If $X$ has a flow, then this flow is a symplectic map.
\end{defn}

\section{Symmetries and momentum maps} \label{sec::cft:symmetriesMomentumMap}
The definition of a momentum map needs more care primarily because its direct generalization as a smooth map $J: \secspace{F} \times M \to \liea{g}^*$ is not possible. First of all, the components $J_A$ are only functions on $\secspace{F} \times M$ in this case, and thus do not induce a Hamiltonian vector field by Equation \eqref{eq::cft:hamiltonianVectorField}. This is circumvented by defining $J_A$ to be a $(n-1)$-form. The other problem is of functional analytic nature. Defining the momentum map to have a dual space as codomain is prone to run in complications similar to the ones encountered in the discussion of the cotangent bundle. In particular, the natural pairing $\liea{g}^* \times \liea{g} \to \R$ cannot be smooth. Even if these continuity problems are overcome, one would necessarily leave the category of Fréchet manifolds. Consequently the Nash-Moser theorem, which at the moment is the only promising tool to generalize symplectic reduction to field theories, would no longer be available. To remedy this situation, the following modification of the usual definition is proposed.
\newpage 
\begin{defn}
	Let $\Upsilon$ be a smooth action of a Lie group $G$ on the symplectic manifold $(\secspace{F} \times M, \Omega)$. A map $J: \liea{g} \to \diffform{n-1}{\secspace{F} \times M}$ is called a \emphDef{momentum map} if 
	\begin{equation}
		A^* \contr \Omega + \Dif J_A = 0
	\end{equation}   	
	holds for all $A \in \liea{g}$; with the notation $J_A \defeq J(A) \in \diffform{n-1}$. If additionally the identity
	\begin{equation}
	 	\Upsilon_g^* J_A = J_{\AdjAction(g^{-1}) A}
	\end{equation} 
	is fulfilled for all $A \in \liea{g}$, then $J$ is called an \emphDef{equivariant momentum map}.
\end{defn}    
The Killing vector field $A^*$ is by construction the Hamiltonian vector field of the component $J_A$ of the momentum map. In particular, the component-wise discussion of \cref{eq::cft:HamiltionianVFComponentWise} also applies to the present situation.

If the symplectic system originates from a Lagrangian description which is invariant under the action, then there is a canonical way to construct an equivariant momentum map. The reader is urged to compare the following proposition with the standard result of an action on the invariant Lagrangian system $(TQ, L)$.
\begin{proposition} \label{prop::classicalFieldTheory:momentumMapForInvariantLagrangian}
	Let $G$ be a Lie group with smooth exponential map $\exp$ and $\Upsilon$ a $G$-action. Let $\mathcal{L}$ denote the total Lagrangian and assume $\Upsilon^*_g \mathcal{L} = \mathcal{L}$. Then 
	\begin{equation} \label{eq::cft:definitionMomentumMap}
		J_A \defeq A^* \contr \mathcal{L}, \quad A \in \liea{g}
	\end{equation}
	defines an equivariant momentum map for $\Upsilon$ with respect to the symplectic form $\Omega = \Dif \mathcal{L}$. 
\end{proposition}
\begin{proof}
	Since $G$ has an exponential map, every Killing vector field $A^*$ possesses a flow given by $(m, t) \mapsto \Upsilon(\exp(tA), m)$. Thus invariance of $\mathcal{L}$ under $\Upsilon$ implies $\difLie_{A^*} \mathcal{L} = 0$. But 
	\begin{equation}
	 	0 = \difLie_{A^*} \mathcal{L} = A^* \contr \Dif \mathcal{L} + \Dif(A^* \contr \mathcal{L}) = A^* \contr \Omega + \Dif J_A
	\end{equation} 
	is just the defining equation for the momentum map. Equivariance follows from the following calculation using the equivariance $\Upsilon'_g A^* = (\AdjAction(g) A)^*$: %
	\begin{equation}
		J_A = A^* \contr \mathcal{L} = A^* \contr (\Upsilon^*_g \mathcal{L}) = \Upsilon^*_g \Bigl((\Upsilon'_g A^*) \contr \mathcal{L} \Bigr) = \Upsilon^*_g \Bigl((\AdjAction(g) A)^* \contr \mathcal{L} \Bigr) = \Upsilon^*_g J_{\AdjAction(g) A}. 
	\end{equation}
\end{proof}

\begin{remark}
	To study initial value problems and the time evolution of fields, covariance of the description has to be broken and a time axis has to be introduced. There are two major ways to achieve this:
	\begin{thmenumerate}
		\item In the traditional instantaneous formalism, spacetime is split into $\Sigma \times T$, where $\Sigma$ is a spatial Cauchy surface and $T$ represents a time interval. Next, all differential forms of bidegree $(p, n-1)$ are integrated over $\Sigma$ to result in a $p$-form on $\secspace{F}$. In particular, $\omega$ yields a $2$-form and all observables of bidegree $(0, n-1)$ become functions. Thus one resides in the usual symplectic landscape with its time evolution picture. The only difference lies in the fact that the state space is infinite-dimensional. 

		But the just described procedure comes with functional analytic and also conceptual problems. On the mathematical side, the Cauchy surface $\Sigma$ is required to be compact in order to guarantee the existence of the integrals over $\Sigma$. One can try to lift this condition by restricting the behaviour of the fields at infinity instead. But in that case, the independence of the integral from the slice delicately depends on the imposed asymptotic conditions \parencite{HollandsMarolf2007}. Furthermore, the whole method only makes sense for on-shell variations, that is, for variations fulfilling the linearised equations of motion. Therefore, one is forced to consider the locus $\secspace{S}$, which is not a manifold in general. Moreover, the possible advantages of an off-shell formulation for quantization described in \cref{rem::cft:offShellAsStartingPointOfQuantization} are nullified. From a conceptual viewpoint, the manifest breaking of covariance at such an early point is a severe shortcoming. In particular, symmetries which move the Cauchy surface cause difficulties in the definition of an appropriate conserved current. In fact, the diffeomorphism symmetry of general relativity in contrast to the gauge invariance of Yang-Mills theory owes much of its complexity to this pitfall \parencite{LeeWald1990}.

		\item The present framework facilitates a `minimal-invasive' and algebraic approach. Consider a symplectic manifold $(\secspace{F} \times M, \Omega)$ and an action $\Upsilon$ of some symmetry group $G$ admitting a momentum map $J$. Note that $G$ combines symmetries of the field degrees of freedom and spacetime symmetries. Under the assumption that the Lie algebra $\liea{g}$ contains an element $A_H$ which can be seen as the generator of time evolution, one arrives at the following Hamiltonian picture. The $(n-1)$-form $\mathcal{H} \equiv J_{A_H}$ plays the role of the Hamiltonian density. 

		Note that Equation \eqref{eq::cft:definitionMomentumMap} obviously can be rewritten as $(A^*_H \contr \omega + \Dif \mathcal{H}) + A^*_H \contr E = 0$. Therefore, the equation
		\begin{equation}
			A^*_H \contr \omega + \Dif \mathcal{H} = 0
		\end{equation}
		is equivalent to the vanishing of $A^*_H \contr E$ and hence it is the field theoretic counterpart to the \emphDef{Hamilton's equations}. 

		This procedure clearly has the advantage to break covariance where it has to be broken while leaving the remaining covariant description intact. In particular, the other components of the momentum map are unaffected. Therefore, spacetime-symmetries should be handled more gracefully. The acid test of general relativity is subject to further work. \qedhere
	\end{thmenumerate}
\end{remark}

Noether's theorem would be nicely captured in $\difLie_{A^*} J_B = J_{\commutator{B}{A}}$, since in that case commutation of $B \in \liea{g}$ with $A_H$ implies that the momentum $J_B$ is conserved along time evolution. Due to continuity issues of the assignment $A \mapsto J_A$, only the following weaker form can be proven.
\begin{proposition}[Noether's theorem]
	Let $G$ be a Lie group with smooth exponential map and $\Upsilon$ a $G$-action with equivariant momentum map $J$. If
	\begin{equation}
		\AdjAction(\exp(tA)) B = B
	\end{equation}
	holds for small $t$ and some $A, B \in \liea{g}$, then $J_B$ is conserved with respect to the evolution of $A$, that is, $\difLie_{A^*} J_B = 0$.
\end{proposition}   
\newpage
\begin{proof}
	\begin{equation}
		\difLie_{A^*} J_B = \diffAt{}{t}{t=0} \Upsilon^*_{\exp(-tA)} J_B = \diffAt{}{t}{t=0} J_{\AdjAction(\exp(tA)) B} = \diffAt{}{t}{t=0} J_{B} = 0. 
	\end{equation}
\end{proof}

In preparation for concrete examples, symmetries induced from spacetime transformations are considered now.
\begin{remark}[Lifted action] \label{rem::classicalFieldTheory:liftedAction}
	Every left action $\psi:G \times M \to M$ of a (possible infinite-dimensional) Lie group $G$ on the base manifold $M$ lifts to the left action
	\begin{equation}
		\Upsilon: G \times (\secspace{F} \times M) \to \secspace{F} \times M, \quad (g, \phi, m) \mapsto (\phi \circ \psi_{g^{-1}}, \psi_g (m)).
	\end{equation}
	One can prove that this map is smooth \parencite[Lemma 2.2.25]{Wockel2006}. This class of symmetries represents the generalisation of point transformations.
\end{remark}

The following examples are used to illustrate the previous concepts.
\begin{example}[Klein-Gordon translation symmetry] \label{ex::classicalFieldTheory:KleinGordonTranslationSymmetry}
  	This example continues the discussion of the complex Klein-Gordon field by analysing the translational symmetry. Recall from \autoref{ex::classicalFieldTheory:KleinGordonEulerLagrange} that the Lagrangian and the symplectic system are specified by
  	\begin{align}
  		\mathcal{L}	&= \frac{1}{2} (\Dif \bar\phi \wedge \star \dif \phi - \star \mu^2 \phi \bar\phi + \diF \phi \wedge \star \dif \bar \phi),\\
	 	\omega 		&= - \frac{1}{2} (\diF \phi \wedge \star \diF \dif \bar \phi + \diF \bar \phi \wedge \star \diF \dif \phi),\\
  		E &= - \frac{1}{2} \left( \diF \phi \wedge (\dif \, \star \dif \bar \phi + \star \mu^2 \bar\phi) + \diF \bar \phi \wedge (\dif \, \star \dif \phi + \star \mu^2 \phi) \right).
  	\end{align}
  	Now consider the action of $(\R^n, +)$ on $n$-dimensional Minkowski-spacetime $M$ by translation:
  	\begin{equation}
  		\psi: \R^n \times M \to M, \quad (x, m) \mapsto m+x.
  	\end{equation}
  	According to \autoref{rem::classicalFieldTheory:liftedAction} this lifts to an action
  	\begin{equation}
  		\Upsilon: \R^n \times (\secspace{F} \times M) \to \secspace{F} \times M, \quad (x, \phi, m) \mapsto (\phi \circ \psi_{-x}, \psi_{x}(m)).
  	\end{equation}
  	Apparently, this action leaves the evaluation map invariant, and therefore also the total Lagrangian and the symplectic structure\footnote{Recall that in the above formulae $\phi$ is just a place holder for the evaluation map.}. To calculate the momentum map associated with $\Upsilon$ the effect of the Killing vector fields on coordinate functions has to be known. After identifying the Lie algebra $\liea{g}$ with $\R^n$ one gets
  	\begin{align}
  		(A^* \contr \dif x^\mu)_{\phi, m} &= A^\mu,\\
  		(A^* \contr \diF (\ev))_{\phi, m} &= - (A^* \contr \dif (\ev))_{\phi, m}.
  	\end{align} %
  	The latter equation is a consequence of $0 = \difLie_{A^*} \ev = A^* \contr \Dif (\ev)$. 

	The calculation of the momentum map is made particularly clear by restricting to the case $n=2$ and using coordinates $(t,x)$ on $M$. Then the Hodge operator acts as follows:
	\begin{align}
		\star 1 = - \dif t \wedge \dif x && \star \dif t = \dif x \\
		\star (\dif t \wedge \dif x) = 1 && \star \dif x = \dif t.
	\end{align}
	Since the total Lagrangian $\mathcal{L}$ is invariant under $\Upsilon$, \autoref{prop::classicalFieldTheory:momentumMapForInvariantLagrangian} yields the momentum map
	\begin{equation}\begin{split}
		2 \cdot J_A &= 2 \cdot (A^* \contr \mathcal{L}) \\
			&= - \Dif \bar\phi \wedge (A^* \contr \star \dif \phi) - A^* \contr \star \mu^2 \phi \bar\phi + (A^* \contr \diF \phi) \star \dif \bar \phi - (A^* \contr \star \dif \bar \phi) \diF \phi \\
			&= \left( (\phi_t \bar\phi_t + \phi_x \bar\phi_x + \mu^2 \phi \bar\phi) \dif x + (\phi_x \bar\phi_t + \phi_t \bar\phi_x) \dif t + (\phi_x \diF \bar\phi + \bar\phi_x \diF \phi)\right) A^t \\
			&\quad + \left( (\phi_t \bar\phi_t + \phi_x \bar\phi_x - \mu^2 \phi \bar\phi) \dif t + (\phi_x \bar\phi_t + \phi_t \bar\phi_x) \dif x + (\phi_t \diF \bar\phi + \bar\phi_t \diF \phi)\right) A^x,
	\end{split}\end{equation}
	where the subindex on $\phi$ represents partial derivation with respect to the specified variable. The bidegree $(0, 1)$-component of $J_A$ equals the usual stress-energy tensor (up to a factor $\sfrac{1}{2}$). %
	Thus one is left with discussing the terms of the form $\phi_x \diF \bar\phi$. For this purpose, recall that the standard Hamiltonian density $\mathcal{H} = \phi_t \bar\phi_t + \phi_x \bar\phi_x + \mu^2 \phi \bar\phi$ does not yield the Klein-Gordon equation as Hamilton's equation for the canonical conjugate coordinates $(\bar\phi, \phi_t)$. However, the modified density
	\begin{equation}
		\mathcal{\tilde{H}} = \phi_t \bar\phi_t - \phi_{xx} \bar\phi + \mu^2 \phi \bar\phi,
	\end{equation}
	gives the correct Hamilton's equations. One transforms $\mathcal{H}$ to $\mathcal{\tilde{H}}$ by integrating over it and then using integration by parts. The occurring boundary term is $\phi_x \bar\phi$, which closely resembles $\phi_x \diF \bar\phi$. The latter term indeed cancels the distracting terms in Hamilton's equation (here $A_H=(1,0)$)
	\begin{equation}
		A^*_H \contr \omega + \Dif J_{A_H} = 0.
	\end{equation}
	A straightforward but lengthy calculation shows that this formula is in all its components equivalent to the Klein-Gordon equation. In this sense, the $(1,0)$-component of $J_A$ plays the role of partial integration similar to the variational form $\theta$.
\end{example}  

\begin{example}[$U(1)$-symmetry of Klein-Gordon]
	The treatment of the $U(1)$-symmetry of the Klein-Gordon proceeds analogous to \autoref{ex::classicalFieldTheory:KleinGordonTranslationSymmetry} but exposes the case of a symmetry not induced from an action on the base manifold. Parametrize elements of $U(1)$ as $e^{\I \alpha}$ and consider the action
	\begin{equation}
		\Upsilon: U(1) \times (\secspace{F} \times M) \to \secspace{F} \times M, \quad (\alpha, \phi, m) \mapsto (e^{\I \alpha} \phi, m).
	\end{equation}
	Since the Lagrangian $\mathcal{L}$ contains only combinations of $\phi$ and its complex conjugate, it is invariant under $\Upsilon$. Clearly the Killing vector fields fulfil
	\begin{align}
  		(A_\alpha^* \contr \dif x^\mu)_{\phi, m} &= 0,\\
  		(A_\alpha^* \contr \diF (\ev))_{\phi, m} &= \I \alpha \phi,
  	\end{align}
  	where $\alpha \in \R$ parametrizes $\liea{u(1)} = \I \R$. \Autoref{prop::classicalFieldTheory:momentumMapForInvariantLagrangian} yields the momentum map
  	\begin{equation}
  		J_\alpha = A_\alpha^* \contr \mathcal{L} = A_\alpha^* \contr \theta = \frac{1}{2} A_\alpha^* \contr (\diF \phi \wedge \star \dif \bar \phi + \diF \bar \phi \wedge \star \dif \phi) = \frac{\I \alpha}{2} (\phi \star \dif \bar \phi - \bar \phi \star \dif \phi). 
  	\end{equation}
\end{example}

\chapter{Yang-Mills theory} \label{sec::gaugeTheory}
A gauge theory is a special type of field theory in which the Lagrangian is invariant under the action of the group of gauge transformations. Since they naturally provide all the basic inputs - an infinite-dimensional state space endowed with a Lagrangian invariant under a group action - such physical models serve as prime examples for applying the general theorems from the preceding chapters. In particular, the following discussion will present positive answers to the questions concerning the existence of slices and the momentum map.

A fixed principal bundle $P \to M$ sets out the scenery. The pseudo-Riemannian manifold $(M, g)$ models spacetime and the physical state space is represented by the space of all connections in $P$. In this context, the structure group $G$ of $P$ is called gauge group and encodes the gauge symmetry of the theory. Now, \autoref{prop::liegroup:sliceTheorem} identifies the following steps in order to construct a slice for the gauge transformation action:
\begin{enumerate}[label=\arabic*.] 
	\item Endow the space $\secspace{Conn}$ of connections in $P$ with a tame Fréchet manifold structure.
	\item Represent the group of gauge transformations $\secspace{Gau}$ as a locally exponential and tame Fréchet Lie group. Prove that the action $\secmap{transf}: \secspace{Gau} \times \secspace{Conn} \to \secspace{Conn}$ is tame smooth. 
	\item Show properness of $\secmap{transf}$.
	\item Identify the stabilizer $\secspace{Gau}_\lambda$ of $\lambda \in \secspace{Conn}$ as a tame Fréchet principal Lie subgroup of $\secspace{Gau}$.
	\item There exists a Kodaira-like decomposition of the tangent space $T_\lambda \secspace{Conn} = \img (\secmap{transf}_\lambda)'_e \oplus F_\lambda$ for some closed subspace $F_\lambda \subseteq T_\lambda \secspace{Conn}$.
	\item Locally and near the identity of $\Gau$, the derivative of $\transf_\lambda$ admits a family of tame smooth inverses.
	\item $\Conn$ carries a $\Gau$-invariant locally equivalent, graded Riemannian metric $g^k$ such that the $0$-exponential map restricts to a local diffeomorphism around every point of the zero section of the normal bundle.
\end{enumerate}

This series of tasks will be tackled in sequential order in the present chapter, with a discussion of the momentum map appended at the end. As the next sections will show, the program outlined above can be successfully completed in the case of compact base manifold $M$ and some assumptions on the gauge group $G$. This leads to the main result:
\begin{theorem}
	Let $P \to M$ be a principal $G$-bundle. If $M$ is compact and $G$ of the form $G = G_K \times \R^k$ for some compact Lie group $G_K$ and $k \in \N$, then at every point $\lambda \in \secspace{Conn}$ there exists a slice for the action $\secmap{transf}: \secspace{Gau} \times \secspace{Conn} \to \secspace{Conn}$.
\end{theorem}

Since the eighties, gauge theory and its functional analytic questions attracted the attention of mathematical physicists. In particular, the works of \textcite{Subramaniam1984} and of \citeauthor{AbbatiCirelliEtAl1986} \parencite{AbbatiCirelliEtAl1986,AbbatiCirelliEtAl1989} are noteworthy. Both groups independently proved a slice theorem for the action of the group of gauge transformations in the setting of Fréchet spaces. A similar result was obtained by \textcite{KondrackiRogulski1986} using Sobolev space techniques, see also the recent review \parencite{RudolphSchmidtEtAl2002}. Accordingly, the results presented in this chapter are already known to experts, although the exposition differs from the aforementioned literature at some points. 

\section{Space of connections: The connection bundle}
In this section, the space of connections in a fixed principal $G$-bundle $\pi: P \to M$ will be endowed with a smooth manifold structure by representing connections as sections of an appropriate bundle. For this purpose, the viewpoint of a connection as a horizontal lift mechanism will be advantageous. 

The following exact sequence of vector bundles over $P$ marks the starting point for the further discussion,
\begin{equation}
	\begin{tikzcd}
			0 \arrow{r} & VP \arrow[hook]{r}{\iota} & TP \arrow{r}{\pi'} & \pi^* TM \arrow{r} & 0,
	\end{tikzcd}
\end{equation}
where $\iota: VP \hookrightarrow TP$ denotes the natural inclusion of the vertical tangent bundle. The right action $\Psi$ of $G$ on $P$ lifts to the tangent action $\Psi'_g: TP \to TP$, which leaves the above exact sequence invariant. Thus, factoring yields the $G$-invariant version
\begin{equation} \label{eq::gaugeTheory:connectionAtiyahSequence}
	\begin{tikzcd}
			0 \arrow{r} & \AdjBundle{P} \arrow[hook]{r}{\iota} & \AtiyahBundle{P} \arrow{r}{\pi'} & TM \arrow{r} & 0.
	\end{tikzcd}
\end{equation}
Here the bundle $TP/G$ is denoted by $\AtiyahBundle{P}$ in honour of M. Atiyah, who had introduced the above sequence while discussing complex analytic connections in \parencite{Atiyah1957}. As usual, the adjoint bundle $\AdjBundle{P}$ is the associated bundle $P \times_G \liea{g}$, where $G$ acts via the adjoint representation. It is identified with $VP/G$ via the $G$-equivariant vector bundle isomorphism
\begin{equation}
	P \times \liea{g} \to VP, \quad (p,A) \mapsto A^*_p.	
\end{equation} 

Now a connection is simply defined as a splitting of \eqref{eq::gaugeTheory:connectionAtiyahSequence}, 
\begin{equation} \label{eq::gaugeTheory:connectionAtiyahSequenceSplit}
	\begin{tikzcd}[row sep=tiny]
			0 \arrow{r} & \AdjBundle{P} \arrow[hook]{r}{\iota} & \AtiyahBundle{P} \arrow[white]{d}[black, description]{\verteq}\arrow{r}{\pi'} \arrow[swap, bend right=50]{l}{\alpha} & TM \arrow{r} \arrow[swap, bend right=50]{l}{\lambda} & 0.\\
			& & \AdjBundle{P} \oplus TM & &
	\end{tikzcd}
\end{equation}
Thereby, the Atiyah sequence unifies the different incarnations of connections in the three discriminate ways to split an exact sequence. First, the direct sum representation $\AtiyahBundle{P} =  \AdjBundle{P} \oplus TM$ corresponds to choosing a $G$-invariant horizontal subbundle. Second, a splitting in form of a vertical vector bundle isomorphism $\alpha: \AtiyahBundle{P} \to \AdjBundle{P}$ with $\alpha \circ \iota = \id_{TP}$ is just a transcription of a connection form. But it is the third way that allows a representation as sections of a bundle over $M$. In fact, vertical vector bundle morphisms $\lambda: TM \to \AtiyahBundle{P}$ with $\pi' \circ \lambda = \id_{TM}$ are in a bijective correspondence with sections of the affine \emphDef{connection bundle $CP$}, whose fibres are
\begin{equation}
	C_m P \defeq \setc{\lambda_m: T_mM \to \AtiyahBundle{P}_m }{\text{linear}, \pi'\circ \lambda_m = \id_{T_mM}}.
\end{equation}
The affine space $C_m P$ is modelled on the vector space of linear maps $\sigma_m: T_mM \to \AtiyahBundle{P}_m$ with $\pi'\circ \sigma_m = 0$, that is on $L(T_mM, \AdjBundleOp_m P)$. A map $\lambda_m \in C_m P$ assigns to every tangent vector $Y_m \in T_mM$ its horizontal lift $\lambda_m (Y_m)$. Finally, the space of all connections $\secspace{Conn} \defeq \secspaceEx{CP}$ carries the natural locally convex manifold structure discussed in \autoref{sec::sectionSpace}. In particular, the space $\secspace{Conn}$ is a tame Fréchet manifold for compact base $M$.

\section{Action of the group of gauge transformation}
The present section discusses the group of gauge transformation and its action on the space of connections. Starting with this section, only compact base manifolds $M$ are considered.

It is well known that the group of gauge transformation $\secspace{Gau}$ can be represented as the section space of the associated group bundle $P \times_G G$, where $G$ acts on itself via conjugation. The case of sections of a bundle with groups as fibres was already studied in \autoref{ex::lieGroup:currentGroup}. There, it was shown that the section space carries a natural tame Fréchet Lie group structure with a locally diffeomorphic exponential map, which stems from the pushforward of the exponential map on the fibres. Moreover, the Lie algebra was identified with sections of the bundle created by applying the Lie algebra functor fibrewise; in the present case, this yields $\liea{gau} = \secspaceEx{\AdjBundle{P}}$. That is, the $G$-invariant, vertical vector fields on $P$ constitute the Lie algebra of $\secspace{Gau}$.
\newpage
\begin{remark}
	There exists well-known algebraic isomorphisms of $\secspace{Gau}$ with %
	\begin{align}
		\secspace{Gau}^\# &\defeq \setc{\Phi^\# \in C^\infty(P,P)}{G\text{-equivariant and fibre-preserving}}, \\
		\secspace{Gau}\,\hat{}\;\; &\defeq \setc{\hat\Phi \in C^\infty(P,G)}{G\text{-equivariant}}.
	\end{align}
	But for a non-compact gauge group $G$ it is not clear how to topologize these spaces in an appropriate way. For the compact-open topology, the diffeomorphism group carries no Lie group structure. Endowed with the fine Whitney-topology the ambient groups $\DiffGr(P)$ and $C^\infty(P,G)$ are now Lie groups, however $\secspace{Gau}^\#$ and $\secspace{Gau}\,\hat{}\,$ are only discrete subgroups \parencite{AbbatiCirelliEtAl1986}. If $G$ is compact, then the above groups are tame Fréchet Lie groups diffeomorphic to $\secspace{Gau}$ \parencite{AbbatiCirelliEtAl1986}. 
\end{remark}

In view of the previous remark, in order to study smoothness properties, the usual left action
\begin{equation} \label{eq::gaugeTheory:ActionGauSharpOnConnForm}
	\secspace{Gau}^\# \times \secspace{ConnForm} \to \secspace{ConnForm}, \quad (\Phi^\#, \alpha) \mapsto \tilde{\alpha} \defeq (\Phi^{\# \, -1})^* \alpha
\end{equation}
has to be replaced by an action of $\secspace{Gau}$ on the space of connections, where connections are viewed as sections of $CP$ and not as connection forms. This is now achieved, starting from the well known formula
\begin{equation} \label{eq::gaugeTheory:ActionGauHatOnConnForm}
	\tilde{\alpha}(X_p) = \AdjAction(\hat\Phi(p))\alpha(X_p) - \Theta^R(\hat\Phi' X_p) \qquad X_p \in T_pP,		
\end{equation}
where $\Theta^R$ denotes the right Maurer-Cartan form. Due to \eqref{eq::gaugeTheory:connectionAtiyahSequenceSplit} the horizontal lift of $Y_m$ to a point $p \in P$ in the fibre over $m$ is given by $(\lambda_m (Y_m))_p = X_p - (\alpha(X_p))^*_p$ for $X_p \in T_pP$ projecting on $Y_m$. This equation is independent from the chosen $X_p$ as a direct calculation reveals. Thus, the difference of the horizontal lifts with respect to a transformed connection evaluates to
\begin{equation}
 	\Bigl[\tilde{\lambda}_m (Y_m) - \lambda_m (Y_m)\Bigr]_p = - \Bigl[ \tilde{\alpha}(X_p) - \alpha(X_p)\Bigr]^*_p = - \Bigl[ (\AdjAction(\hat\Phi(p)) - 1) \alpha(X_p) - \Theta^R(\hat\Phi' X_p)\Bigr]^*_p.
\end{equation} 
The aim is now to rewrite this expression as an object living on $TG$. The tangent bundle $TG$ is itself a Lie group. In fact, the group operations arise as the derivative of the corresponding operations on $G$:
\begin{equation}\begin{split}
	\mult_{TG} = \mult'_G&: TG \times TG \to TG, \\
	\inv_{TG} = \inv'_G&: TG \to TG. 
\end{split}\end{equation}
It is well known that $TG$ can be identified with the semi-direct product $G \ltimes \liea{g}$ via the right trivialization $\triv^R(V_g) = (g, \Theta^R(V_g))$. On $G \ltimes \liea{g}$ the group operations correspond to 
\begin{equation}\begin{split}
	(g, A) \cdot (h, B) = (g \cdot h, A + \AdjAction(g) B),\\
	(g, A)^{-1} = (g^{-1}, - \AdjAction(g^{-1}) A).
\end{split}\end{equation}
Furthermore, left translation induces the map $L'_g: TG \to TG$ which acts relative to $\triv^R$ as $L'_g(h, B) = (g \cdot h, \AdjAction(g)B)$. Now the identity
\begin{equation}
	\La'_g (e, A) \cdot (V_g^{-1}) \cdot (e, A)^{-1} = (e, (\AdjAction(g) - 1) A - \Theta^R(V_g)) \qquad \text{for } V_g \in T_gG
\end{equation}
is easily verified. This observation leads to the consideration of the following map
\begin{equation}
	\text{transf}: (TP \times TG) \times_{TM} TP \to TP, \quad (Z_p + A^*_p, V_g, Z_p) \mapsto Z_p - [\La'_g \conjAction_{e, A} (V_g^{-1})]^*_p,
\end{equation}
which descends to the quotient $(TP \times_{TG} TG) \times_{TM} TP/G \to TP/G$, as a lengthy but straightforward calculation verifies. Finally, the desired left action of $\secspace{Gau}$ is expressed as
\begin{equation} \label{eq::gaugeTheory:gaugeTransformation}
	\secmap{transf}: \secspace{Gau} \times \secspace{Conn} \to \secspace{Conn}, \quad (\Phi, \lambda) \mapsto \tilde{\lambda}(Y_m) \defeq \text{transf}(\Phi'Y_m, \lambda(Y_m)).
\end{equation}
It is a tame smooth map since it is the composition of the tangent functor with the pushforward of the map $\text{transf}$. The equivalence with \eqref{eq::gaugeTheory:ActionGauHatOnConnForm} follows from the above mentioned identities and the identification $\Phi'Y_m = \equivClass{X_p, \hat\Phi' X_p}_{TG}$, which results in $A = \alpha(X_p)$.

For later use, the infinitesimal gauge transformation is recorded here,
\begin{equation} \label{eq::gaugeTheory:gaugeTransformationTangent}
 	(\secmap{transf})'_{\Phi, \lambda} (\Xi, \sigma) = \secmap{transf}(\Phi, \sigma) - \nabla_\lambda \Xi,
\end{equation} 
where $\nabla_\lambda: \secspaceEx{\AdjBundle{P}} \to \secspaceEx{T^*M \otimes \AdjBundle{P}}$ is the covariant derivative with respect to $\lambda$ in the associated bundle $\AdjBundle{P}$. Correctness of the first term is easily verified by acknowledging that the map $\text{transf}$ is linear in $Z_p$ for vertical vectors. The covariant derivative $\nabla_\lambda$ manifests itself by replacing $\Phi = \exp(t \Xi)$ in \eqref{eq::gaugeTheory:gaugeTransformation} and subsequently differentiating with respect to $t$.

\section{Properness and stabilizer of the gauge transformation action}
\begin{proposition} \label{prop::gaugeTheory:properActionForCompactGaugeGroup}
	If $G \subseteq GL(r, \R)$ is a compact matrix Lie group, then the gauge transformation action $\secmap{transf}$ on the manifold of connections is proper. 
\end{proposition}
\begin{proof}
	$\Gau$ is metrizable as a $0$-tame Fréchet Lie group by \cref{prop::lcLieGroup:metrisableZeroTame} and in \cref{sec::gaugeTheory:InvariantRiemannianMetric} a path-metric on $\Conn$ will be constructed (independently of the present claim). Hence, by \autoref{prop:lieGroupAction:properEquivalent}, it suffices to show the following: Let $(\Phi_i)$ be a sequence in $\secspace{Gau}$ and $\lambda_i \to \lambda$ a convergent sequence in $\secspace{Conn}$ such that $\secmap{transf}(\Phi_i, \lambda_i)$ converges to $\tilde \lambda$. Then $(\Phi_i)$ contains a convergent subsequence. The proof of this fact consists of the subsequent steps which will be addressed individually in separate lemmas.
	\begin{itemize}
		\item On an open subset, convergence of a subsequence $\Phi_{i_k}$ at one point is enough to ensure uniform convergence on compacta.
		\item Using the first point and compactness of $G$, for an open cover $U_\alpha$ one obtains a family of subsequences $\Phi_{i_k, \alpha}$ converging to maps $\Phi_\alpha: U_\alpha \to P \times_G G$. These combine to a continuous section $\Phi: M \to P \times_G G$, which is the compactly uniform limit of a subsequence $\Phi_{i_k}$.
		\item A local analysis of $\secmap{transf}(\Phi_i, \lambda_i) \to \tilde \lambda$ shows that $\Phi$ is actually smooth and $\Phi_{i_k}$ converges to it in the topology of uniform convergence of all partial derivatives on compacta.   
	\end{itemize}
	This approach is adapted from \parencite[Proposition 2.4]{NarasimhanRamadas1979}, where the special case $M=S^3$ and $G=SU(2)$ is studied.
\end{proof}

The following lemmas build upon the notation introduced in \autoref{prop::gaugeTheory:properActionForCompactGaugeGroup}.
\begin{lemma}
	Let $U \subseteq M$ be an open subset. If there exist a point $p \in U$ and a subsequence $(\Phi_{i_k})$ such that $(\Phi_{i_k}(p))$ converges, then $(\Phi_{i_k})$ converges uniformly on compacta to a set-theoretic section $\Phi: U \to P \times_G G$.  
\end{lemma}
\begin{proof}
	Without loss of generality, the subset $U$ can be assumed to be a chart domain, $\kappa: U \to \R^n$. Relative to this chart (and the induced trivialization) a connection is identified with a $\liea{g}$-valued one-form on $U$ and gauge transformations are smooth maps $U \to G$. The local representation of the action
	\begin{equation} \label{eq::gaugeTheory:properActionForCompactGaugeGroupProof_localRepresentationAction}
		\tilde{\lambda}_m = \AdjAction(\Phi(m)) \lambda_m - (\Phi^* \Theta^R)_m = \Phi(m) \cdot \lambda_m \cdot (\Phi(m))^{-1} - (\dif \Phi)_m \cdot (\Phi(m))^{-1}
	\end{equation}
	can be solved for the differential $\dif \Phi$. After normalizing $\underline{\Phi}_{i_k} \defeq (\Phi_{i_k}(p))^{-1} \cdot \Phi_{i_k}$ and $\underline{\tilde{\lambda}} = (\Phi_{i_k}(p))^{-1} \cdot \tilde{\lambda} \cdot \Phi_{i_k}(p)$, it reads
	\begin{equation} \label{eq::gaugeTheory:properActionForCompactGaugeGroupProof_localRepresentationDerivativeGaugeTrafo}
		(\dif \underline{\Phi}_{i_k})_m = \underline{\Phi}_{i_k}(m) \cdot (\lambda_{i_k})_m - (\underline{\tilde{\lambda}}_{i_k})_m \cdot \underline{\Phi}_{i_k}(m).
	\end{equation}
	Consider now the ordinary linear differential equation
	\begin{align}
		\diff{}{t} h(t) &= h(t) p^j \lambda_j (p t)  -  p^j \underline{\tilde{\lambda}}_j (p t) h(t),\\
		h(0) &= e,
	\end{align}
	where $h$ is a real $r \times r$-matrix, $t \in \R$ and $j$ denotes the components of $p \in \R^n$. The usual existence and uniqueness theory
	yields a unique solution $h_{p, \lambda, \underline{\tilde\lambda}}(t)$, which is differentiable in $t$ and continuous with respect to the parameters $p, \lambda$ and $\underline{\tilde\lambda}$. But $h(t) = \underline{\Phi}_{i_k}(p t)$ is a solution for $\lambda = \lambda_{i_k}$ and $\underline{\tilde\lambda} = \underline{\tilde\lambda}_{i_k}$, and thus depends continuously on them. This continuity, together with the fact that $\lambda_{i_k}$ and $\tilde{\lambda}_{i_k} = \secmap{transf}(\Phi_{i_k}, \lambda_{i_k})$ both converge, yields a uniform limit $\Phi: U \to \mathfrak{M}(r \times r, \R)$, where $\mathfrak{M}(r \times r, \R)$ denotes the space of real-valued $r \times r$-matrices. Since $G$ is closed in $\mathfrak{M}(r \times r, \R)$, this map takes values in $G$.
\end{proof}

\begin{lemma}
	There exists a subsequence $(\Phi_{i_k})$ converging to a smooth map $\Phi: M \to P \times_G G$ with respect to the $C^\infty$-compact-open topology.
\end{lemma}
\begin{proof}
	By compactness of $M$, one can choose a finite cover by open subsets $U_\alpha \subseteq M$. Fix a point $p_\alpha \in U_\alpha$. Since $G$ is compact, the sequence $\Phi_i(p_\alpha) \in G$ contains a convergent subsequence. Thus, the previous lemma ensures a convergent subsequence $\Phi_{i_k}$ on every $U_\alpha$ with limit $\Phi^\alpha: U_\alpha \to P \times_G G$. The subsequence can be chosen to be indexed by the same $i_k$ on overlapping sets (just choose the point $p$ from the intersection $U_\alpha \cap U_\beta$). Hence, the maps $\Phi^\alpha$ automatically agree on overlaps and so define a global section $\Phi: M \to P \times_G G$. As a uniform limit of continuous maps, it is continuous itself. Now, equation \eqref{eq::gaugeTheory:properActionForCompactGaugeGroupProof_localRepresentationDerivativeGaugeTrafo} implies by induction that $\Phi$ is actually smooth and that all partial derivatives converge uniformly on compact subsets. 
\end{proof}

\begin{remark}
	In \autoref{prop::gaugeTheory:properActionForCompactGaugeGroup} the condition of $G$ being a linear group  is not essential for the properness of the action. In fact, \parencite[pp. 75f]{Subramaniam1984} proves the same result for arbitrary compact groups $G$. See also \parencite[Theorem 2.4.9]{KondrackiRogulski1986} in the context of Sobolev spaces, which then carries over to Fréchet spaces as shown in \parencite[Theorem 4.1]{AbbatiCirelliEtAl1989}. In the latter reference, one also finds the slightly more general case of Lie groups of the form form $G = G_K \times \R^k$ for some compact Lie group $G_K$ and $k \in \N$. Although the used techniques are similar to the proof presented here, a detour via the diffeomorphism group seems to be necessary. 
\end{remark}

Properness of the action directly implies that every stabilizer subgroup is compact (see \autoref{prop:lieGroupProperAction:closedOribtMapCompactStabilizer}). In the present case, one can go a step further and characterize $\secspace{Gau}_\lambda$ as the centralizer of the holonomy group of $\lambda$. To see this, consider for a point $p \in P$ the smooth group homomorphism $\Ev_p: \secspace{Gau} \to G$  defined by 
\begin{equation}
	\Phi( \pi(p) ) = \equivClass{(p, \Ev_p(\Phi))}, \qquad \text{for } \Phi \in \secspace{Gau}.
\end{equation}
It is not difficult to see that the restriction of $\Ev_p$ to the stabilizer $\secspace{Gau}_\lambda$ is injective and hence a group isomorphism onto its image, which is identified as the centralizer of the holonomy group of $\lambda$ based at $p$. %
In particular, $\secspace{Gau}_\lambda$ is a closed finite-dimensional subgroup of $\secspace{Gau}$ and hence \autoref{prop::lieGroup:finiteDimSubgroupIsStrongSplittingLieSubgroup} can be applied to identify the stabilizer as a principal Lie subgroup.

\section{Hodge-Kodaira-like decomposition} 
The aim of this section consists in proving the Hodge-Kodaira-like decomposition
\begin{equation} \label{eq::gaugeTheory:HodgeKodairaDecomposition}
	T_\lambda \secspace{Conn} = \img (\secmap{transf}_\lambda)'_e \oplus F_\lambda
\end{equation}
for some closed subspace $F_\lambda \subseteq T_\lambda \secspace{Conn}$. Furthermore, a tame smooth inverse family of $(\secmap{transf}_\lambda)'_e$ is constructed. Both results heavily rely on the theory of elliptic differential operators which therefore plays a fundamental role in the following. The connection between the decomposition \eqref{eq::gaugeTheory:HodgeKodairaDecomposition} and elliptic partial differential equations stems from \eqref{eq::gaugeTheory:gaugeTransformationTangent} and the implied identity $(\secmap{transf}_\lambda)'_e \Xi = - \nabla_\lambda \Xi$. The covariant derivative is a first-order overdetermined elliptic operator. An extensive account of elliptic theory including full proofs is beyond the scope of this thesis and thus the reader is referred to the literature \parencite{Hamilton1982} for details. See also \parencite{Palais1965} for elliptic complexes in the usual setting of Sobolev spaces.

\begin{defn}[{\parencite[p. 155]{Hamilton1982}}]
	Let $M$ be a compact manifold and $E, F$ be vector bundles over $M$. Denote the corresponding section spaces by $\secspace{E}$ and $\secspace{F}$, respectively. Consider a tame smooth family of linear maps $L: \secspaceEx{ L(J^rE, F) } \times \secspace{E} \to \secspace{F}$. The induced map $L_f: \secspace{E} \to \secspace{F}$ is called a \emphDef{linear partial differential operator} of degree $r$ with coefficients $f \in \secspaceEx{ L(J^rE, F) }$. Locally, its action is represented by
	\begin{equation}
		L_f \phi = \sum_{\abs{I} \leq r} f_I \cdot \partial^I \phi.  		
	\end{equation}   	

	Every differential operator $L_f$ induces a homogeneous polynomial $\sigma_f$ of degree $r$ on the cotangent bundle $T^*M$ with values in $L(E,F)$, which is locally given by $\sigma_f \xi = \sum_{\abs{I} = r} f_I \cdot \xi^I$. It is called the \emphDef{principal symbol} of $L_f$. Based on the properties of $\sigma_f$ one distinguishes between:
	\begin{itemize}
	 	\item \emphDef{overdetermined elliptic} operator: injective principal symbol $\sigma_f \xi$.
	 	\item \emphDef{underdetermined elliptic} operator: surjective principal symbol $\sigma_f \xi$.
	 	\item \emphDef{elliptic} operator: bijective principal symbol $\sigma_f \xi$.
	\end{itemize} 
	Here, $\xi \in T^*_mM$ is non-zero but otherwise arbitrary.

	Additionally assume that $M$ carries a Riemannian metric $g$ and $E, F$ are endowed with fibre metrics. The $L^2$-scalar products are given by $\scalarprodr{\phi}{\varphi}_E = \int_M \scalarprod{\phi(m)}{\varphi(m)}_E \dif \vol_g(m)$ and analogously $\scalarprodr{\cdot}{\cdot}_F$. Now, the formal adjoint of $L_f: \secspace{E} \to \secspace{F}$ is the operator $L^*_f: \secspace{F} \to \secspace{E}$ fulfilling
	\begin{equation}
		\scalarprodr{L_f \phi}{\psi}_F = \scalarprodr{\phi}{L^*_f \psi}_E \qquad \text{for every } \phi \in \secspace{E}, \psi \in \secspace{F}. 
	\end{equation}
\end{defn}  
An intrinsic characterisation of the principal symbol proceeds along the following lines. Upon evaluation\footnote{For the moment, assume that $E$ is a complex vector bundle. The argumentation is valid also in the real case with obvious modifications.} of $L_f$ on strongly localised `plain waves' $a e^{\I \lambda \vartheta}$ with $a \in \secspace{E}, \vartheta \in C^\infty(M)$, one gets a polynomial in $\lambda$ (modulus a common factor $e^{\I \lambda \vartheta}$). The monomial of degree $r$ is just the principal symbol, or more precisely
\begin{equation}
	\lim_{\lambda \to \infty} \lambda^{-r} e^{- \I \lambda \vartheta} L_f (a e^{\I \lambda \vartheta}){\bigg|_{\raisebox{1.8pt}{{$\scriptstyle m$}}} } = \I^r a(m) \sigma_f(m, (\dif \vartheta)_m). 
\end{equation}
\begin{example}[Covariant derivative]
	Consider the operator of covariant derivation $\nabla: \secspaceEx{E} \to \secspaceEx{T^*M \otimes E}$ induced by a connection on the vector bundle $E$. The principal symbol of $\nabla$ is calculated based on
	\begin{equation}
		e^{- \I \lambda \vartheta} \nabla (a e^{\I \lambda \vartheta}) = \I \lambda \dif \vartheta \otimes a + \nabla a
	\end{equation}
	to $\sigma_\nabla(m, \xi) = \xi \otimes \bcdot$. Therefore, $\nabla$ is an overdetermined elliptic, first order differential operator. If $\nabla$ is compatible with the bundle metric, then an explicit formula for the adjoint $\nabla^*$ can be derived. Stokes theorem in conjunction with Cartan's formula imply $\int_M \difLie_X(\scalarprodr{\phi}{\varphi}_E \dif \vol_g(m)) = 0$. But on the other hand,
	\begin{equation}
		\int_M \difLie_X(\scalarprodr{\phi}{\varphi}_E \dif \vol_g(m)) = \int_M (\scalarprodr{\nabla_X \phi}{\varphi}_E + \scalarprodr{\phi}{\nabla_X \varphi}_E + \scalarprodr{\phi}{\varphi}_E \divergence_g(X)) \dif \vol_g(m).
	\end{equation}
	Hence, $(\nabla_X)^* = - \nabla_X - \divergence_g(X)$. From this expression one deduces
	\begin{equation} \label{eq::gaugeTheory:CovariantDerivativeFormalAdjoint}
		\nabla^*: \secspaceEx{T^*M \otimes E} \to \secspaceEx{E}, \qquad X^\flat \otimes \phi \mapsto - \nabla_X - \divergence_g(X) \phi = - \trace_g (\nabla (X^\flat \otimes \phi)),
	\end{equation}
	where $^\flat: \vectorf{M} \to \diffform{1}{M}$ denotes the natural isomorphism induced by $g$. The Bochner Laplacian is defined as the composition $\nabla^* \nabla$. Equation \eqref{eq::gaugeTheory:CovariantDerivativeFormalAdjoint} now implies the expression
	\begin{equation}
		\nabla^* \nabla: \secspace{E} \to \secspace{E}, \qquad \phi \mapsto - \trace_g(\nabla^2_{\cdot, \cdot} \phi),
	\end{equation}
	from which one infers the principal symbol $\sigma_{\nabla^* \nabla}(m, \xi) = - \norm{\xi}^2 \id_E$. Hence, $\nabla^* \nabla$ has the same symbol as the usual Laplacian\footnotemark and obviously is an elliptic operator.
	\footnotetext{Differential operators with the same principal symbol as the Laplacian are often called generalized Laplacians.}
\end{example}

The following properties of elliptic operators are well known in the Sobolev setting and carry over to smooth section spaces by projective limit techniques.
\begin{proposition}[{\parencite[pp. 155-159]{Hamilton1982}}] \label{prop::tameFrechet:ellipticDiffOpProperties}
	Let $L_f: \secspace{E} \to \secspace{F}$ be an elliptic differential operator. The following statements are true:
	\begin{enumerate}
		\item The kernel of $L_f$ is finite-dimensional.
		\item $L_f$ has closed range with finite codimension. 
		\item Denote by $V \subseteq \secspaceEx{ L(J^rE, F) }$ the open subset of coefficients on which $L_f$ is invertible. The solution operator
		\begin{equation}
			S: V \times \secspace{F} \to \secspace{E}, \quad S(f, \psi) = \phi \text{ such that } L_f \phi = \psi
		\end{equation}
		is a smooth tame family. \qedhere
	\end{enumerate}
\end{proposition}
The last point essentially encompasses the regularity of elliptic systems. 
\begin{remark}[Green's operator]
If the elliptic operator $L_f$ is not invertible, then nonetheless one can define an inverse up to some finite-dimensional spaces. Indeed, \cref{prop::tameFrechet:ellipticDiffOpProperties} together with the fact that all finite-dimensional subspaces are tamely complemented ensures the existence of closed subspaces $V_f \subseteq \secspace{E}, W_f \subseteq \secspace{F}$ such that 
\begin{equation}
	\secspace{E} = \ker L_f \oplus V_f \quad \text{and} \quad \secspace{F} = \img L_f \oplus W_f.
\end{equation}
Denote the projection along $V_f$ by $P_f: \secspace{E} \to \ker L_f$ and the natural injection by $\iota_f: W_f \to \secspace{F}$. Both maps are tame. The extended operator
\begin{equation}
	\tilde{L}_f: \secspace{E} \times W_f \to \secspace{F} \times \ker L_f, \qquad (\phi, w) \mapsto (L_f \phi + \iota_f (w), P_f(\phi))
\end{equation}
is bijective and thus gives rise to an inverse operator $\tilde{S}_f: \secspace{F} \times \ker L_f \to \secspace{E} \times W_f$. Now the Green's operator is the induced map $G_f: \secspace{F} \to \secspace{E}$. It is the subject of \parencite[Theorem 3.3.3]{Hamilton1982} to show that there exists an open neighbourhood $U$ of every $f_0 \in \secspaceEx{ L(J^rE, F) }$ such that the Green's operator is a smooth tame family $G: U \times \secspace{F} \to \secspace{E}$.
\end{remark}

Consider now the case of a covariant derivative and the associated Bochner Laplacian $\nabla^* \nabla$. The usual Hodge theorem for elliptic operators implies the decomposition $\secspace{E} = \img \nabla^* \nabla \oplus \ker \nabla^* \nabla$ on the level of vector spaces without topology. Equivalently,
\begin{equation}
	\secspaceEx{T^*M \otimes E} = \img \nabla \oplus \ker \nabla^*
\end{equation}
holds as basic linear algebra considerations show. The aim is now to proof tameness of this decomposition. First, the sum is topological by the following lemma in conjunction with the fact that every finite-dimensional subspace of a Fréchet space is closed.%
\begin{lemma}[{\parencite[p. 55]{Subramaniam1984}}]
	Let $T: X \to Y$ be a continuous, linear map between Fréchet spaces such that $\img T$ is algebraically complemented by a closed subspace. Then the image of $T$ is closed in $Y$ and the sum is topological. 
\end{lemma}
\begin{proof}
	Denote the complement of $\img T$ by $W$. The spaces $X / \ker T$ and $W$ are Fréchet as a quotient and a closed subspace, respectively. By construction the induced map
	\begin{equation}
	 	\tilde{T}: \sfrac{X}{\ker T} \times W \to Y, \qquad (\equivClass{x},w) \mapsto Tx + w
	 \end{equation} 
	 is a continuous linear bijection. Now, the open mapping theorem \parencite[Theorem~2.11]{Rudin1973} implies that $\img T$ is closed and $Y = \img T \oplus W$ topologically.
\end{proof}

In order to deduce tameness of the decomposition one has to show tameness of the map $\tilde{T}$ and its inverse. Directly consider the covariant derivative $T = \nabla$, which is a $r$-tame map. Denote $W = \ker \nabla^*$. Now, $\tilde{\nabla}$ is tame as the following calculation shows:
\begin{equation}\begin{split}
	\norm{\tilde{\nabla}(\equivClass{\phi}, w)}_k &= \norm{\nabla\phi + w}_k \leq \norm{\nabla\phi}_k + \norm{w}_k \\ &\leq C \norm{\phi}_{k+r} + \norm{w}_k \leq C (\norm{\phi}_{k+r} + \norm{w}_{k+r}).
\end{split}\end{equation}
The inverse is given by
\begin{equation}
	\secspaceEx{T^*M \otimes E} \to \sfrac{\secspace{E}}{\ker \nabla} \times \ker \nabla^*, \qquad \psi \mapsto (\equivClass{G_\nabla \psi}, \psi - \nabla \circ G_\nabla \psi),
\end{equation}
which is a tame map if the Green's operator $G_\nabla$ of $\nabla$ is tame. But $G_\nabla = G_{\nabla^* \nabla} \circ \nabla^*$ is the composition of the tame Green's operator $G_{\nabla^* \nabla}$ and the tame differential operator $\nabla^*$. Hence, $G_\nabla$ is itself tame. The previous discussion is summarized in:
\begin{proposition}
	Let $E$ be a vector bundle over some compact, Riemannian manifold. Endow $E$ with a fibre metric and let $\nabla: \secspaceEx{E} \to \secspaceEx{T^*M \otimes E}$ be a covariant derivative in $E$. Let $\nabla^*$ denote the formal adjoint relative to the $L^2$-scalar product. Then
	\begin{equation}
		\secspaceEx{T^*M \otimes E} = \img \nabla \oplus \ker \nabla^*
	\end{equation}
	is a tame direct sum.
\end{proposition}
Applied to the present case of gauge theory the previous proposition establishes the decomposition
\begin{equation} \label{eq::gaugeTheory:HodeDecomposition}
 	T_\lambda \secspace{Conn} = \img \nabla_\lambda \oplus \ker \nabla^*_\lambda.
\end{equation} 

Finally, the desired tame smooth family of inverses to 
\begin{equation} \label{eq::gaugeTheory:covariantDerivativeComposedWithProjection}
	\begin{tikzcd}
		\secspaceEx{\AdjBundle{P}} \arrow{r}{\nabla_\lambda} & \secspaceEx{T^*M \otimes \AdjBundle{P}} = \img \nabla_{\lambda_0} \oplus \ker \nabla^*_{\lambda_0} \arrow{r} & \img \nabla_{\lambda_0}
	\end{tikzcd}
\end{equation}
is constructed for $\lambda \in \Conn$ close to some fixed connection $\lambda_0$. By continuity of the action, this claim is equivalent to the corresponding condition in the slice theorem. For $\Xi \in \secspaceEx{\AdjBundle{P}}$, \cref{eq::gaugeTheory:HodeDecomposition} implies the existence of $\varphi \in \secspaceEx{\AdjBundle{P}}$ and $w \in \ker \nabla^*_{\lambda_0}$ such that $\nabla_\lambda \Xi = \nabla_{\lambda_0}\varphi + w$. Hence, $\nabla^*_{\lambda_0}\nabla_\lambda \Xi = \nabla^*_{\lambda_0} (\nabla_{\lambda_0}\varphi)$ holds and shows that a tame inverse to $\nabla^*_{\lambda_0}\nabla_\lambda$ is required. Since $\nabla^*_{\lambda_0}\nabla_\lambda$ has an invertible symbol at $\lambda = \lambda_0$ and the symbol map is continuous with respect to its coefficient $\lambda$, the operator $\nabla^*_{\lambda_0}\nabla_\lambda$ is elliptic in some open neighbourhood $U$ of $\lambda_0$ and possesses a smooth tame Green's operator $G_{\nabla^*_{\lambda_0}\nabla_\lambda}$. Now,
\begin{equation}
	U \times \img \nabla_{\lambda_0} \to \secspaceEx{\AdjBundle{P}}, \qquad (\lambda, \nabla_{\lambda_0}\varphi) \mapsto \Xi = (G_{\nabla^*_{\lambda_0}\nabla_\lambda} \nabla^*_{\lambda_0})(\nabla_{\lambda_0}\varphi)
\end{equation}
is a tame smooth family inverse to \eqref{eq::gaugeTheory:covariantDerivativeComposedWithProjection}.

\section{Invariant Riemannian metric} \label{sec::gaugeTheory:InvariantRiemannianMetric} 
The construction of a $\Gau$-invariant graded Riemannian metric on $\Conn$ proceeds in a way similar to the definition of the seminorms on section spaces in \autoref{sec::sectionSpace}. Fix a Riemannian metric $g$ on $M$ and denote the associated Levi-Civita connection by $\nabla^M$. A connection $\lambda$ gives rise to a covariant derivative $\nabla_\lambda: \secspaceEx{\AdjBundle{P}} \to \secspaceEx{T^*M \otimes \AdjBundle{P}}$ in the adjoint bundle. These two connections combine via the Leibniz identity and finally define by recursion 
\begin{equation}
	\tilde{\nabla}_\lambda^k: \secspaceEx{T^*M \otimes \AdjBundle{P}} \to \secspaceEx{\otimes_{k+1} T^*M \otimes \AdjBundle{P}}.
\end{equation}
Note that a gauge transformation $\lambda \xrightarrow{\Phi} \tilde{\lambda}$ is reflected in a $\AdjAction(\Phi)$-equivariance of the induced covariant derivatives $\nabla_\lambda$ and $\nabla_{\tilde{\lambda}}$. Hence, the assignment $\Conn \ni \lambda \to g^k_\lambda$ with
\begin{equation}
	g^k_\lambda: T_\lambda \Conn \times T_\lambda \Conn \to \R, \quad (\sigma, \varsigma) \mapsto \sup_{i \leq k} \int_M \scalarprod{\tilde{\nabla}_\lambda^i \sigma(m)}{\tilde{\nabla}_\lambda^i \varsigma(m)} \dif \vol_g(m)
\end{equation} 
defines a $\Gau$-invariant Riemannian metric on $\Conn$ if the bundle $\AdjBundle{P}$ carries an $\AdjAction$-invariant metric. Suppose that the structure group is of the form $G = G_K \times \R^l$ for some compact Lie group $G_K$ and $l \in \N$. Then its Lie algebra $\liea{g}$ admits a $\AdjAction$-invariant scalar product and thus yields the desired $\AdjAction$-invariant Riemannian metric. The resulting $\Gau$-invariant graded Riemannian metric is equivalent to the Sobolev metrics \parencite{EbinMarsden1970}and hence locally equivalent in the notation of \cref{defn::Riemann:LocallyEquivalentMetric}. %

Observe that for $k=0$ the exponential map of $g^k$ is just the addition
\begin{equation}
	\exp^0_\lambda: T_\lambda \Conn \to \Conn, \qquad (\lambda, \sigma) \mapsto \lambda + \sigma.
\end{equation}
Restricted to the normal bundle, it is a local diffeomorphism around every point of the orbit as the application of the Nash-Moser theorem and the above discussed properties of elliptic operators show (see \parencite[pp. 60-64]{Subramaniam1984}).

\section{Yang-Mills Lagrangian and its gauge symmetry}
In this section, the Yang-Mills Lagrangian will be discussed. The momentum map corresponding to the symmetry under gauge transformation is determined. For the time being, the base manifold $M$ is not required to be compact but, in order to simplify the calculation, attention is restricted to matrix Lie groups $G$ (that is, subgroups of some general linear group).

The Yang-Mills equation derives from a variational principle applied to the local representation of a connection. Let $U \subseteq M$ be an open subset on which the principal bundle $P$ (and thus also $CP$) trivializes. 
Via a local section, the connection $\lambda$ can be identified with an element $A \in \diffform{1}{U}{\liea{g}}$.
The assignment $\lambda \mapsto A$ is smooth since it decomposes into
\begin{equation}
	\begin{tikzcd}[column sep=large]
			\secspaceEx{CP} \arrow{r}{\restriction_U} & \secspaceEx{CP_{\restriction_U}} \arrow{r}{\comp_{\text{chart}}} & \diffform{1}{U}{\liea{g}},
	\end{tikzcd}
\end{equation}  
where the restriction and composition map are smooth, see \parencite[Lemma 2.2.6.]{Wockel2006} and \parencite[Theorem 2.3.5]{Hamilton1982}, respectively. As usual, the local representation of the curvature is defined by $F_A \defeq \dif A + \sfrac{1}{2} \, \wedgeLie{A}{A}$, where $\wedgeLie{\cdot}{\cdot}$ denotes the wedge product of differential forms relative to the commutator on $\liea{g}$. In the present case of a matrix Lie group $G$, the expression for the field strength simplifies to $F_A = \dif A + A \wedge A$. 

Now the Lagrangian system is given by
\begin{align}
	L 		&= - \frac{1}{2} \trace(F_A \wedge \star F_A), \\
	\theta	&= - \trace (\diF A \wedge \star F_A).
\end{align}

Using the Hodge-star identity \eqref{eq::differentialgeometry:hodgeStarProductManifoldWedgeCommutative}, the Euler-Lagrange form is computed as
\begin{equation}\begin{split}
	E_{\mathcal{L}} &= \diF L + \dif \theta\\
					&= - \frac{1}{2} \trace( \diF F_A \wedge \star F_A + F_A \wedge \star \diF F_A) + \dif \theta\\
					&= - \trace( \diF F_A \wedge \star F_A) + \dif \theta.
\intertext{Since the variation of the curvature is given as $\diF F_A = \diF \dif A + \diF A \wedge A - A \wedge \diF A$, this expands to}
					&= - \trace( \diF \dif A \wedge \star F_A + \diF A \wedge A \wedge \star F_A - A \wedge \diF A \wedge \star F_A  + \dif \diF A \wedge \star F_A + \diF A \wedge \dif \star F_A) \\
					&= - \trace( \diF A \wedge (A \wedge \star F_A - (-1)^n \star F_A \wedge A + \dif \star F_A)),
\end{split}\end{equation} 
where the last equality follows from the $\AdjAction$-invariance of the trace. %
The expression between the brackets is identified as the covariant derivative of $\star F_A$ relative to $A$, where
\begin{equation}
	\Dif_A \beta = \dif \beta + \wedgeLie{A}{\beta} = \dif \beta + A \wedge \beta + (-1)^{\# \beta + 1} \beta \wedge A.
\end{equation}
Hence the Euler-Lagrange equation yields the known Yang-Mills equation $\Dif_A \star F_A = 0$. The 
subsymplectic form is given by $\omega = \diF \theta = - \trace (\diF A \wedge \star \diF F_A)$.
\begin{remark}
	On a purely formal level, the previous calculation also applies in the global setting where the local form $A \in \diffform{1}{U}{\liea{g}}$ is replaced by the connection form $\alpha \in \diffform{1}{P}{\liea{g}}$. However, the assignment $\lambda \mapsto \alpha$ of a connection $\lambda$ to the associated connection form $\alpha$ cannot be expected to be smooth in general. In particular, the above defined differential forms would no longer be smooth objects on $\Conn$.  
\end{remark}

The momentum map associated with the gauge symmetry is considered now. Due to \eqref{eq::gaugeTheory:ActionGauHatOnConnForm}, it is not hard to see that $L$ as well as $\theta$ are invariant under gauge transformations. Thus, \autoref{prop::classicalFieldTheory:momentumMapForInvariantLagrangian} yields the momentum map
\begin{equation}
	J_\Xi = \Xi^* \contr \mathcal{L}.
\end{equation}
The Killing vector field $\Xi^*$ (on $\secspace{Conn}$) associated to $\Xi \in \secspace{gau}$ is easily determined with the aid of \eqref{eq::gaugeTheory:gaugeTransformationTangent},
\begin{equation}
	\Xi^*_\lambda = (\secmap{transf}_\lambda)'_e (\Xi) = - \nabla_\lambda \Xi.
\end{equation}
It can be viewed as a $M$-independent vector field $(- \nabla_\lambda \Xi, 0)$ on the product $\secspace{Conn} \times M$, where it is denoted by $\Xi^*$ as well. The identity $(\Xi^* \contr \Dif(\ev))_{\lambda, m} = - (\nabla_\lambda \Xi)(m) = (\Xi^* \contr \diF (\ev))_{\lambda, m}$
holds. Hence the contraction $\Xi^* \contr L$ yields zero and thus finally
\begin{equation}
	J_\Xi = \Xi^* \contr \mathcal{L} = \Xi^* \contr \theta = \trace (\nabla_A \Xi \wedge \star F_A).
\end{equation}

The results of this section are summarized in the following theorem.
\begin{theorem}
	The Yang-Mills Lagrangian 
	\begin{equation}
		\mathcal{L} = - \frac{1}{2} \trace(F_A \wedge \star F_A) - \trace (\diF A \wedge \star F_A) \in \diffform{n}{\secspace{Conn} \times M}
	\end{equation}
	is invariant under the action $\secmap{transf}$ of the group of gauge transformation and the momentum map associated to $\Xi \in \secspace{gau}$ is given by
	\begin{equation}
		J_\Xi = \Xi^* \contr \mathcal{L} = \Xi^* \contr \theta = \trace (\nabla_A \Xi \wedge \star F_A).
	\end{equation}
\end{theorem}
\chapter{Summary and outlook}
\sectionmark{Summary and outlook}

\begin{description}[wide]
	\item[Slice Theorem] Based on the differential calculus of locally convex spaces, a general slice theorem was established. Thereby, the finite-dimensional result of Palais is extended to a wide class of infinite-dimensional Lie group actions. Instead of employing Sobolev techniques, the present investigation completely proceeded in the subcategory of tame Fréchet spaces. Functional analytic obstacles were overcome by employing the Nash-Moser theorem. Finally, the power of the presented theorem was illustrated by its application to gauge theories: The action of the gauge transformation group admits smooth slices at every point.
	
	\item[Symplectic formalism for classical field theory] A covariant and symplectic formulation of classical field theory was proposed and extensively discussed. At the root of this novel framework is the incorporation of field degrees of freedom $\secspace{F}$ and spacetime $M$ into the product manifold $\secspace{F} \times M$. Using the induced bigrading of differential forms, the details of this symplectic theory were worked out for this setting. It was possible to carry over all basic notions, such as the Hamiltonian vector field and the momentum maps, which finally culminated in a covariant, symplectic Noether theorem. The examples of the Klein-Gordon field and general Yang-Mills theory illustrated that the presented approach conveniently handles the occurring symmetries.

	\newpage
	\item[Outlook] The main motivation for this work is the generalization of singular symplectic reduction to Fréchet manifolds. Although the above slice theorem and the studied symplectic structure are an important step in this direction, further work is needed to accomplish the final goal. In particular, these two concepts have to be merged into a symplectic slice theorem. However, in order to accomplish this in the finite-dimensional setting, the present approach heavily relies on automatic topological complementation of subspaces and non-degeneracy of the symplectic form. Hence, a fundamental rethinking of the usual techniques is necessary. Nonetheless, it will be worth the effort, since singular symplectic reduction would give rewarding insights into the stratified structure of the orbit space.

	Furthermore, from a physicist's perspective, the symplectic structure of field theories opens doors to new and enriching possibilities, which will result from transferring the finite-dimensional techniques to this setting. In order to gain new insights into quantum field theory, one procedure appears especially promising: the study of geometric quantization of field theories. %
\end{description}

\vfill
\setlength\epigraphwidth{0.75\textwidth}
\epigraph{%
[...] symmetry is not an easy thing to achieve. [...] That is why only the fittest and healthiest individual plants have enough energy to spare to create a shape with balance. The superiority of the symmetrical flower is reflected in a greater production of nectar, and that nectar has a higher sugar content. \emph{Symmetry tastes sweet.}
}{\textcite[p. 12]{Sautoy2009}}
\vspace{-1cm}
%

{
\appendix
\chapter{Proper maps} \label{cha::topology:properMapCompactlyGeneratedSpace}
\sectionmark{Proper maps}
The purpose of this appendix is to show that properties usually stated for locally compact spaces extend to the more general setting of compactly generated spaces. Although the following statements have elementary character, no explicit reference for them could be found.

\begin{defn}
 	A topological vector space $X$ is called \emphDef{compactly generated} if every subset $A$ is closed in $X$ if and only $A \cap K$ is closed in $K$ for all compact subsets $K \subseteq X$. Compactly generated spaces fulfilling the Hausdorff separation axiom are often shortly named $\emphDef{CGH-space}$ (for compactly generated Hausdorff). 
 \end{defn} 
In the definition, `closed' can be replaced by `open'. Thus, in compactly generated spaces the compact sets can be used to test openness of subsets and therefore the topology is coherent with its collection of compact subsets. In particular, first countable or locally compact spaces are compactly generated.

Compactly generated spaces provide the general stage for the following equivalent characterisations of proper maps.
\begin{proposition} \label{prop::topo:properMapEquivCharacterization}
	Let $X$ be a topological space, $Y$ a CGH-space and $\Phi: X \to Y$ a continuous map. Then the following are equivalent:
	\begin{thmenumerate}
		\item $\Phi$ is a proper map, that is, the inverse images of every compact subset is compact again.
		\item $\Phi$ is a closed map and the inverse image of every single point set is compact. \label{prop::topo:properMapEquivCharacterization_closedMapAndInverseImageOfPointIsCompact}
	\end{thmenumerate}
	If $X$ and $Y$ are additionally metric spaces, then the following is also equivalent:
	\begin{enumerate}[resume]
		\item Every sequence $(x_i)$ such that $\Phi(x_i)$ converge has a convergent subsequence. \qedhere
	\end{enumerate}
\end{proposition}
\begin{proof}
	\begin{thmenumerate}*[wide]
		\item[(i) $\Rightarrow$ (ii)] Every singleton set is compact and hence by properness of $\Phi$ its inverse image is compact. Therefore it is left to show that the image $\Phi(A)$ of every closed subset $A \subseteq X$ is closed in $Y$. As $Y$ is compactly generated, the collection of compact subsets can be used to test closedness. For this purpose, let $K \subseteq Y$ be an arbitrary compact subset. Its inverse image $\Phi^{-1}(K)$ is compact by properness of $\Phi$. Since $\Phi^{-1}(K) \cap A$ is closed and contained in the compact subset $\Phi^{-1}(K)$, it is compact itself. Finally, its image $\Phi(\Phi^{-1}(K) \cap A) = K \cap \Phi(A)$ under the continuous map $\Phi$ is compact, implying the desired closedness of $\Phi(A)$. %
		
		\item[(ii) $\Rightarrow$ (i)] Let $K \subseteq Y$ be a compact subset and fix an open cover $U_\alpha$ of $\Phi^{-1}(K)$. For every $y \in K$ the inverse image $\Phi^{-1}(y)$ is compact by assumption and thus can be covered by a finite subcover $U_{i(y)}$, where the notation $i(y)$ displays the dependence of the index set on the chosen point $y$. Now, $X \backslash \bigcup U_{i(y)}$ as well as its image under $\Phi$ is closed and hence $Y \backslash \Phi(X \backslash \bigcup U_{i(y)}) \equiv V_y$ is an open subset of $Y$ containing $y$. Since $K$ is compact, there exist finitely many points $y_j$ such that the associated $V_{y_j}$ cover $K$. Then, there exists also a finite subcover of $\Phi^{-1}(K)$ indexed by $i(y_j)$. %
		
		\item[(i) $\Rightarrow$ (iii)] Let $(x_i)$ be a sequence in $X$ such that $\Phi(x_i) \rightarrow y$. The existence of the convergent subsequence of $(x_i)$ follows from a proof by contradiction. For this purpose, assume that $(x_i)$ has no cluster point and thus an open neighbourhood $U_p \subseteq X$ of every point $p \in \Phi^{-1}(y)$ can be chosen in such a way that it does not contain any (or only finitely many) $x_i$. Since $\Phi^{-1}(y)$ is compact, there exists a finite covering of open subsets of the form $U_{p_j}$. Now, following the same argument as above one can show that $V_y \equiv Y \backslash \Phi(X \backslash \bigcup U_{p_j)})$ is an open neighbourhood of $y$. But then, the limit $\Phi(x_i) \rightarrow y$ implies that all $\Phi(x_i)$ lie in $V_y$ for sufficiently large $i$, which is in contradiction with the construction of $U_{p_j}$.    

		\item[(iii) $\Rightarrow$ (i)] Let $K \subseteq Y$ be a compact subset and $(x_i)$ a sequence in $\Phi^{-1}(K)$. Since sequential compactness is equivalent to compactness for metric spaces, it is enough to show that $(x_i)$ has a convergent subsequence. As $\Phi(x_i)$ is by construction a sequence in $K$, it thus has a convergent subsequence. Now ,the claim follows from the starting hypothesis. \qedhere %
	\end{thmenumerate}
\end{proof}
\chapter{Hodge star operator on product manifolds} \label{sec::differentialgeometry:hodgeStarProductManifold}
\sectionmark{Hodge star operator on product manifolds}

This appendix shortly clarifies how the Hodge star operator can be defined for differential forms on product manifolds. 

Let $\secspace{F}$ be a possible infinite-dimensional manifold and $(M, g)$ be a finite-dimensional Riemannian manifold with associated Hodge dual operator $\star$. For differential forms of product type $\alpha = \alpha_{\secspace{F}} \wedge \alpha_M$ and $\beta = \beta_{\secspace{F}} \wedge \beta_M$ with $\alpha_{\secspace{F}}, \beta_{\secspace{F}} \in \diffformBi{\bcdot}{0}{\secspace{F} \times M}$ and $\alpha_M, \beta_M \in \diffformBi{0}{\bcdot}{\secspace{F} \times M}$ the Hodge star is simply defined by moving over the $\secspace{F}$-component, that is,
\begin{equation}
	\alpha \wedge \star \beta \defeq \alpha_{\secspace{F}} \wedge \alpha_M \wedge \beta_{\secspace{F}} \wedge \star \beta_M,
\end{equation}
where a $(0, q)$-form on $\secspace{F} \times M$ is identified with a $q$-form on $M$. 

The commutative law $\alpha \wedge \star \beta = \beta \wedge \star \alpha$ is modified due to the additional minus sign occurring when the $\secspace{F}$-component changes places with the $M$-component.  
\begin{proposition} \label{prop::differentialgeometry:hodgeStarProductManifoldWedgeCommutative}
	Let $\alpha = \alpha_{\secspace{F}} \wedge \alpha_M$ and $\beta = \beta_{\secspace{F}} \wedge \beta_M$ be differential forms of product type as above. If $\#$ denotes the total degree of the differential form and the subindexed version refers to the corresponding partial degree, then the following identity holds
	\begin{equation} \label{eq::differentialgeometry:hodgeStarProductManifoldWedgeCommutative}
		\alpha \wedge \star \beta = (-1)^{\#\alpha \, \cdot \, \#\beta \,-\, \#_M \alpha \, \cdot \, \#_M \beta} \beta \wedge \star \alpha. 
	\end{equation} 
\end{proposition}
\begin{proof}
	\vspace*{-\baselineskip}
	\begin{align}
			\alpha \wedge \star \beta 	&= \alpha_{\secspace{F}} \wedge \alpha_M \wedge \beta_{\secspace{F}} \wedge \star \beta_M && \\
										&= \alpha_{\secspace{F}} \wedge \beta_{\secspace{F}} \wedge \alpha_M \wedge \star \beta_M && \cdot (-1)^{ \#_M \alpha \, \cdot \, \#_{\secspace{F}} \beta } \\
										&= \alpha_{\secspace{F}} \wedge \beta_{\secspace{F}} \wedge \beta_M \wedge \star \alpha_M && \cdot (-1)^{\#_M \alpha \, \cdot \, \#_{\secspace{F}} \beta} \\
										&= \beta_{\secspace{F}} \wedge \beta_M \wedge \alpha_{\secspace{F}} \wedge \star \alpha_M && \cdot (-1)^{\#_M \alpha \, \cdot \, \#_{\secspace{F}} \beta \,+\, \#_{\secspace{F}} \alpha \, \cdot \, \#_{\secspace{F}} \beta \,+\, \#_{\secspace{F}} \alpha \, \cdot \, \#_M \beta}\\ 
										&= \beta \wedge \star \alpha && \cdot (-1)^{\# \alpha \, \cdot \, \# \beta \, - \, \#_M \alpha \, \cdot \, \#_M \beta} 
	\end{align}
\end{proof}
}
\begin{spacing}{1} 		
\clearpage
\printbibliography
\end{spacing}

\clearpage
\end{document}